\documentclass[a4paper,UKenglish,cleveref,autoref,numberwithinsect]{lipics-v2019}

\usepackage{stmaryrd}
\usepackage{proof}
\usepackage[dvipsnames]{xcolor}

\pdfoutput=1
\nolinenumbers

\theoremstyle{definition}
\newtheorem{defn}[theorem]{Definition}

\newcommand{\Fomega}{\mathtt{F}_\omega}

\newcommand{\Vars}{\mathcal{V}}
\newcommand{\Rules}{\mathcal{R}}
\newcommand{\Iterms}{\mathcal{I}}
\newcommand{\ITypes}{\mathcal{Y}}

\newcommand{\arrkind}{\Rightarrow}
\newcommand{\arrtype}{\rightarrow}
\newcommand{\quant}[2]{\forall #1.#2}

\newcommand{\abstraction}[2]{\backslash #1.#2}
\newcommand{\app}[2]{#1 \cdot #2}
\newcommand{\tapp}[2]{#1 * #2}
\newcommand{\subst}[2]{#1:=#2}

\newcommand{\abs}[2]{\lambda #1.#2}
\newcommand{\tabs}[2]{\Lambda #1.#2}
\newcommand{\pair}[2]{\langle #1,#2 \rangle}
\newcommand{\expair}[2]{[#1,#2]}

\newcommand{\arrW}{\leadsto}
\newcommand{\arr}[1]{\longrightarrow_{#1}}
\newcommand{\red}{\longrightarrow}
\newcommand{\arrrbeta}{\arrW_\beta^*}

\newcommand{\nat}{\mathtt{nat}}
\newcommand{\flatten}{\mathtt{flatten}}
\newcommand{\lift}{\mathtt{lift}}

\newcommand{\typeinterpret}[1]{\llbracket #1 \rrbracket}
\newcommand{\interpret}[1]{\llbracket #1 \rrbracket}

\newcommand{\refsec}[1]{Section~\ref{sec:#1}}

\newcommand{\FTV}{\mathrm{FTV}}
\newcommand{\FV}{\mathrm{FV}}
\newcommand{\Tc}{\mathcal{T}}
\newcommand{\Vc}{\mathcal{V}}

\newcommand{\cl}{\mathcal{C}}

\newcommand{\nf}{\mathrm{nf}}

\newcommand{\da}{\mathord{\downarrow}}
\newcommand{\SN}{\mathrm{SN}}
\newcommand{\Cb}{\mathbb{C}}
\newcommand{\Nbb}{\mathbb{N}}
\newcommand{\val}[3]{\ensuremath{\llbracket#1\rrbracket_{#2}^{#3}}}
\newcommand{\gteq}[3]{\ensuremath{\ge_{#1}^{#2,#3}}}

\newcommand{\Typemap}{\mathcal{T\!M}}
\newcommand{\Termmap}{\mathcal{J}}
\newcommand{\succinterpret}{\succ^{\Termmap}}
\newcommand{\succeqinterpret}{\succeq^{\Termmap}}

\newcommand{\List}{\mathtt{List}}
\newcommand{\Pair}{\mathtt{Pair}}
\newcommand{\nil}{\mathtt{nil}}
\newcommand{\cons}{\mathtt{cons}}

\newcommand{\xlet}[4]{\mathtt{let}_{#1}\,#2\,\mathtt{be}\,[#3]\,\mathtt{in}\,#4}
\newcommand{\proj}{\mathtt{pr}}

\newcommand{\onlypaper}[1]{}
\newcommand{\onlyarxiv}[1]{#1}

\title{Polymorphic Higher-order Termination}

\author{{\L}ukasz Czajka}{Faculty of Informatics, TU Dortmund, Germany \and \url{http://www.mimuw.edu.pl/~lukaszcz/} }{lukaszcz@mimuw.edu.pl}{https://orcid.org/0000-0001-8083-4280}{}

\author{Cynthia Kop}{Institute of Computer Science, Radboud University Nijmegen, Netherlands \and \url{https://www.cs.ru.nl/~cynthiakop/}}{c.kop@cs.ru.nl}{https://orcid.org/0000-0002-6337-2544}{}

\authorrunning{\L. Czajka and C. Kop}

\Copyright{{\L}ukasz Czajka and Cynthia Kop}



\EventEditors{}
\EventNoEds{0}
\EventLongTitle{}
\EventShortTitle{}
\EventAcronym{}
\EventYear{}
\EventDate{}
\EventLocation{}
\EventLogo{}
\SeriesVolume{}
\ArticleNo{}
\begin{document}

\maketitle

\begin{abstract}
  We generalise the termination method of higher-order polynomial
  interpretations to a setting with impredicative
  polymorphism. Instead of using weakly monotonic functionals, we
  interpret terms in a suitable extension of System~$\Fomega$. This
  enables a direct interpretation of rewrite rules which make
  essential use of impredicative polymorphism.  In addition, our
  generalisation eases the applicability of the method in the
  non-polymorphic setting by allowing for the encoding of inductive data
  types. As an illustration of the potential of our method, we prove
  termination of a substantial fragment of full intuitionistic
  second-order propositional logic with permutative conversions.
\end{abstract}

\section{Introduction}

Termination of higher-order term rewriting
systems~\cite[Chapter~11]{Terese2003} has been an active area of
research for several decades.
One powerful method, introduced by v.d. Pol \cite{Pol1993,pol:96},
interprets terms into \emph{weakly monotonic algebras}.  In later work
\cite{FuhsKop2012,Kop2012}, these algebra interpretations are specialised
into \emph{higher-order polynomial interpretations}, a generalisation of
the popular -- and highly automatable -- technique of polynomial
interpretations for first-order term rewriting.

The methods of weakly monotonic algebras and polynomial
interpretation are both limited to \emph{monomorphic} systems.
In this paper, we will further generalise polynomial
interpretations to a higher-order formalism with full impredicative
polymorphism.
%
%\medskip\noindent\textbf{Term rewriting with full impredicative
%polymorphism.}
%
This goes beyond shallow (rank-1, weak) polymorphism,
where type quantifiers are effectively allowed only at the top of a
type: it would be relatively easy to extend the methods to a system
with shallow polymorphism since shallowly polymorphic rules can be seen as defining an
infinite set of monomorphic rules.
While shallow polymorphism often suffices
in functional programming practice, there do exist interesting
examples of rewrite systems which require higher-rank impredicative
polymorphism.

For instance, in recent extensions of Haskell one may define a type of
heterogeneous lists.
\[
\begin{array}{ll}
  \List : * &
  \mathtt{foldl}_\sigma(f,a,\nil) \red a \\
  \mathtt{nil} : \List &
  \mathtt{foldl}_\sigma(f,a,\cons_\tau(x,l)) \red \mathtt{foldl}_\sigma(f,f \tau a x,l) \\
  \mathtt{cons} : \forall \alpha . \alpha \arrtype \List \arrtype \List \quad\quad \\
  \multicolumn{2}{l}{\mathtt{foldl} : \forall \beta . (\forall \alpha . \beta \arrtype \alpha \arrtype \beta) \arrtype \beta \arrtype \List \arrtype \beta}
\end{array}
\]
The above states that $\List$ is a type ($*$), gives the types of its
two constructors $\nil$ and $\cons$, and defines the corresponding
fold-left function~$\mathtt{foldl}$. Each element of a heterogeneous
list may have a different type. In practice, one would
  constrain the type variable~$\alpha$ %above
  with a type class to
  guarantee the existence of some operations on list elements.  The
function argument of~$\mathtt{foldl}$ receives the element together
with its type. The $\forall$-quantifier binds type variables: a term
of type $\forall \alpha . \tau$ takes a type~$\rho$ as an argument and
the result is a term of type~$\tau[\subst{\alpha}{\rho}]$.

Impredicativity of polymorphism means that the type itself may be
substituted for its own type variable, e.g., if $\mathtt{f} : \forall
\alpha . \tau$ then $f (\forall \alpha . \tau) :
\tau[\subst{\alpha}{\forall\alpha.\tau}]$. Negative occurrences of
impredicative type quantifiers prevent a translation
into an infinite set of simply typed rules by instantiating the type
variables. The above example is not directly reducible to shallow
polymorphism as used in the~ML programming language.

\medskip\noindent\textbf{Related work.} The term rewriting literature
has various examples of higher-order term rewriting systems with some
forms of polymorphism.  To start, there are several studies that
consider shallow polymorphic rewriting
(e.g., \cite{ham:18,jou:rub:07,wah:04}), where (as in ML-like
languages) systems like $\mathtt{foldl}$ above cannot be handled.
Other works consider extensions of the $\lambda\Pi$-calculus
\cite{cou:dow:07,dow:17} or the calculus of constructions
\cite{bla:05,wal:03} with rewriting rules; only the latter
includes full impredicative polymorphism.  The termination techniques
presented for these systems are mostly syntactic (e.g., a recursive
path ordering \cite{jou:rub:07,wal:03}, or general schema
\cite{bla:05}), as opposed to our more semantic method based on
interpretations.
An exception is \cite{dow:17}, which defines
interpretations into $\Pi$-algebras; this technique bears some
similarity to ours, although the methodologies are
quite different.  A categorical definition for a general polymorphic
rewriting framework is presented in \cite{fio:ham:13}, but no
termination methods are considered for it.

\medskip\noindent\textbf{Our approach.}
The technique we develop in this paper operates on \emph{Polymorphic
Functional Systems (PFSs)}, a form of higher-order term rewriting systems
with full impredicative polymorphism (Section \ref{sec_systems}), that
various systems of interest can be encoded into (including the example
of heterogeneous fold above). Then, our methodology follows a standard
procedure:
\begin{itemize}
\item we define a well-ordered set $(\Iterms,\succ,\succeq)$
  (Section \ref{sec:World});
\item we provide a general methodology to map each PFS term $s$ to a
  natural number $\interpret{s}$, parameterised by a core interpretation
  for each function symbol (Section \ref{sec_reduction_pairs});
\item we present a number of lemmas to make it easy to prove that
  $s \succ t$ or $s \succeq t$ whenever $s$ reduces to $t$
  (Section \ref{sec_rule_removal}).
\end{itemize}
Due to the additional complications of full polymorphism, we have
elected to only generalise higher-order polynomial interpretations,
and not v.d. Pol's weakly monotonic algebras.  That is, terms of base
type are always interpreted to natural numbers and all functions are
interpreted to combinations of addition and multiplication.

We will use the system of heterogeneous fold above as a running
example to demonstrate our method.  However, termination of this
system can be shown in other ways (e.g., an enco\-ding in
System~$\mathtt{F}$). Hence, we will also study a more complex example
in Section~\ref{sec:examples}: termination of a substantial fragment
of~IPC2, i.e., full intuitionistic second-order propositional logic
with permutative conversions. Permutative
conversions~\cite[Chapter~6]{TroelstraSchwichtenberg1996} are used in
proof theory to obtain ``good'' normal forms of natural deduction
proofs, which satisfy e.g.~the subformula property. Termination proofs
for systems with permutative conversions are notoriously tedious and
difficult, with some incorrect claims in the literature and no uniform
methodology. It is our goal to make such termination proofs
substantially easier in the future.

\onlypaper{Complete proofs for the results in this paper are available
in an online appendix.~\cite{versionwithappendix}.}%
\onlyarxiv{This is a pre-publication copy of a paper at FSCD 2019.  In
particular, it contains an appendix with complete proofs for the results
in this paper.}

\section{Preliminaries}\label{sec_preliminaries}

In this section we recall the definition of System~$\Fomega$ (see
e.g.~\cite[Section~11.7]{SorensenUrzyczyn2006}), which will form a
basis both of our interpretations and of a general syntactic framework
for the investigated systems. In comparison to System~$\mathrm{F}$,
System~$\Fomega$ includes type constructors which results in a more
uniform treatment. We assume familiarity with core notions of lambda
calculi such as substitution and $\alpha$-conversion.

\begin{defn}\label{def_types}
  \emph{Kinds} are defined inductively: $*$ is a kind, and if
  $\kappa_1,\kappa_2$ are kinds then so is $\kappa_1 \arrkind
  \kappa_2$. We assume an infinite set~$\Vc_\kappa$ of \emph{type
    constructor variables} of each kind~$\kappa$. Variables of
  kind~$*$ are \emph{type variables}. We assume a fixed
  set~$\Sigma^T_\kappa$ of \emph{type constructor symbols} paired with a
  kind~$\kappa$, denoted $c : \kappa$.
  We define the set~$\Tc_\kappa$ of \emph{type constructors} of
  kind~$\kappa$ by the following grammar.
  Type constructors of kind~$*$ are \emph{types}.
  \[
  \begin{array}{rcl}
    \Tc_{*} &::=& \Vc_{*}
    \mid \Sigma^T_{*} \mid
    \Tc_{\kappa\arrkind *}\Tc_{\kappa} \mid \forall\Vc_\kappa\Tc_* \mid \Tc_*\arrtype\Tc_* \\
    \Tc_{\kappa_1\arrkind\kappa_2} &::=& \Vc_{\kappa_1\arrkind\kappa_2}
    \mid \Sigma^T_{\kappa_1\arrkind\kappa_2} \mid
    \Tc_{\kappa\arrkind(\kappa_1\arrkind\kappa_2)}\Tc_{\kappa} \mid \lambda \Vc_{\kappa_1} \Tc_{\kappa_2}
  \end{array}
  \]

  We use the standard notations $\forall \alpha . \tau$ and $\lambda
  \alpha . \tau$. When $\alpha$ is of kind $\kappa$ then we use the
  notation $\forall \alpha : \kappa . \tau$. If not indicated
  otherwise, we assume~$\alpha$ to be a type variable. We treat type
  constructors up to $\alpha$-conversion.

  \begin{example}
  If $\Sigma^T_{*} = \{ \List \}$ and $\Sigma^T_{* \arrkind * \arrkind
  *} = \{ \Pair \}$, types are for instance $\List$ and
  $\forall \alpha.\Pair\,\alpha\,\List$.  The expression
  $\Pair\,\List$ is a type constructor, but not a type.  If
  $\Sigma^T_{(* \arrkind *) \arrkind *} = \{ \exists \}$ and
  $\sigma \in \Tc_{* \arrkind *}$, then both
  $\exists(\sigma)$ and $\exists (\lambda \alpha.\sigma\alpha)$ are
  types.
  \end{example}

  The compatible closure of the rule $(\lambda\alpha.\varphi)\psi \to
  \varphi[\alpha := \psi]$ defines $\beta$-reduction on type
  constructors. As type constructors are (essentially) simply-typed
  lambda-terms, their $\beta$-reduction terminates
  and is confluent; hence every type constructor~$\tau$ has a unique
  $\beta$-normal form~$\nf_\beta(\tau)$. A \emph{type atom} is a type
  in $\beta$-normal form which is neither an arrow $\tau_1\arrtype\tau_2$
  nor a quantification $\forall\alpha.\tau$.

  We define $\FTV(\varphi)$ -- the set of free type constructor
  variables of the type constructor~$\varphi$ -- in an obvious way by
  induction on~$\varphi$. A type constructor~$\varphi$ is
  \emph{closed} if $\FTV(\varphi) = \emptyset$.

  We assume a fixed type symbol~$\chi_* \in
  \Sigma^T_*$. For $\kappa=\kappa_1\arrkind\kappa_2$ we define
  $\chi_\kappa = \lambda \alpha:\kappa_1 . \chi_{\kappa_2}$.
\end{defn}

\begin{defn}\label{def_preterms}
  We assume given an infinite set $\Vars$ of variables, each paired
  with a type, denoted $x : \tau$. We assume given a fixed set
  $\Sigma$ of \emph{function symbols}, each paired with a closed type,
  denoted $\mathtt{f} : \tau$. Every variable~$x$ and every function
  symbol $\mathtt{f}$ occurs only with one type declaration.

  The set of \emph{preterms} consists of all expressions~$s$ such that
  $s : \sigma$ can be inferred for some type $\sigma$ by the following
  clauses:
  \begin{itemize}
  \item $x : \sigma$ for $(x : \sigma) \in \Vars$.
  \item $\mathtt{f} : \sigma$ for all
    $(\mathtt{f} : \sigma) \in \Sigma$.
  \item $\abs{x:\sigma}{s} : \sigma \arrtype \tau$ if
    $(x : \sigma) \in \Vars$ and $s : \tau$.
  \item $(\tabs{\alpha:\kappa}{s}) : (\quant{\alpha:\kappa}{\sigma})$ if
    $s : \sigma$ and $\alpha$ does not occur free in the type of a
    free variable of~$s$.
  \item $\app{s}{t} : \tau$ if $s : \sigma \arrtype \tau$ and
    $t : \sigma$
  \item $\tapp{s}{\tau} : \sigma[\subst{\alpha}{\tau}]$ if
    $s : \quant{\alpha:\kappa}{\sigma}$ and~$\tau$ is a type
    constructor of kind~$\kappa$,
  \item $s : \tau$ if $s : \tau'$ and $\tau =_\beta \tau'$.
  \end{itemize}
  The set of free variables of a preterm~$t$, denoted $\FV(t)$, is
  defined in the expected way. Analogously, we define the
  set~$\FTV(t)$ of type constructor variables occurring free in~$t$.
  If $\alpha$ is a type then we use the notation $\tabs{\alpha}{t}$.
  We denote an occurrence of a variable~$x$ of type~$\tau$
  by~$x^\tau$,
  e.g.~$\lambda x : \tau\arrtype\sigma
  . x^{\tau\arrtype\sigma}y^\tau$. When clear or irrelevant, we omit
  the type annotations, denoting the above term by~$\lambda x . x
  y$. Type substitution is defined in the expected way except that it
  needs to change the types of variables. Formally, a type
  substitution changes the types associated to variables in~$\Vars$. We
  define the equivalence relation~$\equiv$ by: $s \equiv t$ iff $s$
  and $t$ are identical modulo $\beta$-conversion in types.
\end{defn}

Note that we present terms in orthodox Church-style, i.e.,
instead of using contexts each variable has a globally fixed type
associated to it.

\begin{lemma}
  If $s : \tau$ and $s \equiv t$ then $t : \tau$.
\end{lemma}

\begin{proof}
  Induction on~$s$.
\end{proof}

\begin{defn}\label{def_terms}
  The set of \emph{terms} is the set of the equivalence classes
  of~$\equiv$.
\end{defn}

Because $\beta$-reduction on types is confluent and terminating, every
term has a canonical preterm representative -- the one with all types
occurring in it $\beta$-normalised.
We define $\FTV(t)$
as the value of~$\FTV$ on the canonical representative of~$t$.
We say that $t$ is \emph{closed} if both $\FTV(t) = \emptyset$
and $\FV(t) = \emptyset$.
Because typing and term formation operations (abstraction,
application, \ldots) are invariant under~$\equiv$, we may denote terms
by their (canonical) representatives and informally treat them
interchangeably.

We will often abuse notation to omit $\cdot$ and $*$. Thus, $s t$ can
refer to both $\app{s}{t}$ and $\tapp{s}{t}$. This is not ambiguous
due to typing. When writing $\sigma[\subst{\alpha}{\tau}]$ we
implicitly assume that $\alpha$ and $\tau$ have the same
kind. Analogously with $t[\subst{x}{s}]$.

\begin{lemma}[Substitution lemma]\label{lem:substitution}
  \begin{enumerate}
  \item If $s : \tau$ and $x : \sigma$ and $t : \sigma$ then
    $s[\subst{x}{t}] : \tau$.
  \item If $t : \sigma$ then
    $t[\subst{\alpha}{\tau}] : \sigma[\subst{\alpha}{\tau}]$.
  \end{enumerate}
\end{lemma}

\begin{proof}
  Induction on the typing derivation.
\end{proof}

\begin{lemma}[Generation lemma]\label{lem:generation}
  If $t : \sigma$ then there is a type~$\sigma'$ such that
  $\sigma' =_\beta \sigma$ and $\FTV(\sigma') \subseteq \FTV(t)$ and
  one of the following holds.
  \begin{itemize}
  \item $t \equiv x$ is a variable with $(x : \sigma') \in \Vars$.
  \item $t \equiv \mathtt{f}$ is a function symbol with
    $\mathtt{f} : \sigma'$ in $\Sigma$.
  \item $t \equiv \abs{x:\tau_1}{s}$ and
    $\sigma'=\tau_1\arrtype\tau_2$ and $s : \tau_2$.
  \item $t \equiv \tabs{\alpha:\kappa}{s}$ and
    $\sigma' = \quant{\alpha:\kappa}{\tau}$ and $s : \tau$ and
    $\alpha$ does not occur free in the type of a free variable
    of~$s$.
  \item $t \equiv \app{t_1}{t_2}$ and $t_1 : \tau \arrtype \sigma'$
    and $t_2 : \tau$ and $\FTV(\tau) \subseteq \FTV(t)$.
  \item $t \equiv \tapp{s}{\tau}$ and
    $\sigma' = \rho[\subst{\alpha}{\tau}]$ and
    $s : \quant{(\alpha:\kappa)}{\rho}$ and~$\tau$ is a type
    constructor of kind~$\kappa$.
  \end{itemize}
\end{lemma}

\begin{proof}
  By analysing the derivation $t : \sigma$. To ensure
  $\FTV(\sigma') \subseteq \FTV(t)$, note that if
  $\alpha \notin \FTV(t)$ is of kind~$\kappa$ and~$t : \sigma'$, then
  $t : \sigma'[\subst{\alpha}{\chi_\kappa}]$ by the substitution lemma
  (thus we can eliminate~$\alpha$).
\end{proof}

\section{Polymorphic Functional Systems}\label{sec_systems}

In this section, we present a form of higher-order term rewriting
systems based on $\Fomega$: \emph{Polymorphic Functional Systems
  (PFSs)}. Systems of interest, such as logic systems like~ICP2 and
higher-order TRSs with shallow or full polymorphism can be encoded
into PFSs, and then proved terminating with the technique we will
develop in Sections \ref{sec:World}--\ref{sec_rule_removal}.

\begin{defn}\label{def_pafs_types_terms}
  \emph{Kinds}, \emph{type constructors} and \emph{types} are defined
  like in Definition~\ref{def_types}, parameterised by a fixed
  set~$\Sigma^T = \bigcup_{\kappa}\Sigma^T_\kappa$ of type constructor
  symbols.

  Let~$\Sigma$ be a set of function symbols such
  that for $\mathtt{f} : \sigma \in \Sigma$:
    \[
    \sigma = \forall (\alpha_1 : \kappa_1) \ldots \forall (\alpha_n : \kappa_n)
    . \sigma_1 \arrtype \ldots \arrtype \sigma_k \arrtype \tau
    \quad\quad (\text{with}\ \tau\ \text{a type atom})
    \]
    We define \emph{PFS terms} as in Definition~\ref{def_terms} (based
    on Definition~\ref{def_preterms}), parameterised by~$\Sigma$, with
    the restriction that for any subterm $\app{s}{u}$ of a term~$t$,
    we have $s = \mathtt{f} \rho_1 \ldots \rho_n u_1 \ldots u_m$
    where:
    \[
    \mathtt{f} : \forall (\alpha_1 : \kappa_1) \ldots
    \forall (\alpha_n : \kappa_n) . \sigma_1 \arrtype \ldots \arrtype
    \sigma_k \arrtype \tau
    \quad\quad (\text{with}\ \tau\ \text{a type atom and}\ k > m)
    \]
\end{defn}

This definition does not allow for a variable or
abstraction to occur at the head of an application, nor can we have
terms of the form $s \cdot t * \tau \cdot q$ (although terms of the
form $s \cdot t * \tau$, or $x * \tau$ with $x$ a variable,
\emph{are} allowed to occur).  To stress this restriction, we will
use the notation
$\mathtt{f}_{\rho_1,\ldots,\rho_n}(s_1,\ldots,s_m)$ as an alternative
way to denote
$\mathtt{f} \rho_1 \ldots \rho_n s_1 \ldots s_m$ when
$
  \mathtt{f} : \forall (\alpha_1 : \kappa_1) \ldots
  \forall (\alpha_n : \kappa_n) . \sigma_1 \arrtype \ldots \arrtype
  \sigma_k \arrtype \tau
$
is a function symbol in~$\Sigma$ with~$\tau$ a type atom and $m \leq k$.
This allows us to represent terms in a ``functional'' way, where
application does not explicitly occur (only implicitly in the
construction of $\mathtt{f}_{\rho_1,\ldots,\rho_n}(s_1,\ldots,s_m)$).

The following result follows easily by induction on term
structure:

\begin{lemma}
  If $t,s$ are PFS terms then so is $t[\subst{x}{s}]$.
\end{lemma}

PFS terms will be rewritten through a reduction relation
$\arr{\Rules}$ based on a (usually infinite) set of rewrite rules. To
define this relation, we need two additional notions.

\begin{defn}\label{def_replacement}
  A \emph{replacement} is a function $\delta = \gamma \circ \omega$
  satisfying:
  \begin{enumerate}
  \item $\omega$ is a type constructor substitution,
  \item $\gamma$ is a term substitution such that
    $\gamma(\omega(x)) : \omega(\tau)$ for every
    $(x : \tau) \in \Vars$.
  \end{enumerate}

  For~$\tau$ a type constructor, we use $\delta(\tau)$ to denote
  $\omega(\tau)$. We use the notation
  $\delta[\subst{x}{t}] = \gamma[\subst{x}{t}] \circ \omega$. Note
  that if $t : \tau$ then $\delta(t) : \delta(\tau)$.
\end{defn}

\begin{defn}\label{def:context}
  A \emph{$\sigma$-context}~$C_\sigma$ is a PFS term with a fresh function
  symbol $\Box_\sigma \notin \Sigma$ of type~$\sigma$ occurring
  exactly once. By~$C_\sigma[t]$ we denote a PFS term obtained
  from~$C_\sigma$ by substituting~$t$ for~$\Box_\sigma$. We drop the
  $\sigma$ subscripts when clear or irrelevant.
\end{defn}

Now, the rewrite rules are simply a set of term pairs, whose
monotonic closure generates the rewrite relation.

\begin{defn}\label{def_rules}
  A set $\Rules$ of term pairs $(\ell,r)$ is a set of \emph{rewrite
    rules} if: (a) $\FV(r) \subseteq \FV(\ell)$; (b) $\ell$ and $r$
  have the same type; and (c) if $(\ell,r) \in \Rules$
  then $(\delta(\ell),\delta(r)) \in \Rules$ for any
  replacement~$\delta$.  The reduction relation $\arr{\Rules}$
  on PFS terms is defined by:

\begin{center}
  $t \arr{\Rules}
  s$ iff $t = C[\ell]$ and $s = C[r]$ for some $(\ell,r)\in\Rules$ and
  context~$C$.
\end{center}
\end{defn}

\begin{defn}\label{def_pafs}
  A \emph{Polymorphic Functional System (PFS)} is a triple
  $(\Sigma^T,\Sigma,\Rules)$ where~$\Sigma^T$ is a set of type
  constructor symbols, $\Sigma$ a set of function symbols (restricted
  as in Def.~\ref{def_pafs_types_terms}), and $\Rules$ is a set
  of rules as in Definition~\ref{def_rules}. A term of a
  PFS~$A$ is referred to as an $A$-term.
\end{defn}

While PFS-terms are a restriction from the general terms
of system $\Fomega$, the reduction relation allows us to actually encode,
e.g., system~$\mathtt{F}$ as a PFS: we can do so by including the symbol
${@} : \forall\alpha\forall\beta . (\alpha \arrtype \beta) \arrtype \alpha
\arrtype \beta$ in $\Sigma$ and adding all rules of the form
$@_{\sigma,\tau}(\abs{x}{s},t) \red s[x:=t]$.
Similarly, $\beta$-reduction of type abstraction can be modelled
by including a symbol
$\mathtt{A} : \forall \alpha : * \arrkind * . \forall \beta . (\forall
\gamma.\alpha \gamma) \arrtype \alpha \beta$ and rules
$\mathtt{A}_{\abs{\gamma}{\sigma},\tau}(\tabs{\gamma}{s}) \red s[\gamma:=\tau]$.%
  \footnote{The use of a type constructor variable $\alpha$ of kind
  $* \arrkind *$ makes it possible to do type substitution as part of
  a rule.
  An application $s * \tau$ with $s : \quant{\gamma}{\sigma}$ is
  encoded as $\mathtt{A}_{\abs{\gamma}{\sigma},\tau}(s)$, so
  $\alpha$ is substituted with $\abs{\gamma}{\tau}$.
  This is well-typed because $(\abs{\gamma}{\sigma})\gamma =_\beta
  \sigma$ and $(\abs{\gamma}{\sigma})\tau =_\beta \sigma[\gamma:=\tau]$.
  }
We can also use rules
$(\tabs{\alpha}{s})*\tau \red s[\alpha:=\tau]$ without the extra
symbol, but to apply our method it may be
convenient
to use the extra symbol, as
it creates more liberty in choosing an interpretation.

\begin{example}[Fold on heterogenous lists]\label{ex_fold_pafs}
  The example from the introduction may be represented as a PFS with
  one type symbol $\mathtt{List} : *$, the following function symbols:
  \[
  \begin{array}{rcl}
    @ & : & \forall \alpha \forall \beta . (\alpha \arrtype \beta) \arrtype \alpha \arrtype \beta \\
    \mathtt{A} & : & \forall \alpha : * \arrkind * . \forall \beta .
    (\forall \gamma .\alpha \gamma) \arrtype \alpha \beta \\
    \mathtt{nil} & : & \List \\
    \mathtt{cons} & : & \forall \alpha . \alpha \arrtype \List \arrtype \List \\
    \mathtt{foldl} & : & \forall \beta . (\forall \alpha . \beta \arrtype \alpha \arrtype \beta) \arrtype \beta \arrtype \List \arrtype \beta
  \end{array}
  \]
  and the following rules (which formally represents an
  infinite set of rules: one rule for each choice of types $\sigma,
  \tau$ and PFS terms $s$, $t$, etc.):
  \[
  \begin{array}{rcl}
    @_{\sigma,\tau}(\abs{x:\sigma}{s},t) & \red & s[x:=t] \\
    \mathtt{A}_{\abs{\alpha}{\sigma},\tau}(\tabs{\alpha}{s}) & \red &
    s[\alpha:=\tau] \\
    \mathtt{foldl}_\sigma(f,s,\nil) & \red & s \\
    \mathtt{foldl}_\sigma(f,s,\cons_\tau(h,t)) & \red & \mathtt{foldl}_\sigma(f,@_{\tau,\sigma}(@_{\sigma,\tau
    \arrtype\sigma}(\mathtt{A}_{\abs{\alpha}{\sigma\arrtype\alpha\arrtype\sigma},\tau}(f),s),h),t)
  \end{array}
  \]
\end{example}

\section{A well-ordered set of interpretation terms}\label{sec:World}

In polynomial interpretations of first-order term
rewriting~\cite[Chapter 6.2]{Terese2003}, each term $s$ is mapped to a
natural number $\interpret{s}$, such that $\interpret{s} > \interpret{t}$
whenever $s \arr{\Rules} t$.  In higher-order rewriting, this is not
practical; instead, following \cite{pol:96}, terms are mapped to weakly
monotonic functionals according to their type (i.e., terms with a $0$-order
type are mapped to natural numbers, terms with a $1$-order type to weakly
monotonic functions over natural numbers, terms with a $2$-order type to
weakly monotonic functionals taking weakly monotonic functions as
arguments, and so on).  In this paper, to account for full polymorphism,
we will interpret PFS terms to a set $\Iterms$ of \emph{interpretation terms}
in a specific extension of System~$\Fomega$.  This set is defined in Section
\ref{subsec:I}; we provide a well-founded partial ordering $\succ$ on
$\Iterms$ in Section \ref{subsec:succ}.

Although our world of interpretation terms is quite different
from the weakly monotonic functionals of \cite{pol:96}, there are many
similarities.  Most pertinently, every interpretation term $\abs{x}{s}$
essentially defines a weakly monotonic function from $\Iterms$ to
$\Iterms$.  This, and the use of both addition and multiplication in
the definition of $\Iterms$, makes it possible to lift higher-order
polynomial interpretations \cite{FuhsKop2012} to our setting.  We
prove weak monotonicity in Section \ref{subsec:weakmono}.

\subsection{Interpretation terms}\label{subsec:I}

\begin{defn}\label{def_iterms}
  The set~$\ITypes$ of \emph{interpretation types} is the set of types
  as in Definition~\ref{def_types} with $\Sigma^T = \{ \nat : * \}$,
  i.e., there is a single type constant~$\nat$. Then $\chi_* = \nat$.

  The set~$\Iterms$ of \emph{interpretation terms} is the set of terms
  from Definition~\ref{def_terms} (see also
  Definition~\ref{def_preterms}) where as types we take the
  interpretation types and for the set~$\Sigma$ of function symbols we
  take $\Sigma = \{ n : \nat \mid n \in \Nbb \} \cup \Sigma_f$, where
  $ \Sigma_f = \{ \oplus : \forall \alpha . \alpha \arrtype \alpha
  \arrtype \alpha, \otimes : \forall \alpha . \alpha \arrtype \alpha
  \arrtype \alpha, \flatten : \forall \alpha . \alpha \arrtype \nat,
  \lift : \forall \alpha . \nat \arrtype \alpha \} $.
\end{defn}

For easier presentation, we write $\oplus_\tau$, $\otimes_\tau$, etc.,
instead of $\tapp{\oplus}{\tau}$, $\tapp{\otimes}{\tau}$, etc. We will
also use $\oplus$ and $\otimes$ in \emph{infix, left-associative}
notation, and omit the type denotation where it is clear from
context. Thus, $s \oplus t \oplus u$ should be read as
$\oplus_\sigma\,(\oplus_\sigma\,s\,t)\,u$ if $s$ has type
$\sigma$. Thus, our interpretation terms include natural
  numbers with the operations of addition and multiplication. It would
  not cause any fundamental problems to add more monotonic operations, e.g.,
  exponentiation, but we refrain from doing so for the sake of
  simplicity.

\paragraph*{Normalising interpretation terms}

The set $\Iterms$ of interpretation terms can be reduced through
a relation $\arrW$, that we will define below.  This relation will
be a powerful aid in defining the partial ordering $\succ$ in Section
\ref{subsec:succ}.

\begin{defn}
  We define the relation $\arrW$ on interpretation terms as the
  smallest relation on~$\Iterms$ for which the following properties
  are satisfied:
  \begin{enumerate}
  \item\label{arrW:mono:abs}
    if $s \arrW t$ then both $\abs{x}{s} \arrW \abs{x}{t}$ and
    $\tabs{\alpha}{s} \arrW \tabs{\alpha}{t}$
  \item\label{arrW:mono:right}
    if $s \arrW t$ then $\app{u}{s} \arrW \app{u}{t}$
  \item\label{arrW:mono:left}
    if $s \arrW t$ then both $\app{s}{u} \arrW \app{t}{u}$ and
    $\tapp{s}{\sigma} \arrW \tapp{t}{\sigma}$
  \item\label{arrW:beta:abs} $\app{(\abs{x:\sigma}{s})}{t} \arrW
    s[\subst{x}{t}]$
    and
    $\tapp{(\tabs{\alpha}{s})}{\sigma}
    \arrW s[\subst{\alpha}{\sigma}]$
    ($\beta$-reduction)
  \item\label{arrW:plus:base}
    $\app{\app{\oplus_{\nat}}{n}}{m} \arrW n+m$
    and
    $\app{\app{\otimes_{\nat}}{n}}{m}
    \arrW n \times m$
  \item\label{arrW:circ:arrow} $\app{\app{\circ_{\sigma \arrtype
        \tau}}{s}}{t} \arrW
    \abs{x:\sigma}{\app{\app{\circ_\tau}{(\app{s}{x})}}{(\app{t}{x})}}$
    for $\circ \in \{ \oplus, \otimes \}$
  \item\label{arrW:circ:forall}
    $\app{\app{\circ_{\quant{\alpha}{\sigma}}}{s}}{t} \arrW
    \tabs{\alpha}{\app{\app{\circ_\sigma}{(\tapp{s}{\alpha})}}{(
        \tapp{t}{\alpha})}}$ for $\circ \in \{ \oplus, \otimes \}$
  \item $\app{\flatten_\nat}{s} \arrW s$
  \item $\app{\flatten_{\sigma \arrtype \tau}}{s} \arrW
    \app{\flatten_\tau}{(\app{s}{(\app{\lift_\sigma}{0})})}$
  \item $\app{\flatten_{\quant{\alpha:\kappa}{\sigma}}}{s} \arrW
    \app{\flatten_{\sigma[\subst{\alpha}{\chi_\kappa}]}}{(\tapp{s}{\chi_\kappa})}$
  \item $\app{\lift_\nat}{s} \arrW s$
  \item $\app{\lift_{\sigma \arrtype \tau}}{s} \arrW
    \abs{x:\sigma}{\app{\lift_{\tau}}{s}}$
  \item $\app{\lift_{\quant{\alpha}{\sigma}}}{s} \arrW
    \tabs{\alpha}{\app{\lift_{\sigma}}{s}}$
  \end{enumerate}
  Recall Definition~\ref{def_terms} and Definition~\ref{def_iterms} of
  the set of interpretation terms~$\Iterms$ as the set of the
  equivalence classes of~$\equiv$. So, for instance,
  $\lift_\nat$ above denotes
  the equivalence class of all preterms $\lift_\sigma$ with
  $\sigma =_\beta \nat$. Hence, the above rules are invariant
  under~$\equiv$ (by confluence of $\beta$-reduction on types), and
  they correctly define a relation on interpretation terms. We say
  that $s$ is a \emph{redex} if $s$ reduces by one of the rules 4--13.
  A \emph{final interpretation term} is an interpretation term $s \in
  \Iterms$ such that (a) $s$ is closed, and (b) $s$ is in normal form
  with respect to $\arrW$.  We let $\Iterms^f$ be the set of all final
  interpretation terms. By~$\Iterms_\tau$ ($\Iterms^f_\tau$) we denote
  the set of all (final) interpretation terms of interpretation
  type~$\tau$.
\end{defn}

An important difference with System~$\Fomega$ and related ones is that
the rules for $\oplus_\tau$, $\otimes_\tau$, $\flatten_\tau$ and
$\lift_\tau$ depend on the type~$\tau$. In particular, type
substitution in terms may create redexes. For instance, if $\alpha$ is
a type variable then $\oplus_\alpha t_1 t_2$ is not a redex, but
$\oplus_{\sigma\arrtype\tau} t_1 t_2$ is. This makes the question of
termination subtle. Indeed, System~$\Fomega$ is extremely sensitive to
modifications which are not of a logical nature. For instance, adding
a constant $\mathtt{J} : \forall \alpha \beta . \alpha \arrtype \beta$
with a reduction rule
$\mathtt{J} \tau \tau \leadsto \lambda x : \tau . x$ makes the system
non-terminating~\cite{Girard1971}. This rule breaks parametricity by
making it possible to compare two arbitrary types. Our rules do not allow such a
definition. Moreover, the natural number constants cannot be
  distinguished ``inside'' the system. In other words, we could
  replace all natural number constants with 0 and this would not
  change the reduction behaviour of terms. So for the purposes of
  termination, the type $\nat$ is essentially a singleton. This
  implies that, while we have polymorphic functions between an
  arbitrary type $\alpha$ and $\nat$ which are not constant when seen
  ``from outside'' the system, they are constant for the purposes of
  reduction ``inside'' the system (as they would have to be in a
  parametric $\Fomega$-like system). Intuitively, these properties of
  our system ensure that it stays ``close enough'' to $\Fomega$ so
  that the standard termination proof still generalises.

Now we state some properties of~$\arrW$, including strong
normalisation. Because of space limitations, most (complete)
proofs are delegated to 
\onlypaper{\cite[Appendix~A.1]{versionwithappendix}}%
\onlyarxiv{Appendix~\ref{app_proofs_SN}}.

\begin{lemma}[Subject reduction]
  If $t : \tau$ and $t \arrW t'$ then $t' : \tau$.
\end{lemma}

\begin{proof}
  By induction on the definition of $t \arrW t'$, using
  Lemmas \ref{lem:substitution} and \ref{lem:generation}.
\end{proof}

\begin{theorem}\label{thm_sn}
  If $t : \sigma$ then $t$ is terminating with respect to $\arrW$.
\end{theorem}

\begin{proof}
  By an adaptation of the Tait-Girard computability method. The proof
  is an adaptation of chapters~6 and~14 from the
  book~\cite{Girard1989}, and chapters~10 and~11 from the
  book~\cite{SorensenUrzyczyn2006}.
  Details are available in 
  \onlypaper{\cite[Appendix A.1]{versionwithappendix}}%
  \onlyarxiv{Appendix~\ref{app_proofs_SN}}.
\end{proof}

\begin{lemma}\label{lem_unique_final}
  Every term $s \in \Iterms$ has a unique normal form~$s\da$. If~$s$
  is closed then so is~$s\da$.
\end{lemma}

\begin{proof}
  One easily
  checks that~$\arrW$ is locally confluent. Since the
  relation is
  terminating by Theorem~\ref{thm_sn}, it is confluent by Newman's
  lemma.
\end{proof}

\begin{lemma}\label{lem_final_nat}
  The only final interpretation terms of type $\nat$ are the natural
  numbers.
\end{lemma}

\begin{example}\label{ex:arrWreduce}
Let $s \in \Iterms_{\nat \arrtype \nat}$ and $t \in
\Iterms_\nat$. Then we can reduce
$(s \oplus \lift_{\nat \arrtype \nat}(1)) \cdot t \arrW
(\abs{x}{s x \oplus \lift_{\nat \arrtype \nat}(1)x}) \cdot t \arrW
s t \oplus \lift_{\nat \arrtype \nat}(1)t \arrW
s t \oplus (\abs{y}{\lift_{\nat}(1)})t \arrW
s t \oplus \lift_\nat(1) \arrW
s t \oplus 1$. If $s$ and $t$ are variables, this term is in normal
form.
\end{example}

\subsection{The ordering pair $(\succeq,\succ)$}\label{subsec:succ}

With these ingredients, we are ready to define the well-founded
partial ordering $\succ$ on $\Iterms$.  In fact, we will do more: rather
than a single partial ordering, we will define an \emph{ordering pair}: a
pair of a quasi-ordering $\succeq$ and a compatible well-founded ordering
$\succ$. The quasi-ordering $\succeq$ often makes it easier to prove
%that
$s \succ t$, since it suffices to show that $s \succeq s' \succ t'
\succeq t$ for some interpretation terms $s',t'$.  Having $\succeq$ will
also allow us to use rule removal (Theorem \ref{thm:ruleremove}).
%instead of merely
%orienting all rules with $\succ$.

\begin{defn}\label{def:succ}
  Let $R \in \{ \succ^0,\succeq^0 \}$. For \emph{closed}~$s,t\in\Iterms_\sigma$
  and closed~$\sigma$ in $\beta$-normal form, the relation
  $s\ R_{\sigma}\ t$ is defined coinductively by the following rules.
  \[
  \begin{array}{ccc}
    \infer={s\ R_\nat\ t}{s\da\ R\ t\da \text{ in }\mathbb{N}} \quad&\quad
    \infer={s\ R_{\sigma\arrtype\tau}\ t}{\app{s}{q}\ R_{\tau}\ \app{t}{q} \text{ for all } q \in \Iterms^f_\sigma} &
    \infer={s\ R_{\forall(\alpha:\kappa).\sigma}\ t}{\tapp{s}{\tau}\ R_{\nf_\beta(\sigma[\subst{\alpha}{\tau}])}\ \tapp{t}{\tau} \text{ for all closed } \tau \in \Tc_{\kappa}}
  \end{array}
  \]
  We define $s \approx_\sigma^0 t$ if both $s \succeq_\sigma^0 t$ and
  $t \succeq_\sigma^0 s$.  We drop the type subscripts when clear or
  irrelevant.
\end{defn}

Note that in the case for~$\nat$ the terms~$s\da$, $t\da$ are natural
numbers by Lemma~\ref{lem_final_nat} ($s\da,t\da$ are closed and in
normal form, so they are final interpretation terms).

Intuitively, the above definition means that e.g. $s \succ^0 t$ iff
there exists a possibly infinite derivation tree using the above
rules. In such a derivation tree all leaves must witness $s\da > t\da$
in natural numbers. However, this also allows for infinite branches,
which solves the problem of repeating types due to impredicative
polymorphism. If e.g.~$s \succ_{\forall \alpha . \alpha}^0 t$ then
$\tapp{s}{\forall\alpha.\alpha} \succ_{\forall \alpha . \alpha}^0
\tapp{t}{\forall\alpha.\alpha}$, which forces an infinite branch in
the derivation tree. According to our definition, any infinite branch
may essentially be ignored.

Formally, the above coinductive definition of e.g.~$\succ_\sigma^0$
may be interpreted as defining the largest relation such that if $s
\succ_\sigma^0 t$ then:
\begin{itemize}
\item $\sigma = \nat$ and $s\da > t\da$ in $\mathbb{N}$, or
\item $\sigma = \tau_1\arrtype\tau_2$ and
  $\app{s}{q} \succ_{\tau_2}^0 \app{t}{q}$ for all
  $q \in \Iterms^f_{\tau_1}$, or
\item $\sigma = \forall(\alpha:\kappa).\rho$ and
  $\tapp{s}{\tau} \succ_{\nf_\beta(\rho[\subst{\alpha}{\tau}])}^0
  \tapp{t}{\tau}$ for all closed $\tau \in \Tc_{\kappa}$.
\end{itemize}
For more background on coinduction see
e.g.~\cite{KozenSilva2017,Sangiorgi2012,JacobsRutten2011}. In this
paper we use a few simple coinductive proofs to establish the basic
properties of~$\succ$ and~$\succeq$. Later, we just use these
properties and the details of the definition do not matter.

\begin{defn}\label{def_closure}
  A \emph{closure}~$\cl = \gamma \circ \omega$ is a
  replacement such that $\omega(\alpha)$ is closed for each
  type constructor variable~$\alpha$, and $\gamma(x)$ is closed for
  each term variable~$x$.
  For arbitrary types~$\sigma$ and arbitrary terms $s,t \in \Iterms$
  we define $s \succ_\sigma t$ if for every closure~$\cl$ we can
  obtain $\cl(s) \succ_{\nf_\beta(\cl(\sigma))}^c \cl(t)$
  coinductively with the above rules. The relations $\succeq_\sigma$
  and $\approx_\sigma$ are defined analogously.
\end{defn}

Note that for closed $s,t$ and closed~$\sigma$ in $\beta$-normal form,
$s \succ_\sigma t$ iff $s \succ_\sigma^0 t$ (and analogously
for~$\succeq,\approx$). In this case we shall often omit the
superscript~$0$.

The definition of~$\succ$ and~$\succeq$ may be reformulated as
follows.

\begin{lemma}\label{lem_succ_explicit}
  $t \succeq s$ if and only if for every closure~$\cl$ and every
  sequence $u_1,\ldots,u_n$ of closed terms and closed type
  constructors such that $\cl(t) u_1 \ldots u_n : \nat$ we have
  $(\cl(t) u_1 \ldots u_n)\da \ge (\cl(s) u_1 \ldots u_n)\da$ in
  natural numbers. An analogous result holds with $\succ$ or $\approx$
  instead of~$\succeq$.
\end{lemma}

\begin{proof}
  The direction from left to right follows by induction on~$n$; the
  other by coinduction.
\end{proof}

In what follows, all proofs by coinduction could be reformulated to
instead use the lemma above. However, this would arguably make the
proofs less perspicuous. Moreover, a coinductive definition is better
suited for a formalisation -- the coinductive proofs here could be
written in Coq almost verbatim.

Our next task is to show that $\succeq$ and $\succ$ have the
desired properties of an ordering pair; e.g., transitivity and
compatibility. We first state a simple lemma that will be used
implicitly.

\begin{lemma}
  If $\tau \in \ITypes$ is closed and $\beta$-normal, then
  $\tau = \nat$ or $\tau = \tau_1\arrtype\tau_2$ or
  $\tau = \forall\alpha\sigma$.
\end{lemma}

%The most important property of~$\succ$ is that it is a well-founded
%ordering.

\begin{lemma}\label{lem_well_founded}
  $\succ$ is well-founded.
\end{lemma}

\begin{proof}
  It suffices to show this for closed terms and closed types in
  $\beta$-normal form, because any infinite sequence $t_1 \succ_\tau
  t_2 \succ_\tau t_3 \succ_\tau \ldots$ induces an infinite sequence
  $\cl(t_1) \succ_{\nf_\beta(\cl(\tau))} \cl(t_2)
  \succ_{\nf_\beta(\cl(\tau))} \cl(t_3) \succ_{\nf_\beta(\cl(\tau))}
  \ldots$ for any closure~$\cl$. By induction on the size of a
  $\beta$-normal type~$\tau$ (with size measured as the number of
  occurrences of~$\forall$ and~$\arrtype$) one proves that there does
  not exist an infinite sequence $t_1 \succ_\tau t_2 \succ_\tau t_3
  \succ_\tau \ldots$ For instance, if $\alpha$ has kind~$\kappa$ and
  $t_1 \succ_{\forall\alpha\tau} t_2 \succ_{\forall\alpha\tau} t_3
  \succ_{\forall\alpha\tau} \ldots$ then $\tapp{t_1}{\chi_\kappa}
  \succ_{\tau'} \tapp{t_2}{\chi_\kappa} \succ_{\tau'}
  \tapp{t_3}{\chi_\kappa} \succ_{\tau'} \ldots$, where
  $\tau'=\nf_\beta(\tau[\subst{\alpha}{\chi_\kappa}])$. Because $\tau$
  is in $\beta$-normal form, all redexes in
  $\tau[\subst{\alpha}{\chi_\kappa}]$ are created by the substitution
  and must have the form $\chi_\kappa u$. Hence, by the definition
  of~$\chi_\kappa$ (see Definition~\ref{def_types}) the
  type~$\tau'$ is smaller than~$\tau$. This
  contradicts the inductive hypothesis.
\end{proof}

\begin{lemma}\label{lem_transitive}
  Both $\succ$ and $\succeq$ are transitive.
\end{lemma}

\begin{proof}
  We show this for~$\succ$, the proof for~$\succeq$ being
  analogous. Again, it suffices to prove this for closed
  terms and closed types in $\beta$-normal form. We proceed by
  coinduction.

  If $t_1 \succ_\nat t_2 \succ_\nat t_3$ then $t_1\da > t_2\da >
  t_3\da$, so $t_1\da > t_3\da$. Thus $t_1 \succ_\nat t_3$.

  If $t_1 \succ_{\sigma\arrtype\tau}t_2\succ_{\sigma\arrtype\tau}t_3$
  then $\app{t_1}{q}\succ_{\tau}\app{t_2}{q}\succ_\tau\app{t_3}{q}$
  for $q \in \Iterms^f_\sigma$. Hence
  $\app{t_1}{q}\succ_\tau\app{t_3}{q}$ for $q \in \Iterms^f_\sigma$ by
  the coinductive hypothesis. Thus $t_1\succ_{\sigma\arrtype\tau}
  t_3$.

  If $t_1
  \succ_{\forall(\alpha:\kappa)\sigma}t_2\succ_{\forall(\alpha:\kappa)\sigma}t_3$
  then
  $\tapp{t_1}{\tau}\succ_{\sigma'}\tapp{t_2}{\tau}\succ_{\sigma'}\tapp{t_3}{\tau}$
  for any closed~$\tau$ of kind~$\kappa$, where
  $\sigma' = \nf_\beta(\sigma[\subst{\alpha}{\tau}])$. By the
  coinductive hypothesis
  $\tapp{t_1}{\tau}\succ_{\sigma'}\tapp{t_3}{\tau}$; thus
  $t_1\succ_{\forall\alpha\sigma}t_3$.
\end{proof}

\begin{lemma}\label{lem_reflexive}
  $\succeq$ is reflexive.
\end{lemma}

\begin{proof}
  By coinduction one shows that $\succeq_\sigma$ is reflexive on
  closed terms for closed $\beta$-normal~$\sigma$. The case
  of~$\succeq$ is then immediate from definitions.
\end{proof}

\begin{lemma}\label{lem:compatibility}
  The relations~$\succeq$ and~$\succ$ are compatible, i.e., $\succ
  \cdot \succeq\ \subseteq\ \succ$ and $\succeq \cdot
  \succ\ \subseteq\ \succ$.
\end{lemma}

\begin{proof}
  By coinduction, analogous to the transitivity proof.
\end{proof}

\begin{lemma}\label{lem_succ_to_succeq}
  If $t \succ s$ then $t \succeq s$.
\end{lemma}

\begin{proof}
  By coinduction.
\end{proof}

\begin{lemma}\label{lem_leadsto_to_approx}
  If $t \arrW s$ then $t \approx s$.
\end{lemma}

\begin{proof}
  Follows from Lemma~\ref{lem_succ_explicit}, noting that $t \arrW s$
  implies $\cl(t) \arrW \cl(s)$ for all closures~$\cl$.
\end{proof}

\begin{lemma}\label{lem_succ_red}
  Assume $t \succ s$ (resp.~$t \succeq s$). If $t \leadsto t'$ or
  $t' \leadsto t$ then $t' \succ s$ (resp.~$t' \succeq s$). If
  $s \leadsto s'$ or $s' \leadsto s$ then $t \succ s'$
  (resp.~$t \succeq s'$).
\end{lemma}

\begin{proof}
  Follows from Lemma~\ref{lem_leadsto_to_approx}, transitivity and
  compatibility.
\end{proof}

\begin{corollary}\label{cor_succ_da}
  For $R \in \{\succ,\succeq,\approx\}$: $s\ R\ t$ if and only if
  $s\downarrow\ R\ t\downarrow$.
\end{corollary}

\begin{example}\label{ex:plus1}
We can prove that $x \oplus \lift_{\nat \arrtype \nat}(1)
\succ x$: by
definition, this holds if $s \oplus \lift_{\nat \arrtype \nat}(1) \succ
s$ for all closed $s$, so if $(s \oplus \lift_{\nat \arrtype \nat}(1))u
\succ s u$ for all closed $s,u$.
Following Example \ref{ex:arrWreduce} and Lemma \ref{lem_succ_red},
this holds if $s u \oplus 1 \succ s u$.  By definition, this is the
case if $(s u \oplus 1)\downarrow > (s u)\downarrow$ in the natural numbers,
which clearly holds for any $s,u$.
\end{example}

\subsection{Weak monotonicity}\label{subsec:weakmono}

We will now show that $s \succeq s'$ implies $t[\subst{x}{s}] \succeq
t[\subst{x}{s'}]$ (weak monotonicity).
For this purpose, we prove a few lemmas, many of
which also apply to~$\succ$, stating the preservation of~$\succeq$
under term formation operations. We will need these results in the next section.

\begin{lemma}\label{lem_app_succ}
  For $R \in \{\succeq,\succ\}$: if $t\:R\:s$ then $t u\:R\:s u$ with
  $u$ a term or type constructor.
\end{lemma}

\begin{proof}
  Follows from definitions.
\end{proof}

\begin{lemma}\label{lem:liftgreater}
  For $R \in \{\succeq,\succ\}$: if $n\:R\:m$ then
  $\lift_\sigma n\:R\:\lift_\sigma m$ for all types $\sigma$.
\end{lemma}

\begin{proof}
  Without loss of generality we may assume $\sigma$ closed and in
  $\beta$-normal form. By coinduction we show
  $\lift(n) u_1 \ldots u_k \succeq \lift(m) u_1 \ldots u_k$ for closed
  $u_1,\ldots,u_k$. First note that
  $(\lift\,t) u_1 \ldots u_k \leadsto^* \lift(t)$ (with a different
  type subscript in~$\lift$ on the right side, omitted for
  conciseness). If $\sigma = \nat$ then
  $(\lift(n) u_1 \ldots u_k)\da = n \ge m = (\lift(m) u_1 \ldots
  u_k)\da$. If $\sigma = \tau_1\arrtype\tau_2$ then by the coinductive
  hypothesis
  $\lift(n) u_1 \ldots u_k q \succeq_{\tau_2} \lift(m) u_1 \ldots u_k
  q$ for any $q \in \Iterms^f_{\tau_1}$, so
  $\lift(n) u_1 \ldots u_k \succeq_{\sigma} \lift(m) u_1 \ldots u_k$
  by definition. If $\sigma = \forall(\alpha:\kappa)\tau$ then by the
  coinductive hypothesis
  $\lift(n) u_1 \ldots u_k \xi \succeq_{\sigma'} \lift(m) u_1 \ldots
  u_k \xi$ for any closed $\xi \in \Tc_\kappa$, where
  $\sigma' = \tau[\subst{\alpha}{\xi}]$. Hence
  $\lift(n) u_1 \ldots u_k \succeq_{\sigma} \lift(m) u_1 \ldots u_k$
  by definition.
\end{proof}

\begin{lemma}\label{lem_flatten_succ}
  For $R \in \{\succeq,\succ\}$: if $t\:R_\sigma\:s$ then
  $\flatten_\sigma t\:R_\nat\: \flatten_\sigma s$ for all types
  $\sigma$.
\end{lemma}

\begin{proof}
  Without loss of generality we may assume~$\sigma$ is closed and in
  $\beta$-normal form. Using Lemma~\ref{lem_succ_red}, the lemma
  follows by induction on~$\sigma$.
\end{proof}

\begin{lemma}\label{lem_abs_succ}
  For $R \in \{\succeq,\succ\}$: if $t\:R\:s$ then
  $\abs{x}{t}\:R\:\abs{x}{s}$ and
  $\tabs{\alpha}{t}\:R\:\tabs{\alpha}{s}$.
\end{lemma}

\begin{proof}
  Assume $t \succeq_\tau s$ and $x : \sigma$. Let~$\cl$
  be a closure. We need to show
  $\cl(\abs{x}{t}) \succeq_{\cl(\sigma\arrtype\tau)}
  \cl(\abs{x}{s})$. Let $u \in \Iterms^f_{\cl(\sigma)}$. Then
  $\cl' = \cl[\subst{x}{u}]$ is a closure and
  $\cl'(t) \succeq_{\cl(\tau)} \cl'(s)$. Hence
  $\cl(t)[\subst{x}{u}] \succeq_{\cl(\tau)} \cl(s)[\subst{x}{u}]$. By
  Lemma~\ref{lem_succ_red} this implies
  $\cl(\abs{x}{t}) u \succeq_{\cl(\tau)} \cl(\abs{x}{s}) u$. Therefore
  $\cl(\abs{x}{t}) \succeq_{\cl(\sigma\arrtype\tau)}
  \cl(\abs{x}{s})$. The proof for $\succ$ is analogous.
\end{proof}

\begin{lemma}\label{lem:plustimesmonotonic}
  Let $s,t,u$ be terms of type $\sigma$.
  \begin{enumerate}
  \item If $s \succeq t$ then $s \oplus_\sigma u \succeq t
    \oplus_\sigma u$, $u \oplus_\sigma s \succeq u \oplus_\sigma t$,
    $s \otimes_\sigma u \succeq t \otimes_\sigma u$, and $u
    \otimes_\sigma s \succeq u \otimes_\sigma t$.
  \item If $s \succ t$ then $s \oplus_\sigma u \succ t \oplus_\sigma
    u$ and $u \oplus_\sigma s \succ u \oplus_\sigma t$. Moreover, if
    additionally $u \succeq \lift_\sigma(1)$ then also $s
    \otimes_\sigma u \succ t \otimes_\sigma u$ and $u \otimes_\sigma s
    \succ u \otimes_\sigma t$.
  \end{enumerate}
\end{lemma}

\begin{proof}
  It suffices to prove this for closed $s,t,u$ and closed $\sigma$ in
  $\beta$-normal form. The proof is similar to the proof of
  Lemma~\ref{lem:liftgreater}. For instance, we show by coinduction
  that for closed $w_1,\ldots,w_n$ (denoted $\vec{w}$): if
  $s \vec{w} \succ t \vec{w}$ and $u \vec{w} \succeq \lift(1) \vec{w}$
  then $(s \otimes u) \vec{w} \succ (t \otimes u) \vec{w}$.
\end{proof}

The following lemma depends on the lemmas above. The full proof may be
found in
\onlypaper{\cite[Appendix~A.2]{versionwithappendix}}%
\onlyarxiv{Appendix~\ref{sec_weakly_monotone_proof}}.
The proof is actually quite
complex, and uses a method similar to Girard's method of candidates
for the termination proof.

\begin{lemma}[Weak monotonicity]\label{lem_succeq_subst}
  If $s \succeq s'$ then $t[\subst{x}{s}] \succeq t[\subst{x}{s'}]$.
\end{lemma}

\begin{corollary}\label{cor_app_wm}
  If $s \succeq s'$ then $t s \succeq t s'$.
\end{corollary}

\section{A reduction pair for PFS terms}\label{sec_reduction_pairs}

Recall that our goal is to prove termination of reduction in
a PFS.  To
do so, in this section we will define a systematic way to generate
\emph{reduction pairs}. We fix a~PFS~$A$, and define:

\begin{defn}
  A binary relation~$R$ on $A$-terms is \emph{monotonic} if $R(s, t)$
  implies $R(C[s], C[t])$ for every context~$C$ (we assume $s,t$ have
  the same type~$\sigma$).

  A \emph{reduction pair} is a pair~$(\succeq^A,\succ^A)$ of a
  quasi-order~$\succeq^A$ on $A$-terms and a well-founded
  ordering~$\succ^A$ on $A$-terms such that:
  %\begin{itemize}
  %\item
  (a)
  $\succeq^A$ and~$\succ^A$ are compatible, i.e., ${\succ^A}
    \cdot {\succeq^A} \subseteq {\succ^A}$ and ${\succeq^A} \cdot
          {\succ^A} \subseteq {\succ^A}$,
  and (b)
  %\item
  $\succeq^A$ and~$\succ^A$ are both monotonic.
  %\end{itemize}
\end{defn}

If we can generate such a pair with $\ell \succ^A r$ for each
rule $(\ell,r) \in \Rules$, then we easily see that the PFS $A$ is
terminating.  (If we merely have $\ell \succ^A r$ for \emph{some}
rules and $\ell \succeq^A r$ for the rest, we can still progress
with the termination proof, as we will discuss in Section
\ref{sec_rule_removal}.)
To generate this pair, we will define the notion of an
\emph{interpretation} from the set of $A$-terms to the set $\Iterms$ of
interpretation terms, and thus lift the ordering pair $(\succeq,\succ)$
to $A$.
In the next section, we will show how this reduction pair can be used
in practice to prove termination of PFSs.

%\medskip
One of the core ingredients of our interpretation function is a
mapping to translate types:

\begin{defn}
  A \emph{type constructor mapping} is a function $\Typemap$ which
  maps each type constructor symbol to a closed interpretation type
  constructor of the same kind. A fixed type constructor mapping
  $\Typemap$ is extended inductively to a function from type
  constructors to closed interpretation type constructors in the
  expected way. We denote the extended \emph{interpretation (type)
    mapping} by~$\typeinterpret{\sigma}$. Thus,
  e.g.~$\typeinterpret{\quant{\alpha}{\sigma}} =
  \quant{\alpha}{\typeinterpret{\sigma}}$ and $\typeinterpret{\sigma
    \arrtype \tau} = \typeinterpret{\sigma} \arrtype
  \typeinterpret{\tau}$.
\end{defn}

\begin{lemma}\label{lem:substitutioninterpret:types}
  $\typeinterpret{\sigma}[\alpha:=\typeinterpret{\tau}] =
  \typeinterpret{\sigma[\alpha:=\tau]}$
\end{lemma}

\begin{proof}
  Induction on~$\sigma$.
\end{proof}

Similarly, we employ a \emph{symbol mapping} as the key
ingredient to interpret PFS terms.

\begin{defn}
  Given a fixed type constructor mapping~$\Typemap$, a \emph{symbol
    mapping} is a function $\Termmap$ which assigns to each function
  symbol $\mathtt{f} : \rho$ a closed interpretation term
  $\Termmap(\mathtt{f})$ of type~$\typeinterpret{\rho}$. For a fixed
  symbol mapping $\Termmap$, we define the \emph{interpretation
    mapping} $\interpret{s}$ inductively:
  \[
    \begin{array}{rclcrclcrcl}
      \interpret{x} & = & x &\quad&
      \interpret{\tabs{\alpha}{s}} & = & \tabs{\alpha}{\interpret{s}} &\quad&
      \interpret{\app{t_1}{t_2}} &=& \app{\interpret{t_1}}{\interpret{t_2}} \\
      \interpret{\mathtt{f}} &=& \Termmap(\mathtt{f}) & \quad &
      \interpret{\abs{x:\sigma}{s}} & = & \abs{x:\typeinterpret{\sigma}}{
                                          \interpret{s}} & \quad &
      \interpret{\tapp{t}{\tau}} &=& \tapp{\interpret{t}}{\typeinterpret{\tau}} \\
    \end{array}
  \]
\end{defn}

Note that $\typeinterpret{\sigma},\typeinterpret{\tau}$ above depend
on~$\Typemap$. Essentially, $\interpret{\cdot}$ substitutes
$\Typemap(\mathtt{c})$ for type constructor symbols $\mathtt{c}$, and
$\Termmap(\mathtt{f})$ for function symbols $\mathtt{f}$, thus mapping
$A$-terms to interpretation terms.  This translation preserves typing:

\begin{lemma}
  If $s : \sigma$ then $\interpret{s} : \typeinterpret{\sigma}$.
\end{lemma}

\begin{proof}
  By induction on the form of $s$, using
  Lemma~\ref{lem:substitutioninterpret:types}.
\end{proof}

\begin{lemma}\label{lem:substitutioninterpret}
  For all $s,t,x,\alpha,\tau$:
  %\begin{enumerate}
  %\item\label{lem:substitutioninterpret:mixed}
    $\interpret{s}[\alpha:=\typeinterpret{\tau}] =
    \interpret{s[\alpha:=\tau]}$
  and
  %\item\label{lem:substitutioninterpret:terms}
    $\interpret{s}[x:=\interpret{t}] = \interpret{s[x:=t]}$.
  %\end{enumerate}
\end{lemma}

\begin{proof}
  Induction on~$s$.
\end{proof}

\begin{defn}
  For a fixed type constructor mapping $\Typemap$ and symbol mapping
  $\Termmap$, the \emph{interpretation pair}
  $(\succeqinterpret,\succinterpret)$ is defined as follows: $s
  \succeqinterpret t$ if $\interpret{s} \succeq \interpret{t}$, and $s
  \succinterpret t$ if $\interpret{s} \succ \interpret{t}$.
\end{defn}

\begin{remark}
The polymorphic lambda-calculus has a much greater expressive power
than the simply-typed lambda-calculus. Inductive data types
may be encoded, along with their constructors and recursors with
appropriate derived reduction rules. This makes
our
interpretation method easier to apply, even in the non-polymorphic
setting, thanks to more
sophisticated ``programming'' in the interpretations.
The reader is advised to consult e.g.~\cite[Chapter~11]{Girard1989}
for more background and explanations. We demonstrate the idea
by presenting an encoding for the recursive type $\List$ and its
fold-left function (see also
Ex.~\ref{ex_fold_interpretation}).
\end{remark}

\begin{example}\label{ex:notyetmono}
Towards a termination proof of Example~\ref{ex_fold_pafs},
we set $\Typemap(\List) = \forall \beta. (\forall \alpha.
\beta \arrtype \alpha \arrtype \beta) \arrtype \beta
\arrtype \beta$ and
$\Termmap(\nil) = \tabs{\beta}{\abs{f:\quant{\alpha}{\beta \arrtype
\alpha \arrtype \beta}}{\abs{x:\beta}{x}}}$.  If we additionally choose
$\Termmap(\mathtt{foldl}) = \tabs{\beta}{\abs{f}{\abs{x}{\abs{l}{l
\beta f x}}} \oplus \lift_\beta(1)}$, we have
$\interpret{\mathtt{foldl}_{\sigma}(f,s,\nil)} = (\tabs{\beta}{
  \abs{f}{\abs{x}{\abs{l}{l \beta f x}}} \oplus \lift_\beta(1)})
  \typeinterpret{\sigma} f s (\tabs{\beta}{\abs{f}{\abs{x}{x}}})
  \leadsto^* s \oplus \lift_{\interpret{\sigma}}(1)$ by
  $\beta$-reduction steps.
An extension of the proof from Example~\ref{ex:plus1} shows that
this term $\succ \interpret{s}$.
\end{example}

It is easy to see that $\succeqinterpret$ and
$\succinterpret$ have desirable properties such as transitivity,
reflexivity (for $\succeqinterpret$) and well-foundedness (for
$\succinterpret$).  However, $\succinterpret$ is not necessarily
monotonic.  Using the interpretation from Example~\ref{ex:notyetmono},
$\interpret{\mathtt{foldl}_{\sigma}(\abs{x}{s},t,\nil)} =
\interpret{\mathtt{fold}_{\sigma}(\abs{x}{w},t,\nil)}$ regardless of
$s$ and $w$, so a reduction in $s$ would not cause a decrease in
$\succinterpret$.  To obtain a reduction pair, we must impose certain
conditions on $\Termmap$; in particular, we will require that
$\Termmap$ is \emph{safe}.

\begin{defn}\label{def_safe}
  If $s_1 \succ s_2$ implies $t[\subst{x}{s_1}] \succ
  t[\subst{x}{s_2}]$, then the interpretation term~$t$ is \emph{safe
    for~$x$}. A symbol mapping~$\Termmap$ is \emph{safe} if for all
  $
%  \[
  \mathtt{f} : \forall (\alpha_1 : \kappa_1) \ldots \forall (\alpha_n
  : \kappa_n) . \sigma_1 \arrtype \ldots \arrtype \sigma_k \arrtype
  \tau
%  \]
$
  with~$\tau$ a type atom we have: $\Termmap(\mathtt{f}) =
  \tabs{\alpha_1 \dots \alpha_n}{\abs{x_1 \dots x_k}{t}}$ with $t$
  safe for each~$x_i$.
\end{defn}

\begin{lemma}\label{lem_safe}
  \begin{enumerate}
  \item $x u_1 \ldots u_m$ is safe for~$x$.
  \item If $t$ is safe for~$x$ then so are~$\lift(t)$
    and~$\flatten(t)$.
  \item If $s_1$ is safe for~$x$ or $s_2$ is safe for~$x$ then $s_1
    \oplus s_2$ is safe for~$x$.
  \item If either
    %\begin{itemize}
    %\item
      (a)
      $s_1$ is safe for~$x$ and $s_2 \succeq \lift(1)$, or
    %\item
      (b)
      $s_2$ is safe for~$x$ and $s_1 \succeq \lift(1)$,
    %\end{itemize}
    then $s_1 \otimes s_2$ is safe for~$x$.
  \item If~$t$ is safe for~$x$ then so is~$\tabs{\alpha}{t}$
    and~$\abs{y}{t}$ ($y \ne x$).
%  \item If $t$ is safe for~$x$ then so is~$\pi^1(t)$ and~$\pi^2(t)$.
%  \item If $t$ is safe for~$x$ then so is~$\xlet{}{t}{\alpha,y}{s}$.
  \end{enumerate}
\end{lemma}

\begin{proof}
  Each point follows from one of the lemmas proven before,
  Lemma~\ref{lem_succ_to_succeq}, Lemma~\ref{lem_succeq_subst},
  Lemma~\ref{lem:compatibility} and the transitivity of~$\succeq$. For
  instance, for the first, assume $s_1 \succ s_2$ and let
  $u_i^j=u_i[\subst{x}{s_j}]$. Then $(x u_1 \ldots
  u_m)[\subst{x}{s_1}] = s_1 u_1^1 \ldots u_m^1$. By
  Lemma~\ref{lem_app_succ} we have $s_1 u_1^1 \ldots u_m^1 \succ s_2
  u_1^1 \ldots u_m^1$. By Lemma~\ref{lem_succ_to_succeq} and
  Lemma~\ref{lem_succeq_subst} we have $u_i^1 \succeq u_i^2$. By
  Corollary~\ref{cor_app_wm} and the transitivity of~$\succeq$ we
  obtain $s_2 u_1^1 \ldots u_m^1 \succeq s_2 u_1^2 \ldots u_m^2$. By
  Lemma~\ref{lem:compatibility} finally $(x u_1 \ldots
  u_m)[\subst{x}{s_1}] = s_1 u_1^1 \ldots u_m^1 \succ s_2 u_1^2 \ldots
  u_m^2 = (x u_1 \ldots u_m)[\subst{x}{s_2}]$.
\end{proof}

\begin{lemma}\label{lem_succinterpret_monotonic}
  If~$\Termmap$ is safe then~$\succinterpret$ is monotonic.
\end{lemma}

\begin{proof}
  Assume $s_1 \succinterpret s_2$. By induction on
  a context~$C$ we show $C[s_1] \succinterpret C[s_2]$. If $C=\Box$ then
  this is obvious. If $C = \abs{x}{C'}$ or $C = \tabs{\alpha}{C'}$
  then $C'[s_1] \succinterpret C'[s_2]$ by the inductive hypothesis,
  and thus $C[s_1] \succinterpret C[s_2]$ follows from
  Lemma~\ref{lem_abs_succ} and definitions. If $C = C' t$ then
  $C'[s_1] \succinterpret C'[s_2]$ by the inductive hypothesis,
  so $C[s_1] \succinterpret C[s_2]$ follows from definitions.

  Finally, assume $C = \app{t}{C'}$. Then $t = \mathtt{f} \rho_1
  \ldots \rho_n t_1 \ldots t_m$ where
  $
  \mathtt{f} : \forall (\alpha_1 : \kappa_1) \ldots \forall (\alpha_n
  : \kappa_n) . \sigma_1 \arrtype \ldots \arrtype \sigma_k \arrtype
  \tau
  $
  with~$\tau$ a type atom, $m < k$, and $\Termmap(\mathtt{f}) =
  \tabs{\alpha_1 \dots \alpha_n}{\abs{x_1 \dots x_k}{u}}$ with $u$
  safe for each~$x_i$. Without loss of generality assume $m=k-1$. Then
  $\interpret{C[s_i]} \leadsto u'[\subst{x_k}{\interpret{C'[s_i]}}]$
  where
  $u'=u[\subst{\alpha_1}{\typeinterpret{\rho_1}}]\ldots[\subst{\alpha_n}{\typeinterpret{\rho_n}}][\subst{x_1}{\interpret{t_1}}]\ldots[\subst{x_{k-1}}{\interpret{t_{k-1}}}]$. By
  the inductive hypothesis $\interpret{C'[s_1]} \succ
  \interpret{C'[s_2]}$. Hence $u'[\subst{x_k}{\interpret{C'[s_1]}}]
  \succ u'[\subst{x_k}{\interpret{C'[s_2]}}]$, because~$u$ is safe
  for~$x_k$. Thus $\interpret{C[s_1]} \succ \interpret{C[s_2]}$ by
  Lemma~\ref{lem_succ_red}.
\end{proof}

\begin{theorem}\label{thm_reduction_pair}
  If~$\Termmap$ is safe then $(\succeqinterpret,\succinterpret)$ is a
  reduction pair.
\end{theorem}

\begin{proof}
  By Lemmas~\ref{lem_transitive} and~\ref{lem_reflexive},
  $\succeqinterpret$ is a
  quasi-order. Lemmas~\ref{lem_well_founded}
  and~\ref{lem_transitive} imply that~$\succinterpret$ is a
  well-founded ordering. Compatibility follows from
  Lemma~\ref{lem:compatibility}. Monotonicity of~$\succeqinterpret$
  follows from Lemma~\ref{lem_succeq_subst}. Monotonicity
  of~$\succinterpret$ follows from
  Lemma~\ref{lem_succinterpret_monotonic}.
\end{proof}

\begin{example}\label{ex_fold_interpretation}
  The following is a safe interpretation for the PFS from
  Example~\ref{ex_fold_pafs}:
  \[
  \begin{array}{rcll}
    \Typemap(\List) & = & \multicolumn{2}{l}{
      \quant{\beta}{(\quant{\alpha}{\beta \arrtype
      \alpha \arrtype \beta}) \arrtype \beta \arrtype \beta}}\\
  \Termmap(\mathtt{@}) & = & \Lambda \alpha.\Lambda\beta.
    \lambda f.\lambda x. &
    f \cdot x \oplus \lift_\beta(\flatten_\alpha(x)) \\
  \Termmap(\mathtt{A}) & = & \Lambda \alpha.\Lambda \beta.\lambda x. &
    x * \beta \\
  \Termmap(\mathtt{nil}) & = & & \Lambda \beta.\lambda f.\abs{x}{x} \\
  \Termmap(\mathtt{cons}) & = & \Lambda \alpha.\lambda h.\lambda t. &
    \Lambda \beta.\lambda f.\lambda x.
    t \beta f (f \alpha x h \oplus \lift_\beta(\flatten_\beta(x)\
    \oplus \\
    & & & \phantom{ABCDEFGHIJKLMNOP,} \flatten_\alpha(h)))\ \oplus\  \\
    & & & \phantom{ABCDE\ }
    \lift_\beta(\flatten_\beta(f\alpha x h) \oplus
    \flatten_\alpha(h) \oplus 1) \\
  \Termmap(\mathtt{foldl}) & = & \Lambda \beta.\lambda f. \lambda x.
    \lambda l. & l \beta f x \oplus \lift_\beta(\flatten_{\forall \alpha.
    \beta \arrtype \alpha \arrtype \beta}(f)\ \oplus \\
    & & & \phantom{ABCDEFG\ \ }
    \flatten_\beta(x) \oplus 1) \\
  \end{array}
  \]
Note that $\Termmap(\mathtt{cons})$ is \emph{not} required to be safe
for $x$, since $x$ is not an argument of $\mathtt{cons}$:
following its declaration, $\mathtt{cons}$ takes one type and two terms
as arguments. The variable $x$ is only part of the \emph{interpretation}.
Note also that the current interpretation is a mostly straightforward
extension of Example~\ref{ex:notyetmono}: we retain the same \emph{core}
interpretations (which, intuitively, encode $\mathtt{@}$ and
$\mathtt{A}$ as forms of application and encode a list as the function
that executes a fold over the list's contents), but we add a clause
$\oplus \lift(\flatten(x))$ for each argument $x$ that the initial
interpretation is not safe for.  The only further change is that, in
$\Termmap(\mathtt{cons})$, the part between brackets has to be extended.
This was necessitated by the change to $\Termmap(\mathtt{foldl})$, in
order for the rules to still be oriented (as we will do in
Example \ref{ex_fold_final}).
\end{example}

\section{Proving termination with rule removal}\label{sec_rule_removal}

A PFS $A$ is certainly terminating if its reduction relation
$\arr{\Rules}$ is contained in a well-founded relation, which holds if
$\ell \succinterpret r$ for all its rules $(\ell,r)$.  However,
sometimes it is cumbersome to find an interpretation that orients all
rules strictly. To illustrate, the interpretation of Example
\ref{ex_fold_interpretation} gives $\ell \succinterpret r$ for two of
the rules and $\ell \succeqinterpret r$ for the others (as we will see
in Example \ref{ex_fold_final}). In such cases, proof progress is still
achieved through \emph{rule removal}.

\begin{theorem}\label{thm:ruleremove}
  Let $\Rules = \Rules_1 \cup \Rules_2$, and suppose that
  $\Rules_1\subseteq{\succ^\Rules}$ and
  $\Rules_2\subseteq{\succeq^\Rules}$ for a reduction pair
  $(\succeq^\Rules,\succ^\Rules)$. Then $\arr{\Rules}$ is terminating
  if and only if $\arr{\Rules_2}$ is
  (so certainly if $\Rules_2 = \emptyset$).
  %In particular, $\arr{\Rules}$ is
  %terminating if $\Rules_2 = \emptyset$.
\end{theorem}

\begin{proof}
  Monotonicity of~$\succeq^\Rules$ and~$\succ^\Rules$ implies that
  ${\arr{\Rules_1}}\subseteq{\succ^\Rules}$ and
  ${\arr{\Rules_2}}\subseteq{\succeq^\Rules}$.

  By well-foundedness of $\succ^\Rules$, compatibility
  of~$\succeq^\Rules$ and~$\succ^\Rules$, and transitivity
  of~$\succeq^\Rules$, every infinite $\arr{\Rules}$ sequence can
  contain only finitely many $\arr{\Rules_1}$ steps.
\end{proof}

The above theorem gives rise to the following \emph{rule removal}
algorithm:
\begin{enumerate}
\item While $\Rules$ is non-empty:
  \begin{enumerate}
  \item Construct a reduction pair $(\succeq^\Rules,\succ^\Rules)$
    such that all rules in $\Rules$ are oriented by $\succeq^\Rules$ or
    $\succ^\Rules$, and at least one of them is oriented using
    $\succ^\Rules$.
  \item Remove all rules ordered by $\succ^\Rules$ from $\Rules$.
  \end{enumerate}
\end{enumerate}
If this algorithm succeeds, we have proven termination.

\medskip
To use this algorithm with the pair $(\succeqinterpret,\succinterpret)$
from Section~\ref{sec_reduction_pairs}, we should identify an
interpretation $(\Typemap,\Termmap)$
such that (a) $\Termmap$ is safe, (b) all rules can be oriented with
$\succeqinterpret$ or $\succinterpret$, and (c) at least one rule is
oriented with $\succinterpret$.
The first requirement guarantees that
$(\succeqinterpret,\succinterpret)$ is a reduction pair (by
Theorem~\ref{thm_reduction_pair}). Lemma~\ref{lem_safe} provides some
sufficient safety criteria. The second and third
requirements have to be verified for each individual rule.

\begin{example}\label{ex_fold_intermediate}
  We continue with our example of fold on heterogeneous lists. We prove
  termination by rule removal, using the symbol mapping from
  Example~\ref{ex_fold_interpretation}.
  We will show:
  \[
  \begin{array}{rcl}
    @_{\sigma,\tau}(\abs{x:\sigma}{s},t) & \succeqinterpret & s[x:=t] \\
    \mathtt{A}_{\abs{\alpha}{\sigma},\tau}(\tabs{\alpha}{s}) &
      \succeqinterpret & s[\alpha:=\tau] \\
    \mathtt{foldl}_\sigma(f,s,\nil) & \succinterpret & s \\
    \mathtt{foldl}_\sigma(f,s,\cons_\tau(h,t)) & \succinterpret &
    \mathtt{foldl}_\sigma(f,@_{\tau,\sigma}(@_{\sigma,\tau
    \arrtype\sigma}(
      \mathtt{A}_{\abs{\alpha}{\sigma\arrtype\alpha\arrtype\sigma},
      \tau}(f),s),h),t) \\
  \end{array}
  \]
Consider the first inequality; by definition it holds if
$\interpret{@_{\sigma,\tau}(\abs{x:\sigma}{s},t)} \succeq
\interpret{s[x:=t]}$.
Since $\interpret{@_{\sigma,\tau}(\abs{x:
\sigma}{s},t)} \arrW^* \interpret{s}[x:=\interpret{t}] \oplus
\lift_{\typeinterpret{\tau}}(\flatten_{\typeinterpret{\sigma}}(
\interpret{t}))$, and $\interpret{s}[x:=\interpret{t}] =
\interpret{s[x:=t]}$ (by Lemma~\ref{lem:substitutioninterpret}),
it suffices by Lemma~\ref{lem_leadsto_to_approx} if
$\interpret{s[x:=t]} \oplus \lift_{\typeinterpret{\tau}}(\flatten_{\typeinterpret{\sigma}}(
\interpret{t})) \succeq \interpret{s[x:=t]}$.
This is an instance of the general rule $u \oplus w \succeq u$ that
we will obtain below.
\end{example}

To prove inequalities $s \succ t$ and $s \succeq t$, we will often
use that $\succ$ and $\succeq$ are transitive and compatible with each
other (Lem.~\ref{lem_transitive} and~\ref{lem:compatibility}), that
$\arrW\:\subseteq\:\approx$ (Lem.~\ref{lem_leadsto_to_approx}),
that $\succeq$ is monotonic (Lem.~\ref{lem_succeq_subst}),
that both $\succ$ and $\succeq$ are monotonic over $\lift$ and $\flatten$
(Lem.~\ref{lem:liftgreater} and \ref{lem_flatten_succ}) and that
interpretations respect substitution
(Lem.~\ref{lem:substitutioninterpret}). We will also use
Lemma \ref{lem:plustimesmonotonic} which states (among other things)
that $s \succ t$ implies $s \oplus u \succ t \oplus u$.
In addition, we can use the
calculation rules below. The proofs may be found in
\onlypaper{\cite[Appendix~A.3]{versionwithappendix}}%
\onlyarxiv{Appendix~\ref{app_rule_removal}}.

\begin{lemma}\label{lem:approxproperties}
For all types $\sigma$ and all terms $s,t,u$ of type $\sigma$, we
have:
\begin{enumerate}
\item\label{lem:approx:symmetry} $s \oplus_\sigma t \approx t
  \oplus_\sigma s$ and $s \otimes_\sigma t \approx t \otimes_\sigma
  s$;
\item\label{lem:approx:assoc} $s \oplus_\sigma (t \oplus_\sigma u)
  \approx (s \oplus_\sigma t) \oplus_\sigma u$ and $s \otimes_\sigma
  (t \otimes_\sigma u) \approx (s \otimes_\sigma t) \otimes_\sigma u$;
\item\label{lem:approx:distribution} $s \otimes_\sigma (t
  \oplus_\sigma u) \approx (s \otimes_\sigma t) \oplus_\sigma (s
  \otimes_\sigma u)$;
\item\label{lem:approx:neutral} $(\lift_\sigma 0) \oplus_\sigma s
  \approx s$ and $(\lift_\sigma 1) \otimes_\sigma s \approx s$.
\end{enumerate}
\end{lemma}

\begin{lemma}\label{lem_lift_approx}
  \begin{enumerate}
  \item\label{lem_lift_approx:plussplit}
    $\lift_\sigma(n+m) \approx_\sigma (\lift_\sigma n)
    \oplus_\sigma (\lift_\sigma m)$;
  \item $\lift_\sigma(n m) \approx_\sigma (\lift_\sigma n)
    \otimes_\sigma (\lift_\sigma m)$;
  \item $\flatten_\sigma(\lift_\sigma(n)) \approx n$.
  \end{enumerate}
\end{lemma}

\begin{lemma}\label{lem:plusparts}
For all types $\sigma$, terms $s,t$ of type $\sigma$ and natural
numbers $n > 0$:
\begin{enumerate}
\item\label{lem:plusparts:removefromsucceq}
  $s \oplus_{\sigma} t \succeq s$ and $s \oplus_{\sigma} t \succeq
  t$;
\item $s \oplus_{\sigma} (\lift_{\sigma} n) \succ s$ and
  $(\lift_{\sigma} n) \oplus_{\sigma} t \succ t$.
\end{enumerate}
\end{lemma}

Note that these calculation rules immediately give the
inequality $x \oplus \lift_{nat \arrtype \nat}(1) \succ x$ from
Example~\ref{ex:plus1}, and also that $\lift_\sigma(n) \succ
\lift_\sigma(m)$ whenever $n > m$.  By
Lemmas~\ref{lem:plustimesmonotonic} and~\ref{lem:plusparts} we
can use \emph{absolute positiveness}: the property that
(a) $s \succeq t$ if we can write $s \approx s_1 \oplus \dots \oplus
s_n$ and $t \approx t_1 \oplus \dots \oplus t_k$ with $k \leq n$ and
$s_i \succeq t_i$ for all $i \leq k$, and (b) if moreover $s_1 \succ
t_1$ then $s \succ t$. This property is typically very
useful to dispense the obligations obtained in a termination proof with
polynomial interpretations.

\begin{example}\label{ex_fold_final}
We now have the tools to finish the example of
heterogeneous lists (still using the interpretation from
Example~\ref{ex_fold_interpretation}).  The proof obligation from Example
\ref{ex_fold_intermediate}, that
$\interpret{@_{\sigma,\tau}(\abs{x:\sigma}{s},t)} \succeq
\interpret{s[x:=t]}$, is completed by
Lemma \ref{lem:plusparts}(\ref{lem:plusparts:removefromsucceq}).
We have $\interpret{\mathtt{A}_{\abs{\alpha}{\sigma},
\tau}(\tabs{\alpha}{s})} \approx \interpret{\tabs{\alpha}{s}} *
\typeinterpret{\tau} \approx \interpret{s[\alpha:=\tau]}$ by Lemma
\ref{lem:substitutioninterpret}, and
$\interpret{\mathtt{foldl}_\sigma(f,s,\nil)} =
\interpret{\nil}*\typeinterpret{\sigma} \cdot \interpret{f} \cdot
\interpret{s} \oplus \lift_{\typeinterpret{\sigma}}(\langle
\text{something}\rangle\oplus 1) \approx \interpret{s} \oplus
\lift_{\typeinterpret{\sigma}}(\langle\text{something}\rangle\oplus 1)
\succ \interpret{s}$ by Lemmas
\ref{lem_lift_approx}(\ref{lem_lift_approx:plussplit}) and
\ref{lem:plusparts}(\ref{lem:plusparts:removefromsucceq}).
For the last rule note that (using only Lemmas
\ref{lem_leadsto_to_approx} and
\ref{lem_lift_approx}(\ref{lem_lift_approx:plussplit})):
\[
\begin{array}{l}
\interpret{\mathtt{foldl}_\sigma(f,s,\cons_\tau(h,t))} \approx \\
\interpret{\cons_\tau(h,t))} * \typeinterpret{\sigma} \cdot \interpret{f}
\cdot \interpret{s} \oplus \lift_{\typeinterpret{\sigma}}(
\flatten(\interpret{f}) \oplus \flatten(\interpret{s}) \oplus 1) \approx \\
(\ \interpret{t} * \typeinterpret{\sigma} \cdot \interpret{f} \cdot
(\interpret{f} * \typeinterpret{\tau} \cdot \interpret{s}
\cdot \interpret{h} \oplus
\lift_{\typeinterpret{\sigma}}(\flatten(\interpret{s}) \oplus
\flatten(\interpret{h})))\ \oplus \\
\phantom{AB}
\lift_{\typeinterpret{\sigma}}(\flatten(\interpret{f} *
\typeinterpret{\tau} \cdot \interpret{s} \cdot \interpret{h}) \oplus
\flatten(\interpret{h}) \oplus 1)\ )\ \oplus \\
\phantom{A}
 \lift_{\typeinterpret{\sigma}}(\flatten(\interpret{f}) \oplus
 \flatten(\interpret{s}) \oplus 1) \approx \\
\interpret{t} * \typeinterpret{\sigma} \cdot \interpret{f} \cdot
(\ \interpret{f} * \typeinterpret{\tau} \cdot \interpret{s}
\cdot \interpret{h} \oplus
\lift_{\typeinterpret{\sigma}}(\flatten(\interpret{s}) \oplus
\flatten(\interpret{h}))\ )\ \oplus \\
\phantom{A}
\lift_{\typeinterpret{\sigma}}(\flatten(\interpret{f} * \typeinterpret{
\tau} \cdot \interpret{s} \cdot \interpret{h}) \oplus
\flatten(\interpret{h}) \oplus
\flatten(\interpret{f}) \oplus\flatten(\interpret{s}) \oplus 2) \\
\end{array}
\]
On the right-hand side of the inequality, noting that
$\lift_{\sigma \arrtype \tau}(u) \cdot w \arrW^*
\lift_{\tau}(u)$, we have:
\[
\begin{array}{l}
\interpret{\mathtt{foldl}_\sigma(f,@_{\tau,\sigma}(@_{\sigma,\tau
    \arrtype\sigma}(
      \mathtt{A}_{\abs{\alpha}{\sigma\arrtype\alpha\arrtype\sigma},
      \tau}(f),s),h),t)} \approx \\
\Termmap(\mathtt{foldl})_\sigma(\interpret{f},\
  \interpret{f} * \typeinterpret{\tau} \cdot \interpret{s} \cdot
  \interpret{h} \oplus \lift_{\typeinterpret{\sigma}}(\flatten(\interpret{s})
  \oplus \flatten(\interpret{h})),\ \interpret{t}) \approx \\
\interpret{t} * \typeinterpret{\sigma} \cdot \interpret{f} \cdot
  (\ \interpret{f} * \typeinterpret{\tau} \cdot \interpret{s} \cdot
  \interpret{h} \oplus \lift_{\typeinterpret{\sigma}}(\flatten(\interpret{s})
  \oplus \flatten(\interpret{h}))\ )\ \oplus \\
\phantom{A}
  \lift_{\typeinterpret{\sigma}}(\flatten(\interpret{f}) \oplus
  \flatten(\interpret{f} * \typeinterpret{\tau} \cdot \interpret{s} \cdot
  \interpret{h}\ \oplus \\
  \phantom{ABCDEFGHIJKLMNOPQRSt}
  \lift_{\typeinterpret{\sigma}}(\flatten(\interpret{s})
  \oplus \flatten(\interpret{h}))) \oplus 1) \approx \\
\interpret{t} * \typeinterpret{\sigma} \cdot \interpret{f} \cdot
  (\ \interpret{f} * \typeinterpret{\tau} \cdot \interpret{s} \cdot
  \interpret{h} \oplus \lift_{\typeinterpret{\sigma}}(\flatten(\interpret{s})
  \oplus \flatten(\interpret{h}))\ )\ \oplus \\
\phantom{A}
  \lift_{\typeinterpret{\sigma}}(\flatten(\interpret{f}) \oplus
  \flatten(\interpret{f} * \typeinterpret{\tau} \cdot \interpret{s} \cdot
  \interpret{h}) \oplus \flatten(\interpret{s}) \oplus \flatten(\interpret{h}) \oplus 1)
\end{array}
\]
Now the right-hand side is the left-hand side $\oplus\ \lift(1)$.
Clearly, the rule is oriented with $\succ$.  Thus, we may remove the
last two rules, and continue the rule removal algorithm with only the
first two, which together define $\beta$-reduction.  This is trivial,
for instance with an interpretation
$\Termmap(@) = \Lambda \alpha.\Lambda \beta.\lambda f.\lambda x.
(f \cdot x) \oplus \lift_\beta(\flatten_\alpha(x) \oplus 1)$ and
$\Termmap(\mathtt{A}) = \Lambda \alpha.\Lambda \beta.\lambda x.
x * \beta \oplus \lift_{\alpha\beta}(1)$.
\end{example}

\section{A larger example}\label{sec:examples}

System~$\mathtt{F}$ is System~$\Fomega$ where no higher kinds are
allowed, i.e., there are no type constructors except types. By the
Curry-Howard isomorphism~$\mathtt{F}$ corresponds to the
universal-implicational fragment of intuitionistic second-order
propositional logic, with the types corresponding to formulas and
terms to natural deduction proofs. The remaining connectives may be
encoded in~$\mathtt{F}$, but the permutative conversion rules do not
hold~\cite{Girard1989}.

In this section we show
termination of the system IPC2
(see~\cite{SorensenUrzyczyn2010}) of intuitionistic second-order
propositional logic with all connectives and permutative conversions,
minus a few of the permutative conversion rules for the existential
quantifier. The paper~\cite{SorensenUrzyczyn2010} depends on
termination of~IPC2, citing a proof from~\cite{Wojdyga2008}, which,
however, later turned out to be incorrect. Termination of
Curry-style~IPC2 without~$\bot$ as primitive was shown
in~\cite{Tatsuta2007}. To our knowledge, termination of the full
system~IPC2 remains an open problem, strictly speaking.

\begin{remark}
Our method builds on the work of van de Pol and Schwichtenberg, who used
higher-order polynomial interpretations to prove termination of a
fragment of intuitionistic first-order logic with permutative
conversions~\cite{PolSchwichtenberg1995}, in the hope of providing a
more perspicuous proof of this well-known result. Notably, they did
not treat disjunction, as we will do. More fundamentally, their method
cannot handle impredicative polymorphism necessary for second-order
logic.
\end{remark}

The system IPC2 can be seen as a PFS with type constructors:
\[
\begin{array}{c}
\Sigma^T_\kappa = \{\quad
  \bot : *,\quad
  \mathtt{or} : * \arrkind * \arrkind *,\quad
  \mathtt{and} : * \arrkind * \arrkind *,\quad
  \exists : (* \arrkind *) \arrkind *
  \}
\end{array}
\]

We have the following function symbols:
\[
\begin{array}{rclcrcl}
@ & : & \forall \alpha \forall \beta . (\alpha \arrtype \beta) \arrtype \alpha \arrtype \beta &
\quad &
\epsilon & : & \forall \alpha . \bot \arrtype \alpha \\

\mathtt{tapp} & : & \forall \alpha : * \arrkind * . \forall \beta .
  (\forall \beta [\alpha \beta]) \arrtype \alpha \beta &
\quad &
\proj^1 & : & \forall \alpha \forall \beta . \mathtt{and}\, \alpha\, \beta \arrtype \alpha \\

\mathtt{pair} & : & \forall \alpha \forall \beta . \alpha \arrtype \beta \arrtype
  \mathtt{and}\, \alpha\, \beta &
\quad &
\proj^2 & : & \forall \alpha \forall \beta . \mathtt{and}\, \alpha\, \beta \arrtype \beta \\

\mathtt{case} & : & \forall \alpha \forall \beta \forall \gamma . \mathtt{or}\, \alpha\, \beta \arrtype
  (\alpha \arrtype \gamma) \arrtype (\beta \arrtype \gamma) \arrtype \gamma &
\quad &
\mathtt{in}^1 & : & \forall \alpha \forall \beta . \alpha \arrtype
  \mathtt{or}\, \alpha\, \beta \\

\mathtt{let} & : & \forall \alpha : * \arrkind * . \forall \beta .
  (\exists (\alpha)) \arrtype
  (\forall \gamma . \alpha \gamma \arrtype \beta) \arrtype \beta &
\quad &
\mathtt{in}^2 & : & \forall \alpha \forall \beta . \beta \arrtype
  \mathtt{or}\, \alpha\, \beta \\

\mathtt{ext} & : & \forall \alpha : * \arrkind * . \forall \beta . \alpha \beta \arrtype
  \exists (\alpha)
\end{array}
\]

The types represent formulas in intuitionistic second-order
propositional logic, and the terms represent proofs.  For example, a
term $\mathtt{case}_{\sigma,\tau,\rho}\ s\ u\ v$ is a proof term of
the formula $\rho$, built from a proof $s$ of
$\mathtt{or}\ \sigma\ \tau$, a proof $u$ that $\sigma$ implies $\rho$
and a proof $v$ that $\tau$ implies $\rho$.  Proof terms can be
simplified using 28 reduction rules, including the following (the full
set of rules is available in
\onlypaper{\cite[Appendix~B]{versionwithappendix}}%
\onlyarxiv{Appendix~\ref{app_ineqs}}):
\[
\begin{array}{rclrcl}
@_{\sigma,\tau}(\abs{x}{s},t) & \red & s[x:=t] &
\mathtt{case}_{\sigma,\tau,\rho}(\mathtt{in}^1_{\sigma,\tau}(u),
  \abs{x}{s},\abs{y}{t}) & \red & s[x:=u] \\

\mathtt{tapp}_{\abs{\alpha}{\sigma},\tau}(\tabs{\alpha}{s}) & \red &
  s[\alpha:=\tau] &
\mathtt{case}_{\sigma,\tau,\rho}(\mathtt{in}^2_{\sigma,\tau}(u),
  \abs{x}{s},\abs{y}{t}) & \red & t[x:=u] \\

\proj^1_{\sigma,\tau}(\mathtt{pair}_{\sigma,\tau}(s,t)) & \red & s &
\mathtt{let}_{\varphi,\rho}(\mathtt{ext}_{\varphi,\tau}(s),\tabs{\alpha}{\abs{x}{t}}) & \red & t[\alpha:=\tau][x:=s] \\

\proj^2_{\sigma,\tau}(\mathtt{pair}_{\sigma,\tau}(s,t)) & \red & t \\
\end{array}
\]
\[
\begin{array}{ll}
& @_{\sigma,\tau}(\epsilon_{\sigma \arrtype \tau}(s),t) \red
  \epsilon_\tau(s) \\
& \mathtt{case}_{\sigma,\tau,\rho}(\epsilon_{\mathtt{or}\,\sigma\,\tau}(
  u),\abs{x}{s},\abs{y}{t}) \red \epsilon_\rho(u) \\
& \epsilon_\rho(\mathtt{case}_{\sigma,\tau,\bot}(u,\abs{x}{s},
  \abs{y}{t})) \red
  \mathtt{case}_{\sigma,\tau,\rho}(u,\abs{x}{\epsilon_\rho(s)},
  \abs{y}{\epsilon_\rho(t)}) \\
& \proj^2_{\rho,\pi}(\mathtt{case}_{\sigma,\tau,\mathtt{and}\,\rho,\pi}(u,
  \abs{x:\sigma}{s},\abs{y:\tau}{t}))\red
  \mathtt{case}_{\sigma,\tau,\pi}(u,\abs{x:\sigma}{\proj^2_{\rho,\pi}(s)},
  \abs{y:\tau}{\proj^2_{\rho,\pi}(t)}) \\
& \mathtt{case}_{\rho,\pi,\xi}(\mathtt{case}_{\sigma,\tau,\mathtt{or}\,
  \rho\,\pi}(u,\abs{x}{s},\abs{y}{t}),\abs{z}{v},
  \abs{a}{w}) \red \\
& \phantom{AB}
  \mathtt{case}_{\sigma,\tau,\xi}(u,\abs{x}{\mathtt{case}_{
  \rho,\pi,\xi}(s,\abs{z}{v},\abs{a}{w})},
  \abs{y}{\mathtt{case}_{\rho,\pi,\xi}(t,\abs{z}{v},
  \abs{a}{w})}) \\
& \mathtt{let}_{\varphi,\rho}(
  \mathtt{case}_{\sigma,\tau,\exists\varphi}(
  u,\abs{x}{s},\abs{y}{t}),v) \red
  \mathtt{case}_{\sigma,\tau,\rho}(u,
  \abs{x}{\mathtt{let}_{\varphi,\rho}(s,v)},
  \abs{y}{\mathtt{let}_{\varphi,\rho}(t,v)}) \\
\hspace{-10pt}
(*) & \mathtt{let}_{\psi,\rho}(\mathtt{let}_{\varphi,\exists\psi}(s,
  \tabs{\alpha}{\abs{x:\varphi\alpha}{t}}),u) \red
  \mathtt{let}_{\varphi,\rho}(s,\tabs{\alpha}{\abs{x:
  \varphi\alpha}{\mathtt{let}_{\psi,\rho}(t,u)}}) \\
\end{array}
\]

To define an interpretation for~IPC2, we will use the
standard encoding of
product and existential types (see~\cite[Chapter~11]{Girard1989} for
more details).
  \[
  \begin{array}{rclcrcl}
    \sigma \times \tau &=& \forall p . (\sigma \arrtype \tau \arrtype p) \arrtype p & \quad &
    \pi^1_{\sigma,\tau}(t) &=& t \sigma (\abs{x:\sigma}{\abs{y:\tau}{x}}) \\
    \pair{t_1}{t_2}_{\sigma,\tau} &=& \tabs{p}{\abs{x:\sigma\arrtype\tau\arrtype p}{x t_1 t_2}} & &
    \pi^2_{\sigma,\tau}(t) &=& t \tau (\abs{x:\sigma}{\abs{y:\tau}{y}}) \\
    \Sigma \alpha . \sigma &=& \forall p . (\forall \alpha . \sigma \arrtype p) \arrtype p & \quad &
    \phantom{ABCD}
    \expair{\tau}{t}_{\Sigma\alpha.\sigma} &=& \tabs{p}{\abs{x:\forall\alpha.\sigma\arrtype p}{x \tau t}} \\
    & &
    \multicolumn{3}{r}{
    \xlet{\rho}{t}{\alpha,x:\sigma}{s}} &=& t \rho (\tabs{\alpha}{\abs{x:\sigma}{s}}) \\
  \end{array}
  \]

We do not currently have an algorithmic method to find a
suitable interpretation.  Instead, we used the following manual process.
We start by noting the minimal requirements given by the first set of
rules (e.g., that $\proj^1_{\sigma,\tau}(\mathtt{pair}_{\sigma,\tau}(s,
t)) \succeq s$); to orient these inequalities, it would be good to for
instance have $\interpret{\mathtt{pair}_{\sigma,\tau}(s,t)} \succeq
\pair{\interpret{s}}{\interpret{t}}_{\typeinterpret{\sigma},
\typeinterpret{\tau}}$ and $\interpret{\proj^i_{\sigma,\tau}(s)} =
\pi^i_{\typeinterpret{\sigma},\typeinterpret{\tau}}(\interpret{s})$.
To make the interpretation safe, we additionally include clauses
$\lift(\flatten(x))$ for any unsafe arguments $x$; to make the rules
\emph{strictly} oriented, we include clauses $\lift(1)$.  %Moreover, we
%multiply the first argument of several function symbols by
%$\lift(2)$, since it is typically this symbol that is decreased.
Unfortunately, this approach does not suffice to orient the rules
where some terms are duplicated, such as the second- and third-last
rules.  To handle these rules, we \emph{multiply} the first argument
of several symbols with the second (and possibly third).  Some further
tweaking gives the following safe interpretation, which orients most
of the rules:
%Now, most of the rules can be oriented strictly (so using $\succinterpret$)
%by the safe interpretation:
\[
\begin{array}{rclcrcl}
\Typemap(\bot) & = & \nat
& \quad &
\Typemap(\mathtt{and}) & = & \lambda\alpha_1\lambda\alpha_2 . \alpha_1\times\alpha_2 \\

\Typemap(\exists) & = & \lambda(\alpha : * \arrkind *) . \Sigma \gamma . \alpha \gamma
& \quad &
\Typemap(\mathtt{or}) & = & \lambda\alpha_1\lambda\alpha_2 . \alpha_1\times\alpha_2 \\
\end{array}
\]
\[
\begin{array}{rcll}
\Termmap(\epsilon) & = & \Lambda \alpha:* . \lambda x:\nat. &
  \mathtt{lift}_\alpha(2 \otimes x \oplus 1) \\
\Termmap(@) & = & \Lambda\alpha\Lambda\beta\lambda x: \alpha \arrtype \beta . \lambda y :
  \alpha . \quad & \lift_\beta(2) \otimes (x \cdot y) \oplus
  \lift_\beta(\flatten_\alpha(y)\ \oplus \\
  & & & \phantom{AB}\flatten_{\alpha \arrtype \beta}(x) \otimes
  \flatten_\beta(y) \oplus 1) \\
\Termmap(\mathtt{tapp}) & = & \Lambda \alpha : * \arrkind * . \Lambda \beta . \lambda x : \quant{\gamma}{\alpha\gamma} . \quad &
  \lift_{\alpha\beta}(2) \otimes
  (x * \beta) \oplus \lift_{\alpha\beta}(1) \\
\Termmap(\mathtt{ext}) & = & \Lambda \alpha : * \arrkind * . \Lambda \beta : * . \lambda x:\alpha\beta . &
  \expair{\beta}{x} \oplus \lift_{\Sigma\gamma.\beta\gamma}(
  \flatten_{\alpha\gamma}(x)) \\
\Termmap(\mathtt{pair}) & = & \Lambda \alpha \Lambda \beta \lambda x :
  \alpha, y : \beta.\quad & \pair{x}{y} \oplus \lift_{
  \alpha \times \beta}(\flatten_\alpha(x) \oplus \flatten_{\beta}(y)) \\
\Termmap(\proj^1) & = & \Lambda \alpha \Lambda \beta \lambda x :
  \alpha \times \beta . \quad
  & \lift_\alpha(2) \otimes \pi^1(x) \oplus \lift_{\alpha}(1) \\
\Termmap(\proj^2) & = & \Lambda \alpha \Lambda \beta \lambda x :
  \alpha\times\beta.\quad
  & \lift_\beta(2) \otimes \pi^2(x) \oplus \lift_{\beta}(1) \\
\Termmap(\mathtt{in}^1) & = & \Lambda \alpha \Lambda \beta
  \lambda x : \alpha.\quad & \pair{x}{\lift_\beta(1)}
  \oplus \lift_{\alpha
  \times \beta}(\flatten_{\alpha}(x)) \\
\Termmap(\mathtt{in}^2) & = & \Lambda \alpha \Lambda \beta
  \lambda x : \beta.\quad & \pair{\lift_\alpha(1)}{x}
  \oplus \lift_{\alpha \times \beta}(\flatten_{\beta}(x)) \\
\end{array}
\]
\[
\begin{array}{rcl}
\Termmap(\mathtt{let}) & = & \Lambda \alpha : * \arrkind * . \Lambda \beta : * . \lambda x : \Sigma \xi . \alpha\xi,
  y : \quant{\xi}{\alpha\xi \arrtype \beta}. \\
  & & \lift_\beta(1) \oplus \lift_\beta(2) \otimes
    (\xlet{\beta}{x}{\xi,z}{y\xi z})\ \oplus \\
  & & \lift_\beta(\flatten_{\Sigma\gamma.\alpha\gamma}(x) \oplus 1)
    \otimes (y * \nat \cdot \lift_{\alpha\nat}(0)) \\
\Termmap(\mathtt{case}) & = & \Lambda \alpha,\beta,\xi . \lambda x :
  \alpha \times \beta, y : (\alpha \arrtype \xi), z : (\beta \arrtype
  \xi). \\
  & & \quad
  \lift_\xi(2) \oplus
  \lift_\xi(3 \otimes \flatten_{\alpha \times \beta}(x)) \oplus \\
  & & \quad\phantom{ABCDE}
  \lift_\xi(\flatten_{\alpha \times \beta}(x) \oplus 1)
    \otimes (y \cdot \pi^1(x) \oplus z \cdot \pi^2(x)) \\
\end{array}
\]

Above, $\otimes$ binds stronger than~$\oplus$. The derivations to
orient rules with these interpretations are also given in
\onlypaper{\cite[Appendix~B]{versionwithappendix}}%
\onlyarxiv{Appendix~\ref{app_ineqs}}.

The only rules that are not oriented with this interpretation -- not
with~$\succeq$ either -- are the ones of the form
$f(\mathtt{let}(s,t), \dots) \red \mathtt{let}(s,f(t,\dots))$, like
the rule marked (*) above. Nonetheless, this is already a significant
step towards a systematic, extensible methodology of termination proofs
for IPC2 and similar systems of higher-order logic. Verifying the
orientations is still tedious, but our method raises hope for at least
partial automation, as was done with polynomial interpretations for
non-polymorphic higher-order rewriting~\cite{FuhsKop2012}.

\section{Conclusions and future work}

We introduced a powerful and systematic methodology to prove
termination of higher-order rewriting with full impredicative
polymorphism. To use the method one just needs to invent safe
interpretations and verify the orientation of the rules with the
calculation rules.

As the method is tedious to apply manually for larger systems, a
natural direction for future work is to look into automation: both for
automatic verification that a given interpretation suffices and --
building on existing termination provers for first- and higher-order
term rewriting -- for automatically finding a suitable interpretation.

In addition, it would be worth exploring improvements of the method
that would allow us to handle the remaining rules of IPC2, or
extending other techniques for higher-order termination such as
orderings (see, e.g., \cite{jou:rub:07}) or dependency pairs
(e.g.,~\cite{kop:raa:12,suz:kus:bla:11}).

\addcontentsline{toc}{section}{References}
\bibliography{references}

\begin{thebibliography}{10}

\bibitem{bla:05}
F.~Blanqui.
\newblock Definitions by rewriting in the calculus of constructions.
\newblock {\em MSCS}, 15(1):37--92, 2005.

\bibitem{cou:dow:07}
D.~Cousineau and G.~Dowek.
\newblock Embedding pure type systemsin the lambda-pi-calculus modulo.
\newblock In {\em TLCA}, pages 102--117, 2017.

\bibitem{dow:17}
G.~Dowek.
\newblock Models and termination of proof reduction in the
  $\lambda\pi$-calculus modulo theory.
\newblock In {\em ICALP}, pages 109:1--109:14, 2017.

\bibitem{fio:ham:13}
M.~Fiore and M.~Hamana.
\newblock Multiversal polymorphic algebraic theories: syntacs, semantics,
  translations and equational logic.
\newblock In {\em LICS}, pages 520--520, 2013.

\bibitem{FuhsKop2012}
C.~Fuhs and C.~Kop.
\newblock Polynomial interpretations for higher-order rewriting.
\newblock In {\em RTA}, pages 176--192, 2012.

\bibitem{Girard1989}
J.-V. Girard, P.~Taylor, and Y.~Lafont.
\newblock {\em Proofs and Types}.
\newblock Cambridge University Press, 1989.

\bibitem{Girard1971}
J.-Y. Girard.
\newblock Une extension de l'interpretation de {Gödel} a l'analyse, et son
  application a l'elimination des coupures dans l'analyse et la theorie des
  types.
\newblock In {\em SLS}, pages 63 -- 92. Elsevier, 1971.

\bibitem{ham:18}
M.~Hamana.
\newblock Polymorphic rewrite rules: Confluence, type inference, and instance
  validation.
\newblock In {\em FLOPS}, pages 99--115, 2018.

\bibitem{JacobsRutten2011}
B.~Jacobs and J.~Rutten.
\newblock An introduction to (co)algebras and (co)induction.
\newblock In {\em Advanced Topics in Bisimulation and Coinduction}, pages
  38--99. Cambridge University Press, 2011.

\bibitem{jou:rub:07}
J.~Jouannaud and A.~Rubio.
\newblock Polymorphic higher-order recursive path orderings.
\newblock {\em JACM}, 54(1):1--48, 2007.

\bibitem{Kop2012}
C.~Kop.
\newblock {\em Higher Order Termination}.
\newblock PhD thesis, VU University Amsterdam, 2012.

\bibitem{kop:raa:12}
C.~Kop and F.~van Raamsdonk.
\newblock Dynamic dependency pairs for algebraic functional systems.
\newblock {\em LMCS}, 8(2):10:1--10:51, 2012.

\bibitem{KozenSilva2017}
D.~Kozen and A.~Silva.
\newblock Practical coinduction.
\newblock {\em Mathematical Structures in Computer Science}, 27(7):1132--1152,
  2017.

\bibitem{pol:96}
J.C. van~de Pol.
\newblock {\em Termination of Higher-order Rewrite Systems}.
\newblock PhD thesis, University of Utrecht, 1996.

\bibitem{Sangiorgi2012}
D.~Sangiorgi.
\newblock {\em Introduction to Bisimulation and Coinduction}.
\newblock Cambridge University Press, 2012.

\bibitem{SorensenUrzyczyn2006}
M.H. S{\o}rensen and P.~Urzyczyn.
\newblock {\em Lectures on the Curry-Howard Isomorphism}.
\newblock Elsevier, 2006.

\bibitem{SorensenUrzyczyn2010}
M.H. S{\o}rensen and P.~Urzyczyn.
\newblock A syntactic embedding of predicate logic into second-order
  propositional logic.
\newblock {\em Notre Dame Journal of Formal Logic}, 51(4):457--473, 2010.

\bibitem{suz:kus:bla:11}
S.~Suzuki, K.~Kusakari, and F.~Blanqui.
\newblock Argument filterings and usable rules in higher-order rewrite systems.
\newblock {\em IPSJ Transactions on Programming}, 4(2):1--12, 2011.

\bibitem{Tatsuta2007}
M.~Tatsuta.
\newblock Simple saturated sets for disjunction and second-order existential
  quantification.
\newblock In {\em {TLCA} 2007}, pages 366--380, 2007.

\bibitem{Terese2003}
Terese.
\newblock {\em Term Rewriting Systems}.
\newblock Cambridge University Press, 2003.

\bibitem{TroelstraSchwichtenberg1996}
A.S. Troelstra and H.~Schwichtenberg.
\newblock {\em Basic Proof Theory}.
\newblock Cambridge University Press, 1996.

\bibitem{Pol1993}
J.C. van~de Pol.
\newblock Termination proofs for higher-order rewrite systems.
\newblock In {\em HOA}, pages 305--325, 1993.

\bibitem{PolSchwichtenberg1995}
J.C. van~de Pol and H.~Schwichtenberg.
\newblock Strict functionals for termination proofs.
\newblock In {\em {TLCA} 95}, pages 350--364, 1995.

\bibitem{wah:04}
D.~Wahlstedt.
\newblock {\em Type Theory with First-Order Data Typesand Size-Change
  Termination}.
\newblock PhD thesis, G\"oteborg University, 2004.

\bibitem{wal:03}
D.~Walukiewicz-Chrząszcz.
\newblock Termination of rewriting in the calculus of constructions.
\newblock {\em JFP}, 13(2):339--414, 2003.

\bibitem{Wojdyga2008}
A.~Wojdyga.
\newblock Short proofs of strong normalization.
\newblock In {\em {MFCS} 08}, pages 613--623, 2008.

\end{thebibliography}

\clearpage
\appendix

\section{Complete proofs}\label{app_proofs}

\subsection{Strong Normalisation of~$\arrW$}\label{app_proofs_SN}

By~$\SN$ we denote the set of all interpretation terms terminating
w.r.t.~$\arrW$. We will use $\abstraction{a}{s}$ for either
$\abs{a}{s}$ or $\tabs{a}{s}$, depending on typing.

For $t \in \SN$ by~$\nu(t)$ we denote the length of the longest
reduction starting at~$t$. The following lemma is obvious, but worth
stating explicitly.

\begin{lemma}\label{lem_reduce_abs}
  If $\abstraction{a}{s} \arrW^* t$, then $t = \abstraction{a}{t'}$
  and $s \arrW^* t'$.  If $s \in \SN$ then both $\abs{x}{s}$ and
  $\tabs{\alpha}{s}$ are also in $\SN$.
\end{lemma}

\begin{proof}
  We observe that every reduct of $\abstraction{x}{s}$ has the form
  $\abstraction{x}{s'}$ with $s \arrW s'$, and analogously for
  $\tabs{\alpha}{s}$.  Thus, the first statement follows by induction
  on the length of the reduction $\abstraction{a}{s} \arrW^* t$,
  and the second statement by induction on $\nu(s)$.
\end{proof}

\begin{lemma}\label{lem_circ_sn_base}
  If $t_1,t_2 \in \SN$ then $\circ_\nat t_1 t_2 \in \SN$ for $\circ
  \in \{\oplus,\otimes\}$.
\end{lemma}

\begin{proof}
  By induction on $\nu(t_1) + \nu(t_2)$. Assume $t_1,t_2 \in \SN$. To
  prove $\circ_\nat t_1 t_2 \in \SN$ it suffices to show $s \in \SN$
  for all~$s$ such that $\circ_\nat t_1 t_2 \arrW s$. If $s =
  \circ_\nat t_1' t_2$ or $s = \circ_\nat t_1 t_2'$ with $t_i \arrW
  t_i'$ then we complete by the induction hypothesis. Otherwise $s \in
  \mathbb{N}$ is obviously in $\SN$.
\end{proof}

In the rest of this section we adapt Tait-Girard's method of candidates
to prove termination of~$\arrW$. The proof is an adaptation of
chapters~6 and~14 from the book~\cite{Girard1989}, and chapters~10
and~11 from the book~\cite{SorensenUrzyczyn2006}.

\begin{defn}\label{def_candidate}
  A term~$t$ is \emph{neutral} if there does not exist a sequence of
  terms and types~$u_1,\ldots,u_n$ with $n \ge 1$ such that $t u_1
  \ldots u_n$ is a redex (by~$\arrW$).

  By induction on the kind~$\kappa$ of a type constructor~$\tau$ we
  define the set~$\Cb_\tau$ of all candidates of type
  constructor~$\tau$.

  First assume $\kappa=*$, i.e., $\tau$ is a type. A set~$X$ of
  interpretation terms of type~$\tau$ is a \emph{candidate of
    type~$\tau$} when:
  \begin{enumerate}
  \item $X \subseteq \SN$;
  \item if $t \in X$ and $t \arrW t'$ then $t' \in X$;
  \item if $t$ is neutral and for every~$t'$ with $t \arrW t'$ we
    have $t' \in X$, then $t \in X$;
  \item if $t_1,t_2 \in X$ then $\circ_\tau t_1 t_2 \in X$ for
    $\circ \in \{\oplus,\otimes\}$;
  \item if $t \in \SN$ and $t : \nat$ then $\lift_\tau t \in X$;
  \item if $t \in X$ then $\flatten_\tau t \in \SN$.
  \end{enumerate}
  Note that item~3 above implies:
  \begin{itemize}
  \item if $t$ is neutral and in normal form then $t \in X$.
  \end{itemize}

  Now assume $\kappa = \kappa_1\arrkind\kappa_2$. A function $f :
  \Tc_{\kappa_1} \times \bigcup_{\xi\in\Tc_{\kappa_1}}\Cb_\xi \to
  \bigcup_{\xi\in\Tc_{\kappa_2}}\Cb_\xi$ is a \emph{candidate of type
    constructor~$\tau$} if for every closed type constructor~$\sigma$
  of kind~$\kappa_1$ and a candidate $X \in \Cb_\sigma$ we have
  $f(\sigma,X) \in \Cb_{\tau\sigma}$.
\end{defn}

Note that the elements of a candidate of type~$\tau$ are required to
have type~$\tau$.

\begin{lemma}\label{lem_beta_candidate}
  If $\sigma =_\beta \sigma'$ then $\Cb_\sigma = \Cb_{\sigma'}$.
\end{lemma}

\begin{proof}
  Induction on the kind of~$\sigma$.
\end{proof}

\begin{defn}\label{def_computability_valuation}
  Let $\omega$ be a mapping from type constructor variables to type
  constructors (respecting kinds). The mapping~$\omega$ extends in an
  obvious way to a mapping from type constructors to type
  constructors. A mapping~$\omega$ is \emph{closed for~$\sigma$} if
  $\omega(\alpha)$ is closed for $\alpha \in \FTV(\sigma)$ (then
  $\omega(\sigma)$ is closed).

  An \emph{$\omega$-valuation} is a mapping~$\xi$ from type
  constructor variables to candidates such that $\xi(\alpha) \in
  \Cb_{\omega(\alpha)}$.

  For each type constructor~$\sigma$, each mapping~$\omega$ closed
  for~$\sigma$, and each $\omega$-valuation~$\xi$, the set
  $\val{\sigma}{\xi}{\omega}$ is defined by induction on~$\sigma$:
  \begin{itemize}
  \item $\val{\alpha}{\xi}{\omega} = \xi(\alpha)$ for a type
    constructor variable~$\alpha$,
  \item $\val{\nat}{\xi}{\omega}$ is the set of all terms~$t \in \SN$
    such that $t : \nat$,
  \item $\val{\sigma \arrtype \tau}{\xi}{\omega}$ is the set of all
    terms~$t$ such that $t : \omega(\sigma\arrtype\tau)$ and for
    every~$s \in \val{\sigma}{\xi}{\omega}$ with $s : \omega(\sigma)$
    we have $\app{t}{s} \in \val{\tau}{\xi}{\omega}$,
  \item $\val{\forall(\alpha:\kappa)\sigma}{\xi}{\omega}$ is the set
    of all terms~$t$ such that $t : \omega(\forall\alpha\sigma)$ and
    for every closed type constructor~$\varphi$ of kind~$\kappa$ and
    every $X \in \Cb_\varphi$ we have $\tapp{t}{\varphi} \in
    \val{\sigma}{\xi[\subst{\alpha}{X}]}{\omega[\subst{\alpha}{\varphi}]}$,
  \item
    $\val{\varphi \psi}{\xi}{\omega} =
    \val{\varphi}{\xi}{\omega}(\omega(\psi),\val{\psi}{\xi}{\omega})$,
  \item
    $\val{\lambda(\alpha:\kappa)\varphi}{\xi}{\omega}(\psi,X) =
    \val{\varphi}{\xi[\subst{\alpha}{X}]}{\omega[\subst{\alpha}{\psi}]}$
    for closed $\psi \in \Tc_\kappa$ and $X \in \Cb_\psi$.
  \end{itemize}
  In the above, if e.g.~$\val{\psi}{\xi}{\omega} \notin
  \Cb_{\omega(\psi)}$ then $\val{\varphi \psi}{\xi}{\omega}$ is
  undefined.
\end{defn}

If~$\varphi$ is closed then $\omega,\xi$ do not affect the value
of~$\val{\varphi}{\xi}{\omega}$, so then we simply
write~$\val{\varphi}{}{}$.

\begin{lemma}\label{lem_nat_computable}
  $\val{\nat}{}{} \in \Cb_{\nat}$.
\end{lemma}

\begin{proof}
  We check the conditions in Definition~\ref{def_candidate}.
  \begin{enumerate}
  \item $\val{\nat}{}{} \subseteq \SN$ follows
    directly from Definition~\ref{def_computability_valuation}.
  \item Let $t \in \val{\nat}{}{}$ and $t \arrW t'$. Then $t :
    \nat$ and $t \in \SN$. Hence $t' \in \SN$, and $t' : \nat$ by the
    subject reduction lemma. Thus $t' \in \val{\nat}{}{}$.
  \item Let $t$ be neutral and $t : \nat$. Assume that for all~$t'$
    with $t \arrW t'$ we have $t' \in \val{\nat}{}{}$, so in
    particular $t' \in \SN$. But then $t \in \SN$. Hence $t \in
    \val{\nat}{}{}$.
  \item Let $t_1,t_2 \in \SN$ be such that $t_i : \nat$. Obviously,
    $\circ_\nat t_1 t_2 : \nat$. Also $\circ_\nat t_1 t_2 \in \SN$
    follows by Lemma~\ref{lem_circ_sn_base}. So $\circ_\nat t_1 t_2
    \in \val{\nat}{}{}$.
  \item Let $t \in \SN$ be such that $t : \nat$. Then $\lift_\nat t :
    \nat$. It remains to show $\lift_\nat t \in \SN$. Any infinite
    reduction from~$\lift_\nat t$ has the form $\lift_\nat t
    \arrW^* \lift_\nat t_0 \arrW t_1 \arrW t_2 \arrW
    \ldots$ or $\lift_\nat t \arrW \lift_\nat t_0 \arrW
    \lift_\nat t_1 \arrW \lift_\nat t_2 \arrW \ldots$, where $t
    \arrW^* t_0$ and $t_i \arrW t_{i+1}$. This contradicts $t
    \in \SN$.
  \item Let $t \in \SN$ be such that $t : \nat$. The proof of
    $\flatten_\nat t \in \SN$ is analogous to the proof of $\lift_\nat
    t \in \SN$ above.\qedhere
  \end{enumerate}
\end{proof}

\begin{lemma}\label{lem_chi_kappa_computable}
  $\val{\chi_\kappa}{}{} \in \Cb_{\chi_\kappa}$.
\end{lemma}

\begin{proof}
  Induction on~$\kappa$. If $\kappa = *$ then this follows from
  Lemma~\ref{lem_nat_computable}. If $\kappa=\kappa_1\arrkind\kappa_2$
  then $\chi_\kappa = \lambda \alpha : \kappa_1
  . \chi_{\kappa_2}$. Let~$\psi$ be a closed type constructor of
  kind~$\kappa_1$ and let $X \in \Cb_{\chi_{\kappa_1}}$. We have
  $\val{\chi_\kappa}{}{}(\psi,X) = \val{\chi_{\kappa_2}}{}{}$ because
  $\chi_{\kappa_2}$ is closed. By the inductive hypothesis
  $\val{\chi_\kappa}{}{}(\psi,X) = \val{\chi_{\kappa_2}}{}{} \in
  \Cb_{\chi_{\kappa_2}}$. This implies $\val{\chi_\kappa}{}{} \in
  \Cb_{\chi_\kappa}$.
\end{proof}

\begin{lemma}\label{lem_abstraction_computable}
  Let $\sigma,\tau$ be types. Suppose $\val{\tau}{\xi'}{\omega'} \in
  \Cb_{\omega'(\tau)}$ and $\val{\sigma}{\xi'}{\omega'} \in
  \Cb_{\omega'(\sigma)}$ for all suitable $\omega',\xi'$. Then
  \begin{itemize}
  \item
    $\abs{x}{s} \in \val{\tau \arrtype \sigma}{\xi}{\omega}$ if and
    only if $\abs{x}{s} : \omega(\tau \arrtype \sigma)$ and $s[x:=t]
    \in \val{\sigma}{\xi}{\omega}$ for all $t \in
    \val{\tau}{\xi}{\omega}$;
  \item
    $\tabs{\alpha}{s} \in
    \val{\quant{(\alpha:\kappa)}{\sigma}}{\xi}{\omega}$ if and only if
    $\tabs{\alpha}{s} : \omega(\quant{(\alpha:\kappa)}{\sigma})$ and
    for every closed type constructor~$\varphi$ of kind~$\kappa$ and
    all $X \in \Cb_\varphi$ we have $s[\alpha:=\varphi] \in
    \val{\sigma}{\xi[\subst{\alpha}{X}]}{\omega[\subst{\alpha}{\varphi}]}$.
  \end{itemize}
\end{lemma}

\begin{proof}
  First suppose
  $\abs{x:\omega(\tau)}{s} \in \val{\tau \arrtype
    \sigma}{\xi}{\omega}$. Then
  $\abs{x:\omega(\tau)}{s} : \omega(\tau\arrtype\sigma)$ and for all
  $t \in \val{\tau}{\xi}{\omega}$ we have
  $\app{(\abs{x:\omega(\tau)}{s})}{t} \in \val{\sigma}{\xi}{\omega}$.
  As this set is a candidate, it is closed under $\arrW$, so also
  $s[x:=t] \in \val{\sigma}{\xi}{\omega}$. Similarly, if
  $\tabs{\alpha}{s} \in \val{\quant{\alpha}{\sigma}}{\xi}{\omega}$,
  then $\tabs{\alpha}{s} : \quant{\alpha}{\sigma}$ and
  $\tapp{(\tabs{\alpha}{s})}{\varphi} \in
  \val{\sigma}{\xi[\subst{\alpha}{X}]}{\omega[\subst{\alpha}{\varphi}]}$,
  and we are done because
  $\tapp{(\tabs{\alpha}{s})}{\tau} \arrW s[\alpha:=\varphi]$ and
  $\val{\sigma}{\xi[\subst{\alpha}{X}]}{\omega[\subst{\alpha}{\varphi}]}$
  is a candidate, so it is closed under~$\arrW$.

  Now suppose $s[x:=t] \in \val{\sigma}{\xi}{\omega}$ for all
  $t \in \val{\tau}{\xi}{\omega}$. Let
  $t \in \val{\tau}{\xi}{\omega}$. Then $t \in \SN$ because
  $\val{\tau}{\xi}{\omega}$ is a candidate. Also $s \in \SN$ because
  every infinite reduction in $s$ induces an infinite reduction in
  $s[x:=t]$ ($\arrW$ is stable) and
  $\val{\sigma}{\xi}{\omega} \subseteq \SN$ is a candidate. For all
  $s',t'$ with $s \arrW^* s'$ and $t \arrW^* t'$, we show by
  induction on~$\nu(s') + \nu(t')$ that
  $\app{(\abs{x}{s'})} t' \in \val{\sigma}{\xi}{\omega}$. We have
  $\app{(\abs{x}{s'})} t' : \omega(\sigma)$ by definition and the
  subject reduction theorem (note that $t : \omega(\tau)$ because
  $\val{\tau}{\xi}{\omega} \in \Cb_{\omega(\tau)}$). The set
  $\val{\sigma}{\xi}{\omega}$ is a candidate, and
  $\app{(\abs{x}{s'})}{t'}$ is neutral, so in
  $\val{\sigma}{\xi}{\omega}$ if all its reducts are. Thus assume
  $\app{(\abs{x}{s'})}{t'} \arrW u$. If
  $u = \app{(\abs{x}{s'})}{t''}$ with $t' \arrW t''$ or
  $u = \app{(\abs{x}{s''})}{t'}$ with $s' \arrW s''$, then
  $u \in \val{\sigma}{\xi}{\omega}$ by the inductive hypothesis. So
  assume $u = s'[x:=t']$. We have $s[x:=t] \arrW^* s'[x:=t']$ by
  monotonicity and stability of $\arrW$. Therefore
  $u = s'[x:=t'] \in \val{\sigma}{\xi}{\omega}$, because
  $s[x:=t] \in \val{\sigma}{\xi}{\omega}$ and
  $\val{\sigma}{\xi}{\omega}$ is a candidate and hence closed under
  $\arrW$.

  A similar reasoning applies to $s[\alpha:=\varphi]$.
\end{proof}

\begin{lemma}\label{lem_val_computable}
  If $\sigma$ is a type constructor, $\omega$ is closed for~$\sigma$,
  and $\xi$ is an $\omega$-valuation, then $\val{\sigma}{\xi}{\omega}
  \in \Cb_{\omega(\sigma)}$.
\end{lemma}

\begin{proof}
  By induction on the structure of~$\sigma$ we show that
  $\val{\sigma}{\xi}{\omega} \in \Cb_{\omega(\sigma)}$ for all
  suitable $\omega,\xi$. First, if $\sigma = \alpha$ is a type
  constructor variable~$\alpha$ then $\val{\sigma}{\xi}{\omega} =
  \xi(\alpha) \in \Cb_{\omega(\sigma)}$ by definition. If $\sigma =
  \nat$ then $\val{\nat}{\xi}{\omega} \in \Cb_{\nat}$ by
  Lemma~\ref{lem_nat_computable}.

  Assume $\sigma = \tau_1 \arrtype \tau_2$. We check the conditions in
  Definition~\ref{def_candidate}.
  \begin{enumerate}
  \item Let $t \in \val{\tau_1\arrtype\tau_2}{\xi}{\omega}$ and assume
    there is an infinite reduction $t \arrW t_1 \arrW t_2
    \arrW t_3 \arrW \ldots$. By the inductive hypothesis
    $\val{\tau_1}{\xi}{\omega}$ and $\val{\tau_2}{\xi}{\omega}$ are
    candidates. Let~$x$ be a fresh variable. Then $x^{\omega(\tau_1)}
    : \omega(\tau_1)$ and $x^{\omega(\tau_1)} \in
    \val{\tau_1}{\xi}{\omega}$ because it is neutral and normal. Thus
    $t x \in \val{\tau_2}{\xi}{\omega} \subseteq \SN$. But $t x
    \arrW t_1 x \arrW t_2 x \arrW t_3 x \arrW
    \ldots$. Contradiction.
  \item Let $t \in \val{\tau_1\arrtype\tau_2}{\xi}{\omega}$ and $t
    \arrW t'$. Let $u \in \val{\tau_1}{\xi}{\omega}$ be such that
    $u : \omega(\tau_1)$. Then $t u \in \val{\tau_2}{\xi}{\omega}$. By
    the inductive hypothesis $\val{\tau_2}{\xi}{\omega}$ is a
    candidate, so $t' u \in \val{\tau_2}{\xi}{\omega}$. Also note that
    $t' : \omega(\tau_1 \arrtype \tau_2)$ by the subject reduction
    lemma. Hence $t' \in \val{\tau_1\arrtype\tau_2}{\xi}{\omega}$.
  \item Let $t$ be neutral such that $t : \omega(\tau_1 \arrtype
    \tau_2)$. Assume for every~$t'$ with $t \arrW t'$ we have $t'
    \in \val{\tau_1\arrtype\tau_2}{\xi}{\omega}$. Let $u \in
    \val{\tau_1}{\xi}{\omega}$ be such that $u : \omega(\tau_1)$. By
    the inductive hypothesis $\val{\tau_1}{\xi}{\omega}$ is a
    candidate, so $u \in \SN$. By induction on~$\nu(u)$ we show that
    $t u \in \val{\tau_2}{\xi}{\omega}$. Assume $t u \arrW t''$. We
    show $t'' \in \val{\tau_2}{\xi}{\omega}$. Because~$t$ is neutral,
    $t u$ cannot be a redex. So there are two cases.
    \begin{itemize}
    \item $t'' = t u'$ with $u \arrW u'$. Then $u' \in
      \val{\tau_1}{\xi}{\omega}$ because~$\val{\tau_1}{\xi}{\omega}$
      is a candidate, and~$u' : \omega(\tau_1)$ by the subject
      reduction lemma. So $t u' \in \val{\tau_2}{\xi}{\omega}$ by the
      inductive hypothesis for~$u$.
    \item $t'' = t' u$ with $t \arrW t'$. Then $t' \in
      \val{\tau_1\arrtype\tau_2}{\xi}{\omega}$ by point~2 above. So
      $t' u \in \val{\tau_2}{\xi}{\omega}$.
    \end{itemize}
    We have thus shown that if $t u \arrW t''$ then $t'' \in
    \val{\tau_2}{\xi}{\omega}$. By the (main) inductive hypothesis
    $\val{\tau_2}{\xi}{\omega}$ is a candidate. Because $t u$
    is neutral, the above implies $t u \in
    \val{\tau_2}{\xi}{\omega}$. Since $u \in
    \val{\tau_1}{\xi}{\omega}$ was arbitrary with $u :
    \omega(\tau_1)$, we have shown $t \in
    \val{\tau_1\arrtype\tau_2}{\xi}{\omega}$.
  \item Assume $t_1,t_2 \in \val{\tau_1\arrtype\tau_2}{\xi}{\omega}$.
    We have already shown that this implies $t_1,t_2 \in \SN$. Let $s
    = \circ_{\omega(\tau_1\arrtype\tau_2)} t_1 t_2$. We show $s \in
    \val{\tau_1\arrtype\tau_2}{\xi}{\omega}$ by induction on $\nu(t_1)
    + \nu(t_2)$. Note that $s : \omega(\tau_1\arrtype\tau_2)$ because
    $t_i : \omega(\tau_1\arrtype\tau_2)$. Since $s$ is neutral, we
    have already seen in point~3 above that to prove $s \in
    \val{\tau_1\arrtype\tau_2}{\xi}{\omega}$ it suffices to show that
    $s' \in \val{\tau_1\arrtype\tau_2}{\xi}{\omega}$ whenever $s
    \arrW s'$. If $s' = \circ_{\omega(\tau_1\arrtype\tau_2)} t_1'
    t_2$ with $t_1 \arrW t_1'$, then note that $t_1' \in
    \val{\tau_1\arrtype\tau_2}{\xi}{\omega}$ because we have already
    shown that $\val{\tau_1\arrtype\tau_2}{\xi}{\omega}$ is closed
    under $\arrW$; thus, we can complete by the induction
    hypothesis. If $s' = \circ_{\omega(\tau_1\arrtype\tau_2)} t_1
    t_2'$, we complete in the same way.  The only alternative is that
    $s' = \abs{x:\omega(\tau_1)}{\circ_{\omega(\tau_2)}(t_1x)(t_2x)}$.
    Let $u \in \val{\tau_1}{\xi}{\omega}$. Then $u : \omega(\tau_1)$
    because $\val{\tau_1}{\xi}{\omega} \in \Cb_{\omega(\tau_1)}$ by
    the inductive hypothesis. Since $t_1,t_2 \in
    \val{\tau_1\arrtype\tau_2}{\xi}{\omega}$, we have that $t_1 u$ and
    $t_2 u$ are in $\val{\tau_2}{\xi}{\omega}$ by definition. Since
    $\val{\tau_2}{\xi}{\omega}$ is a candidate, this means that
    $\circ_{\omega(\tau_2)} (t_1 u) (t_2 u) = (\circ_{\omega(\tau_2)}
    (t_1 x) (t_2 x))[x:=u]$ is in $\val{\tau_2}{\xi}{\omega}$ as well.
    By Lemma~\ref{lem_abstraction_computable}, we conclude that $s'
    \in \val{\tau_1\arrtype\tau_2}{\xi}{\omega}$.
  \item Let $t \in \SN$ satisfy $t : \nat$, and let $s =
    \lift_{\omega(\tau_1\arrtype\tau_2)}(t)$. We show $s \in
    \val{\tau_1\arrtype\tau_2}{\xi}{\omega}$ by induction
    on~$\nu(t)$. We have $s : \omega(\tau_1\arrtype\tau_2)$ because $t
    : \nat$. Since~$s$ is neutral, we have already proved above in
    point~3 that it suffices to show that $s' \in
    \val{\tau_1\arrtype\tau_2}{\xi}{\omega}$ whenever $s \arrW
    s'$. If $s' = \lift_{\omega(\tau_1\arrtype\tau_2)}(t')$ with $t
    \arrW t'$ then still $t' \in \SN$ and $t' : \nat$, so $s' \in
    \val{\tau_1\arrtype\tau_2}{\xi}{\omega}$ by the inductive
    hypothesis. The only alternative is that $s' = \lambda x :
    \omega(\tau_1) . \lift_{\omega(\tau_2)}(t)$. Let $u \in
    \val{\tau_1}{\xi}{\omega}$ be such that $u :
    \omega(\tau_1)$. Because $\val{\tau_2}{\xi}{\omega} \in
    \Cb_{\omega(\tau_2)}$ by the (main) inductive hypothesis
    for~$\sigma$, we have $\lift_{\omega(\tau_2)}(t) \in
    \val{\tau_2}{\xi}{\omega}$. Since $\lift_{\omega(\tau_2)}(t) =
    (\lift_{\omega(\tau_2)}x)[\subst{x}{t}]$ we obtain $s' \in
    \val{\tau_1\arrtype\tau_2}{\xi}{\omega}$ by
    Lemma~\ref{lem_abstraction_computable}.
  \item Let $t \in \val{\tau_1\arrtype\tau_2}{\xi}{\omega}$.  We show
    $s := \flatten_{\omega(\tau_1\arrtype\tau_2)}t \in \SN$. We have
    already shown $t \in \SN$ in point~1 above. Thus any infinite
    reduction starting from~$s$ must have the form $s \arrW^*
    \flatten_{\omega(\tau_1\arrtype\tau_2)}t' \arrW
    \flatten_{\omega(\tau_2)}(t' (\lift_{\omega(\tau_1)}0)) \arrW
    \ldots$ with $t \arrW^* t'$. We have already shown in point~2
    above that $\val{\tau_1\arrtype\tau_2}{\xi}{\omega}$ is closed
    under~$\arrW$, so $t' \in
    \val{\tau_1\arrtype\tau_2}{\xi}{\omega}$. By the inductive
    hypothesis $\val{\tau_1}{\xi}{\omega} \in\Cb_{\omega(\tau_1)}$, so
    $\lift_{\omega(\tau_1)}0 \in \val{\tau_1}{\xi}{\omega}$ by
    property~5 of candidates. Hence $t' (\lift_{\omega(\tau_1)}0) \in
    \val{\tau_2}{\xi}{\omega}$ by definition. But by the inductive
    hypothesis~$\val{\tau_2}{\xi}{\omega}$ is a candidate, so
    $\flatten_{\omega(\tau_2)}(t'(\lift_{\omega(\tau_1)}0))\in\SN$. Contradiction.
  \end{enumerate}

  Assume $\sigma = \forall(\alpha:\kappa)\tau$. We check the
  conditions in Definition~\ref{def_candidate}.
  \begin{enumerate}
  \item Let $t \in \val{\forall(\alpha:\kappa)\tau}{\xi}{\omega}$
    and assume there is an infinite reduction $t \arrW t_1 \arrW
    t_2 \arrW t_3 \arrW \ldots$. Recall that~$\chi_\kappa$ from
    Definition~\ref{def_types} is a closed type constructor of
    kind~$\kappa$. By Lemma~\ref{lem_chi_kappa_computable} we have
    $\val{\chi_{\kappa}}{}{} \in \Cb_{\chi_\kappa}$. Then $t
    \chi_\kappa \in
    \val{\tau}{\xi[\subst{\alpha}{\val{\chi_\kappa}{}{}}]}{\omega[\subst{\alpha}{\chi_\kappa}]}$. By
    the inductive hypothesis
    $\val{\tau}{\xi[\subst{\alpha}{\val{\chi_\kappa}{}{}}]}{\omega[\subst{\alpha}{\chi_\kappa}]}$
    is a candidate, so $t \chi_\kappa \in \SN$. But $t \chi_\kappa
    \arrW t_1 \chi_\kappa \arrW t_2 \chi_\kappa \arrW t_3
    \chi_\kappa \arrW \ldots$. Contradiction.
  \item Let $t \in \val{\forall\alpha\tau}{\xi}{\omega}$ and $t
    \arrW t'$. By the subject reduction lemma $t' :
    \omega(\forall\alpha\tau)$. Let~$\varphi$ be a closed type
    constructor of kind~$\kappa$ and~$X \in \Cb_{\varphi}$. Then $t
    \varphi \in
    \val{\tau}{\xi[\subst{\alpha}{X}]}{\omega[\subst{\alpha}{\varphi}]}$. By
    the inductive hypothesis
    $\val{\tau}{\xi[\subst{\alpha}{X}]}{\omega[\subst{\alpha}{\varphi}]}$
    is a candidate, so $t' \varphi \in
    \val{\tau}{\xi[\subst{\alpha}{X}]}{\omega[\subst{\alpha}{\varphi}]}$. Therefore
    $t' \in \val{\forall\alpha\tau}{\xi}{\omega}$.
  \item Let $t$ be neutral such that $t :
    \omega(\forall\alpha\tau)$, and assume that for every~$t'$ with
    $t \arrW t'$ we have $t' \in
    \val{\forall\alpha\tau}{\xi}{\omega}$. Let~$\varphi$ be a closed
    type constructor of kind~$\kappa$ and~$X \in
    \Cb_{\varphi}$. Assume $t \varphi \arrW t''$. Then $t'' = t'
    \varphi$ with $t \arrW t'$, because~$t$ is neutral. Hence $t
    \varphi \arrW t' \varphi \in
    \val{\tau}{\xi[\subst{\alpha}{X}]}{\omega[\subst{\alpha}{\varphi}]}$. By
    the inductive
    hypothesis~$\val{\tau}{\xi[\subst{\alpha}{X}]}{\omega[\subst{\alpha}{\varphi}]}$
    is a candidate. Also $t \varphi$ is neutral, so $t \varphi \in
    \val{\tau}{\xi[\subst{\alpha}{X}]}{\omega[\subst{\alpha}{\varphi}]}$
    because~$t''$ was arbitrary with $t \varphi \arrW t''$. This
    implies that $t \in \val{\forall\alpha\tau}{\xi}{\omega}$.
  \item Assume $t_1,t_2 \in
    \val{\forall\alpha\tau}{\xi}{\omega}$. We have already shown
    that this implies $t_1,t_2 \in \SN$. We prove
    $\circ_{\omega(\forall\alpha\tau)} t_1 t_2 \in
    \val{\forall\alpha\tau}{\xi}{\omega}$ by induction on $\nu(t_1)
    + \nu(t_2)$. Since $s := \circ_{\omega(\forall\alpha\tau)} t_1
    t_2$ is neutral, we have already proven that it suffices to show
    that $s' \in \val{\forall\alpha\tau}{\xi}{\omega}$ whenever $s
    \arrW s'$. The cases when $t_1$ or $t_2$ are reduced are
    immediate with the induction hypotheses. The only remaining case
    is when $s'=\tabs{\alpha}{\circ_{\omega(\tau)} (t_1 \alpha) (t_2
      \alpha)}$.  For all closed type constructors $\varphi$ of
    kind~$\kappa$ and all $X \in \Cb_{\varphi}$ we have both
    $t_1 \varphi$ and $t_2 \varphi$ in
    $\val{\tau}{\xi[\subst{\alpha}{X}]}{\omega[\subst{\alpha}{\varphi}]}$
    (by definition of $t_1,t_2 \in
    \val{\forall\alpha\tau}{\xi}{\omega}$). Let $\omega' =
    \omega[\subst{\alpha}{\varphi}]$. By bound variable renaming, we
    may assume $\omega(\alpha) = \alpha$ and $\alpha$ does not occur
    in~$t_1,t_2$. Because
    $\val{\tau}{\xi[\subst{\alpha}{X}]}{\omega[\subst{\alpha}{\varphi}]}$
    is a candidate by the inductive hypothesis for~$\sigma$, we have
    \[
    \circ_{\omega'(\tau)} (t_1 \varphi)
    (t_2\varphi) = (\circ_{\omega(\tau)} (t_1 \alpha) (t_2
    \alpha))[\subst{\alpha}{\varphi}] \in
    \val{\tau}{\xi[\subst{\alpha}{X}]}{\omega[\subst{\alpha}{\varphi}]}.
    \]
    Hence $s' \in \val{\forall\alpha\tau}{\xi}{\omega}$ by
    Lemma~\ref{lem_abstraction_computable}.
  \item Let $t \in \SN$ be such that $t : \nat$. By induction
    on~$\nu(t)$ we show $s := \lift_{\omega(\forall\alpha\tau)}(t)
    \in \val{\forall\alpha\tau}{\xi}{\omega}$. First note that $s :
    \omega(\forall\alpha\tau)$. Since~$s$ is neutral, by the already
    proven point~3 above, it suffices to show that $s' \in
    \val{\forall\alpha\tau}{\xi}{\omega}$ whenever $s \arrW
    s'$. The case when~$t$ is reduced is immediate by the inductive
    hypothesis. The only remaining case is when $s' =
    \tabs{\alpha}{\lift_{\omega(\tau)}(t)}$ (without loss of
    generality assuming $\omega(\alpha) = \alpha$). Let $\varphi$ be a
    closed type constructor of kind~$\kappa$ and let $X \in
    \Cb_\varphi$. Because
    $\val{\tau}{\xi[\subst{\alpha}{X}]}{\omega[\subst{\alpha}{\varphi}]}$
    is a candidate, we have
    \[
    \lift_{\omega[\subst{\alpha}{\varphi}](\tau)}(t) =
    (\lift_{\omega(\tau)}(t))[\subst{\alpha}{\varphi}] \in
    \val{\tau}{\xi[\subst{\alpha}{X}]}{\omega[\subst{\alpha}{\varphi}]}.
    \]
    This implies $s' \in \val{\forall\alpha\tau}{\xi}{\omega}$.
  \item Let $t \in \val{\forall\alpha\tau}{\xi}{\omega}$. We show $s
    := \flatten_{\omega(\forall\alpha\tau)}t \in \SN$. We have
    already shown $t \in \SN$ in point~1 above. Thus any infinite
    reduction starting from~$s$ must have the form $s \arrW^*
    \flatten_{\omega(\forall\alpha\tau)}t' \arrW
    \flatten_{\omega(\tau)[\subst{\alpha}{\chi_\kappa}]}(t'
    \chi_\kappa) \arrW \ldots$ with $t \arrW^* t'$ (assuming
    $\omega(\alpha) = \alpha$ without loss of generality). We have
    already shown in point~2 above that
    $\val{\forall\alpha\tau}{\xi}{\omega}$ is closed
    under~$\arrW$, so $t' \in
    \val{\forall\alpha\tau}{\xi}{\omega}$. We have
    $\val{\chi_\kappa}{}{} \in \Cb_{\chi_\kappa}$ by
    Lemma~\ref{lem_chi_kappa_computable}. Since~$\chi_\kappa$ is also
    closed, we have $t' \chi_\kappa \in
    \val{\tau}{\xi[\subst{\alpha}{\val{\chi_\kappa}{}{}}]}{\omega[\subst{\alpha}{\chi_\kappa}]}$
    by definition of $\val{\forall\alpha\tau}{\xi}{\omega}$. By the
    inductive hypothesis
    $\val{\tau}{\xi[\subst{\alpha}{\val{\chi_\kappa}{}{}}]}{\omega[\subst{\alpha}{\chi_\kappa}]}
    \in \Cb_{\omega[\subst{\alpha}{\chi_\kappa}](\tau)}$. Hence
    $\flatten_{\omega[\subst{\alpha}{\chi_\kappa}](\tau)}(t'\chi_\kappa)\in\SN$. But
    $\omega[\subst{\alpha}{\chi_\kappa}](\tau) =
    \omega(\tau)[\subst{\alpha}{\chi_\kappa}]$ because~$\chi_\kappa$
    is closed and $\omega(\alpha) = \alpha$. Contradiction.
  \end{enumerate}

  Assume $\sigma = \varphi\psi$, with $\psi$ of kind~$\kappa_1$ and
  $\varphi$ of kind~$\kappa_1\arrkind\kappa_2$. By the inductive
  hypothesis $\val{\psi}{\xi}{\omega} \in \Cb_{\omega(\psi)}$ and
  $\val{\varphi}{\xi}{\omega} \in \Cb_{\omega(\varphi)}$. Because
  applying~$\omega$ does not change kinds, we have
  $\val{\varphi\psi}{\xi}{\omega} =
  \val{\varphi}{\xi}{\omega}(\omega(\psi), \val{\psi}{\xi}{\omega})
  \in \Cb_{\omega(\varphi\psi)}$, by the definition of candidates of a
  type constructor with kind~$\kappa_1\arrkind\kappa_2$ (note that
  $\omega(\psi)$ is closed, because $\omega$ is closed for~$\sigma$).

  Finally, assume $\sigma = \lambda(\alpha:\kappa)\varphi$. Let $\psi$
  be a closed type constructor of kind~$\kappa$ and $X \in
  \Cb_{\psi}$. By the inductive hypothesis
  $\val{\lambda(\alpha:\kappa)\varphi}{\xi}{\omega}(\psi,X) =
  \val{\varphi}{\xi[\subst{\alpha}{X}]}{\omega[\subst{\alpha}{\psi}]}
  \in \Cb_{\omega[\subst{\alpha}{\psi}](\varphi)}$. Because $\psi$ is
  closed we have $\omega[\subst{\alpha}{\psi}](\varphi) =
  \omega(\varphi[\subst{\alpha}{\psi}]) =_\beta
  \omega((\lambda\alpha.\varphi)\psi) = \omega(\sigma\psi) =
  \omega(\sigma)\psi$. By Lemma~\ref{lem_beta_candidate} this implies
  that $\val{\sigma}{\xi}{\omega} \in \Cb_{\omega(\sigma)}$.
\end{proof}

\begin{lemma}\label{lem_circ}
  $\circ \in \val{\forall\alpha . \alpha \arrtype \alpha \arrtype
    \alpha}{}{}$ for $\circ \in \{ \oplus, \otimes \}$.
\end{lemma}

\begin{proof}
  Follows from definitions and property~4 of candidates.
\end{proof}

\begin{lemma}\label{lem_lift}
  $\lift \in \val{\forall\alpha.\nat\arrtype\alpha}{}{}$.
\end{lemma}

\begin{proof}
  Follows from definitions and property~5 of candidates.
\end{proof}

\begin{lemma}\label{lem_flatten}
  $\flatten \in \val{\forall\alpha.\alpha\arrtype\nat}{}{}$.
\end{lemma}

\begin{proof}
  Follows from definitions and property~6 of candidates.
\end{proof}

\begin{lemma}\label{lem_val_subst}
  For any type constructors~$\sigma,\tau$ with $\alpha \notin
  \FTV(\tau)$, a mapping~$\omega$ closed for~$\sigma$ and for~$\tau$,
  and an $\omega$-valuation~$\xi$, we have:
  \[
  \val{\sigma[\subst{\alpha}{\tau}]}{\xi}{\omega} =
  \val{\sigma}{\xi[\subst{\alpha}{\val{\tau}{\xi}{\omega}}]}{\omega[\subst{\alpha}{\omega(\tau)}]}.
  \]
\end{lemma}

\begin{proof}
  Let~$\omega' = \omega[\subst{\alpha}{\omega(\tau)}]$ and $\xi' =
  \xi[\subst{\alpha}{\val{\tau}{\xi}{\omega}}]$. First note
  that~$\omega$ is closed for~$\sigma[\subst{\alpha}{\tau}]$
  and~$\omega'$ is closed for~$\sigma$. We proceed by induction
  on~$\sigma$. If $\alpha \notin \FTV(\sigma)$ then the claim is
  obvious. If $\sigma = \alpha$ then
  $\val{\sigma[\subst{\alpha}{\tau}]}{\xi}{\omega} =
  \val{\tau}{\xi}{\omega} = \val{\sigma}{\xi'}{\omega'}$.

  Assume $\sigma = \sigma_1\arrtype\sigma_2$. We show
  $\val{\sigma[\subst{\alpha}{\tau}]}{\xi}{\omega} \subseteq
  \val{\sigma}{\xi'}{\omega'}$. Let $t \in
  \val{\sigma[\subst{\alpha}{\tau}]}{\xi}{\omega}$. We have $t :
  \omega(\sigma[\subst{\alpha}{\tau}])$, so $t : \omega'(\sigma)$. Let
  $u \in \val{\sigma_1}{\xi'}{\omega'}$. By the inductive hypothesis
  $u \in \val{\sigma_1[\subst{\alpha}{\tau}]}{\xi}{\omega}$. Hence $t
  u \in \val{\sigma_2[\subst{\alpha}{\tau}]}{\xi}{\omega} =
  \val{\sigma_2}{\xi'}{\omega'}$, where the last equality follows from
  the inductive hypothesis. Thus $t \in
  \val{\sigma}{\xi'}{\omega'}$. The other direction is analogous. The
  case $\sigma = \forall\alpha\sigma'$ is also analogous.

  Assume $\sigma = \varphi\psi$. We have
  $\val{\sigma[\subst{\alpha}{\tau}]}{\xi}{\omega} =
  \val{\varphi[\subst{\alpha}{\tau}]}{\xi}{\omega}(\omega(\psi[\subst{\alpha}{\tau}]),
  \val{\psi[\subst{\alpha}{\tau}]}{\xi}{\omega}) =
  \val{\varphi[\subst{\alpha}{\tau}]}{\xi}{\omega}(\omega'(\psi),
  \val{\psi[\subst{\alpha}{\tau}]}{\xi}{\omega}) =
  \val{\varphi}{\xi'}{\omega'}(\omega'(\psi),
  \val{\psi}{\xi'}{\omega'})$ where the last equality follows from the
  inductive hypothesis.

  Finally, assume $\sigma = \lambda(\beta:\kappa)\varphi$. Let $\psi
  \in \Tc_\kappa$ be closed and let $X \in \Cb_\psi$. We have
  $\val{\sigma[\subst{\alpha}{\tau}]}{\xi}{\omega}(\psi,X) =
  \val{\varphi[\subst{\alpha}{\tau}]}{\xi[\subst{\beta}{X}]}{\omega[\subst{\beta}{\tau}]}
  =
  \val{\varphi}{\xi'[\subst{\beta}{X}]}{\omega'[\subst{\beta}{\tau}]}
  = \val{\sigma}{\xi'}{\omega'}(\psi,X)$ where we use the inductive
  hypothesis in the penultimate equality.
\end{proof}

\begin{lemma}\label{lem_forall}
  Let $\tau$ be a type constructor of kind~$\kappa$. Assume $\omega$
  is closed for $\forall\alpha\sigma$ and for~$\tau$. If $t \in
  \val{\forall(\alpha:\kappa)\sigma}{\xi}{\omega}$ then $t
  (\omega(\tau)) \in \val{\sigma[\subst{\alpha}{\tau}]}{\xi}{\omega}$.
\end{lemma}

\begin{proof}
  By Lemma~\ref{lem_val_computable} we have~$\val{\tau}{\xi}{\omega}
  \in \Cb_{\omega(\tau)}$. So $t (\omega(\tau)) \in
  \val{\sigma}{\xi[\subst{\alpha}{\val{\tau}{\xi}{\omega}}]}{\omega[\subst{\alpha}{\omega(\tau)}]}$
  by $t \in \val{\forall(\alpha:\kappa)\sigma}{\xi}{\omega}$. Hence
  $t (\omega(\tau)) \in
  \val{\sigma[\subst{\alpha}{\tau}]}{\xi}{\omega}$ by
  Lemma~\ref{lem_val_subst}.
\end{proof}

\begin{lemma}\label{lem_beta_val}
  If $\omega$ is closed for~$\sigma,\sigma'$ and $\sigma =_\beta
  \sigma'$ then $\val{\sigma}{\xi}{\omega} =
  \val{\sigma'}{\xi}{\omega}$.
\end{lemma}

\begin{proof}
  It suffices to show the lemma for the case when~$\sigma$ is a
  $\beta$-redex. Then the general case follows by induction
  on~$\sigma$ and the length of reduction to a common reduct.

  So assume $(\lambda\alpha\tau)\sigma \to_\beta
  \tau[\subst{\alpha}{\sigma}]$. We have
  $\val{(\lambda\alpha\tau)\sigma}{\xi}{\omega} =
  \val{\lambda\alpha\tau}{\xi}{\omega}(\omega(\sigma),
  \val{\sigma}{\xi}{\omega}) =
  \val{\tau}{\xi[\subst{\alpha}{\val{\sigma}{\xi}{\omega}}]}{\omega[\subst{\alpha}{\omega(\sigma)}]}
  = \val{\tau[\subst{\alpha}{\sigma}]}{\xi}{\omega}$ where the last
  equality follows from Lemma~\ref{lem_val_subst}.
\end{proof}

A mapping~$\omega$ on type constructors is extended in the obvious way
to a mapping on terms. Note that $\omega$ also acts on the type
annotations of variable occurrences, e.g.~$\omega(\lambda x : \alpha
. x^\alpha) = \lambda x : \omega(\alpha) . x^{\omega(\alpha)}$.

\begin{lemma}\label{lem_typable_computable}
  If $t : \sigma$ and $\omega$ is closed for~$\sigma$ and
  $\FTV(\omega(t)) = \emptyset$ then $\omega(t) \in
  \val{\sigma}{\xi}{\omega}$.
\end{lemma}

\begin{proof}
  We prove by induction on the structure of~$t$ that if $t : \sigma$
  and $\omega$ is closed for~$\sigma$ and $\FTV(\omega(t)) =
  \emptyset$ and $x_1^{\tau_1},\ldots,x_n^{\tau_n}$ are all free
  variable occurrences in the canonical representative of~$t$ (so
  each~$\tau_i$ is $\beta$-normal), then for all
  $u_1\in\val{\tau_1}{\xi}{\omega},\ldots,u_n\in\val{\tau_n}{\xi}{\omega}$
  we have $\omega(t)[\subst{x_1}{u_1},\ldots,\subst{x_n}{u_n}] \in
  \val{\sigma}{\xi}{\omega}$. This suffices because
  $\omega(x_i^{\tau_i}) \in \val{\tau_i}{\xi}{\omega}$. Note that
  $\omega$ is closed for each~$\tau_i$ because $\FTV(\omega(t)) =
  \emptyset$ and~$t$ is typed, so no type constructor variable
  occurring free in~$\tau_i$ can be bound in~$t$ by a~$\Lambda$;
  e.g.~$\Lambda \alpha . x^\alpha$ is not a valid typed term (we
  assume~$\tau_i$ to be in $\beta$-normal form). For brevity, we use
  the notation $\omega^*(t) =
  \omega(t)[\subst{x_1}{u_1},\ldots,\subst{x_n}{u_n}]$. Note that
  $\omega^*(t) : \omega(\sigma)$.

  By the generation lemma for $t : \sigma$ there is a type~$\sigma'$
  such that $\sigma' =_\beta \sigma$ and $\FTV(\sigma') \subseteq
  \FTV(t)$ and one of the cases below holds. Note that~$\omega$ is
  closed for~$\sigma'$ because it is closed for~$\sigma$ and
  $\FTV(\omega(t)) = \emptyset$. By Lemma~\ref{lem_beta_val} it
  suffices to show $\omega^*(t) \in \val{\sigma'}{\xi}{\omega}$.
  \begin{itemize}
  \item If $t = x_1^{\sigma'}$ then $\omega(t)[\subst{x_1}{u_1}] =
    (x_1^{\omega(\sigma')})[\subst{x_1}{u_1}] = u_1 \in
    \val{\sigma'}{\xi}{\omega}$ by assumption.
  \item If $t = n$ is a natural number and $\sigma' = \nat$ then $t
    \in \val{\nat}{}{}$ by definition.
  \item If $t$ is a function symbol then the claim follows from
    Lemma~\ref{lem_circ}, Lemma~\ref{lem_lift} or
    Lemma~\ref{lem_flatten}.
  \item If $t = \abs{x:\sigma_1}{s}$ then
    $\sigma' = \sigma_1\arrtype\sigma_2$ and $s : \sigma_2$. Hence
    $\omega$ is closed for~$\sigma_2$. Let
    $u \in \val{\sigma_1}{\xi}{\omega}$. By the inductive hypothesis
    $\omega^*(s)[\subst{x}{u}] \in \val{\sigma_2}{\xi}{\omega}$. Hence
    $\omega^*(t) \in \val{\sigma'}{\xi}{\omega}$ by
    Lemma~\ref{lem_abstraction_computable}.
  \item If $t = \tabs{\alpha:\kappa}{s}$ then $\sigma' =
    \forall\alpha\tau$ and $s : \tau$. Let $\psi$ be a closed type
    constructor of kind~$\kappa$ and let $X \in \Cb_\psi$. Let
    $\omega_1 = \omega[\subst{\alpha}{\psi}]$ and
    $\xi_1=\xi[\subst{\alpha}{X}]$. Then $\omega_1$ is closed
    for~$\tau$ and $\FTV(\omega_1(s)) = \emptyset$. By the inductive
    hypothesis $\omega_1^*(s) \in \val{\tau}{\xi_1}{\omega_1}$. We
    have $\omega_1^*(s) = \omega^*(s)[\subst{\alpha}{\psi}]$ (assuming
    $\alpha$ chosen fresh such that $\omega(\alpha) = \alpha$). Hence
    $\omega^*(t) \in \val{\tau}{\xi}{\omega}$ by
    Lemma~\ref{lem_abstraction_computable}.
  \item If $t = t_1 t_2$ then $t_1 : \tau\arrtype\sigma'$ and $t_2 :
    \tau$ and $\FTV(\tau) \subseteq \FTV(t)$. Hence~$\omega$ is closed
    for~$\tau$ and for~$\tau\arrtype\sigma'$. By the inductive
    hypothesis $\omega^*(t_1) \in
    \val{\tau\arrtype\sigma'}{\xi}{\omega}$ and $\omega^*(t_2) \in
    \val{\tau}{\xi}{\omega}$. We have $\omega^*(t_2) :
    \omega(\tau)$. Then by definition $\omega^*(t) =
    (\omega^*(t_1))(\omega^*(t_2)) \in \val{\sigma'}{\xi}{\omega}$.
  \item If $t = s \psi$ then $s : \forall\alpha\tau$ and $\sigma' =
    \tau[\subst{\alpha}{\psi}]$. By the inductive hypothesis
    $\omega^*(s) \in \val{\forall\alpha\tau}{\xi}{\omega}$. Because
    $\FTV(\omega(t)) = \emptyset$, the mapping $\omega$ is closed
    for~$\psi$. So by Lemma~\ref{lem_forall} we have $\omega^*(t) =
    \omega^*(s) \omega(\psi) \in
    \val{\tau[\subst{\alpha}{\psi}]}{\xi}{\omega}$.\qedhere
  \end{itemize}
\end{proof}

{ \renewcommand{\thetheorem}{\ref{thm_sn}}
\begin{theorem}
  If $t : \sigma$ then $t \in \SN$.
\end{theorem}
\addtocounter{theorem}{-1}}

\begin{proof}
  For closed terms~$t$ and closed types~$\sigma$ this follows from
  Lemma~\ref{lem_typable_computable}, Lemma~\ref{lem_val_computable}
  and property~1 of candidates (Definition~\ref{def_candidate}). For
  arbitrary terms and types, this follows by closing the terms with an
  appropriate number of abstractions, and the types with corresponding
  $\forall$-quantifiers.
\end{proof}

{ \renewcommand{\thelemma}{\ref{lem_final_nat}}
\begin{lemma}
  The only final interpretation terms of type $\nat$ are the natural
  numbers.
\end{lemma}
\addtocounter{theorem}{-1}}

\begin{proof}
  We show by induction on~$t$ that if $t$ is a final interpretation
  term of type~$\nat$ then $t$ is a natural number. Because~$t$ is
  closed and in normal form, if it is not a natural number then it
  must have the form $\mathtt{f}_\sigma t_1 \ldots t_n$ for a function
  symbol $\mathtt{f}$. For concreteness assume $\mathtt{f} =
  \oplus$. Then $n \ge 2$. Because~$t$ is closed, $\sigma$ cannot be a
  type variable. It also cannot be an arrow or a $\forall$-type,
  because then $t$ would contain a redex. So $\sigma=\nat$. Then
  $t_1,t_2$ are final interpretation terms of type~$\nat$, hence
  natural numbers by the inductive hypothesis. But then $t$ contains a
  redex. Contradiction.
  The case for $\mathtt{f} = \otimes$ is parallel.  If
  $\mathtt{f} \in \{\flatten,\lift\}$ and $\sigma$ is closed, then
  $n \ge 1$ and in all cases $t$ is not in normal form.
\end{proof}

\subsection{Weak monotonicity proof}\label{sec_weakly_monotone_proof}

We want to show that if $s \succeq s'$ then $t[\subst{x}{s}] \succeq
t[\subst{x}{s'}]$. A straightforward proof attempt runs into a problem
that, because of impredicativity of polymorphism, direct induction on
type structure is not possible. We adopt a method similar to Girard's
method of candidates from the termination proof.

\begin{defn}\label{def_wm_candidate}
  By induction on the kind~$\kappa$ of a type constructor~$\tau$ we
  define the set~$\Cb_\tau$ of all candidates of type
  constructor~$\tau$.

  First assume $\kappa=*$, i.e., $\tau$ is a type. A set~$X$ of terms
  of type~$\tau$ equipped with a binary relation~$\ge^X$ is a
  \emph{candidate of type~$\tau$} if it satisfies the following
  properties:
  \begin{enumerate}
  \item if $t \in X$ and $t' : \tau$ and $t' \leadsto t$ then $t' \in
    X$,
  \item if $t_1,t_2 \in X$ then $\circ_\tau t_1 t_2 \in X$ for $\circ
    \in \{\oplus,\otimes\}$,
  \item if $t : \nat$ then $\lift_\tau t \in X$.
  \end{enumerate}
  and the relation~$\ge^X$ satisfies the following properties:
  \begin{enumerate}
  \item ${\succeq} \cap X \times X \subseteq {\ge^X}$,
  \item if $t_1 \ge^X t_2$ and $t_1' \leadsto t_1$ (resp.~$t_2'
    \leadsto t_2$) then $t_1' \ge^X t_2$ (resp.~$t_1 \ge^X t_2'$),
  \item if $t_1 \ge^X t_1'$ and $t_2 \ge^X t_2'$ then $\circ_\tau t_1
    t_2 \ge^X \circ_\tau t_1' t_2'$ for $\circ \in
    \{\oplus,\otimes\}$,
  \item if $t_1 \succeq_\nat t_2$ then $\lift_\tau(t_1) \ge^X
    \lift_\tau(t_2)$,
  \item if $t_1 \ge^X t_2$ then $\flatten_\tau(t_1) \succeq_\nat
    \flatten_\tau(t_2)$,
  \item $\ge^X$ is reflexive and transitive on~$X$.
  \end{enumerate}
  The relation~$\ge^X$ is a \emph{comparison candidate for~$X$},
  and~$X$ is a \emph{candidate set}.

  Now assume $\kappa = \kappa_1\arrkind\kappa_2$. A function $f :
  \Tc_{\kappa_1} \times \bigcup_{\xi\in\Tc_{\kappa_1}}\Cb_\xi \to
  \bigcup_{\xi\in\Tc_{\kappa_2}}\Cb_\xi$ is a \emph{candidate of type
    constructor~$\tau$} if for every closed type constructor~$\sigma$
  of kind~$\kappa_1$ and a candidate $X \in \Cb_\sigma$ we have
  $f(\sigma,X) \in \Cb_{\tau\sigma}$.
\end{defn}

\begin{lemma}\label{lem_beta_wm_candidate}
  If $\sigma =_\beta \sigma'$ then $\Cb_\sigma = \Cb_{\sigma'}$.
\end{lemma}

\begin{proof}
  Induction on the kind of~$\sigma$.
\end{proof}

\begin{defn}\label{def_wm_valuation}
  Let $\omega$ be a mapping from type constructor variables to type
  constructors (respecting kinds). The mapping~$\omega$ extends in an
  obvious way to a mapping from type constructors to type
  constructors. A mapping~$\omega$ is \emph{closed for~$\sigma$} if
  $\omega(\alpha)$ is closed for $\alpha \in \FTV(\sigma)$ (then
  $\omega(\sigma)$ is closed).

  An \emph{$\omega$-valuation} is a mapping~$\xi$ on type constructor
  variables such that $\xi(\alpha) \in \Cb_{\omega(\alpha)}$.

  For each type constructor~$\sigma$, each mapping~$\omega$ closed
  for~$\sigma$, and each $\omega$-valuation~$\xi$, we define
  $\val{\sigma}{\xi}{\omega}$ by induction on~$\sigma$:
  \begin{itemize}
  \item $\val{\alpha}{\xi}{\omega} = \xi(\alpha)$ for a type
    constructor variable~$\alpha$,
  \item $\val{\nat}{\xi}{\omega}$ is the set of all terms~$t \in
    \Iterms$ such that $t : \nat$; equipped with the relation
    $\gteq{\nat}{\xi}{\omega} = \succeq_\nat$,
  \item $\val{\sigma \arrtype \tau}{\xi}{\omega}$ is the set of all
    terms~$t$ such that $t : \omega(\sigma\arrtype\tau)$ and:
    \begin{itemize}
    \item for all $s \in \val{\sigma}{\xi}{\omega}$ we have
      $\app{t}{s} \in \val{\tau}{\xi}{\omega}$, and
    \item if $s_1 \gteq{\sigma}{\xi}{\omega} s_2$ then $\app{t}{s_1}
      \gteq{\tau}{\xi}{\omega} \app{t}{s_2}$;
    \end{itemize}
    equipped with the
    relation~$\gteq{\sigma\arrtype\tau}{\xi}{\omega}$ defined by:
    \begin{itemize}
    \item $t_1 \gteq{\sigma\arrtype\tau}{\xi}{\omega} t_2$ iff
      $t_1,t_2 \in \val{\sigma\arrtype\tau}{\xi}{\omega}$ and for
      every $s \in \val{\sigma}{\xi}{\omega}$ we have $t_1 s
      \gteq{\tau}{\xi}{\omega} t_2 s$,
    \end{itemize}
  \item $\val{\forall(\alpha:\kappa)[\sigma]}{\xi}{\omega}$ is the set
    of all terms~$t$ such that $t : \omega(\forall\alpha[\sigma])$ and:
    \begin{itemize}
    \item for every closed type constructor~$\varphi$ of kind~$\kappa$
      and every $X \in \Cb_\varphi$ we have $\tapp{t}{\varphi} \in
      \val{\sigma}{\xi[\subst{\alpha}{X}]}{\omega[\subst{\alpha}{\varphi}]}$;
    \end{itemize}
    equipped with the
    relation~$\gteq{\forall\alpha[\sigma]}{\xi}{\omega}$ defined by:
    \begin{itemize}
    \item $t_1 \gteq{\forall(\alpha:\kappa)[\sigma]}{\xi}{\omega} t_2$
      iff $t_1,t_2 \in
      \val{\forall(\alpha:\kappa)[\sigma]}{\xi}{\omega}$ and for every
      closed type constructor~$\varphi$ of kind~$\kappa$ and every $X
      \in \Cb_\varphi$ we have $t_1 \varphi
      \gteq{\sigma}{\xi[\subst{\alpha}{X}]}{\omega[\subst{\alpha}{\varphi}]}
      t_2 \varphi$,
    \end{itemize}
  \item
    $\val{\varphi \psi}{\xi}{\omega} =
    \val{\varphi}{\xi}{\omega}(\omega(\psi),\val{\psi}{\xi}{\omega})$,
  \item
    $\val{\lambda(\alpha:\kappa)\varphi}{\xi}{\omega}(\psi,X) =
    \val{\varphi}{\xi[\subst{\alpha}{X}]}{\omega[\subst{\alpha}{\psi}]}$
    for closed $\psi \in \Tc_\kappa$ and $X \in \Cb_\psi$.
  \end{itemize}
  In the above, if e.g.~$\val{\psi}{\xi}{\omega} \notin
  \Cb_{\omega(\psi)}$ then $\val{\varphi \psi}{\xi}{\omega}$ is
  undefined.
\end{defn}

Note that if $t \in \val{\sigma}{\xi}{\omega}$ then $t :
\omega(\sigma)$, and if $t_1 \gteq{\sigma}{\xi}{\omega} t_2$ then
$t_1,t_2\in\val{\sigma}{\xi}{\omega}$. For brevity we
use~$\val{\sigma}{\xi}{\omega}$ to denote both the pair
$(\val{\sigma}{\xi}{\omega},{\gteq{\sigma}{\xi}{\omega}})$ and its
first element, depending on the context. For a type~$\tau$,
by~$\gteq{\tau}{\xi}{\omega}$ we always denote the second element of
the pair~$\val{\tau}{\xi}{\omega}$. If $\tau$ is closed then~$\xi$
and~$\omega$ do not matter and we simply write~$\geq_\tau$
and~$\val{\tau}{}{}$.

\begin{lemma}\label{lem_val_wm_computable}
  If $\sigma$ is a type constructor, $\omega$ is closed for~$\sigma$,
  and $\xi$ is an $\omega$-valuation, then $\val{\sigma}{\xi}{\omega}
  \in \Cb_{\omega(\sigma)}$.
\end{lemma}

\begin{proof}
  Induction on~$\sigma$. If $\sigma=\alpha$ then $\xi(\alpha) \in
  \Cb_{\omega(\alpha)}$ by definition. If $\sigma=\nat$ then this
  follows from definitions.

  Assume $\sigma=\sigma_1\arrtype\sigma_2$. We check the properties of
  a candidate set.
  \begin{enumerate}
  \item The first property follows from the inductive hypothesis and
    property~2 of comparison candidates.
  \item Let $t_1,t_2 \in \val{\sigma}{\xi}{\omega}$. We need to show
    $\circ_\omega(\sigma) t_1 t_2 \in
    \val{\sigma_1\arrtype\sigma_2}{\xi}{\omega}$.

    Let $s \in \val{\sigma_1}{\xi}{\omega}$. Then
    $\circ_{\omega(\sigma)} t_1 t_2 s \leadsto
    \circ_{\omega(\sigma_2)} (t_1 s) (t_2 s)$. Because $t_i \in
    \val{\sigma_1\arrtype\sigma_2}{\xi}{\omega}$, we have $t_i s \in
    \val{\sigma_2}{\xi}{\omega}$. By the inductive
    hypothesis~$\val{\sigma_2}{\xi}{\omega} \in
    \Cb_{\omega(\sigma_2)}$, so $\circ_{\omega(\sigma_2)} (t_1 s) (t_2
    s) \in \val{\sigma_2}{\xi}{\omega}$. Hence
    $\circ_{\omega(\sigma_2)} t_1 t_2 s \in
    \val{\sigma_2}{\xi}{\omega}$ by property~1 of candidate sets.

    Let $s_1 \gteq{\sigma_1}{\xi}{\omega} s_2$. Then $s_i \in
    \val{\sigma_1}{\xi}{\omega}$. Because $t_j \in
    \val{\sigma_1\arrtype\sigma_2}{\xi}{\omega}$, we have $t_j s_i \in
    \val{\sigma_2}{\xi}{\omega}$ and $t_j s_1
    \gteq{\sigma_2}{\xi}{\omega} t_j s_2$. By the inductive
    hypothesis~$\gteq{\sigma_2}{\xi}{\omega}$ is a comparison
    candidate for~$\val{\sigma_2}{\xi}{\omega}$. Thus
    $\circ_{\omega(\sigma_2)} (t_1 s_1) (t_2 s_1)
    \gteq{\sigma_2}{\xi}{\omega} \circ_{\omega(\sigma_2)} (t_1 s_2)
    (t_2 s_2)$ by property~3 of comparison candidates. This suffices
    by property~2 of comparison candidates.
  \item Let $t : \nat$. Then $\lift_{\omega(\sigma)} t :
    \omega(\sigma)$.

    Let $s \in \val{\sigma_1}{\xi}{\omega}$. Then
    $\lift_{\omega(\sigma)}t s \leadsto \lift_{\omega(\sigma_2)}
    t$. By the inductive hypothesis $\lift_{\omega(\sigma_2)} t \in
    \val{\sigma_2}{\xi}{\omega}$. Hence $\lift_{\omega(\sigma)}t s \in
    \val{\sigma_2}{\xi}{\omega}$ by property~1 of candidate sets.

    Let $s_1,s_2 \in \val{\sigma_1}{\xi}{\omega}$. By the inductive
    hypothesis~$\gteq{\sigma_2}{\xi}{\omega}$ is a comparison
    candidate for~$\val{\sigma_2}{\xi}{\omega}$. We have
    $\lift_{\omega(\sigma_2)} t \gteq{\sigma_2}{\xi}{\omega}
    \lift_{\omega(\sigma_2)} t$ by the reflexivity
    of~$\gteq{\sigma_2}{\xi}{\omega}$ (property~6 of comparison
    candidates). This suffices by property~2 of comparison candidates,
    because $\lift_{\omega(\sigma)}t s_i \leadsto
    \lift_{\omega(\sigma_2)} t$.
  \end{enumerate}
  Now we check the properties of a comparison candidate
  for~$\val{\sigma_1\arrtype\sigma_2}{\xi}{\omega}$.
  \begin{enumerate}
  \item Suppose $t_1 \succeq t_2$ with $t_1,t_2 \in
    \val{\sigma}{\xi}{\omega}$. Let $s \in
    \val{\sigma_1}{\xi}{\omega}$. Then $t_1 s \succeq t_2 s$ by the
    definition of~$\succeq$. Hence $t_1 s \gteq{\sigma_2}{\xi}{\omega}
    t_2 s$ by the inductive hypothesis.
  \item Follows from the inductive hypothesis and the already shown
    property~1 of candidate sets
    for~$\val{\sigma_1\arrtype\sigma_2}{\xi}{\omega}$.
  \item Assume $t_i \gteq{\sigma}{\xi}{\omega} t_i'$. Let $s \in
    \val{\sigma_1}{\xi}{\omega}$. We have $\circ_{\omega(\sigma)} t_1
    t_2 s \leadsto \circ_{\omega(\sigma_2)} (t_1 s) (t_2 s)$ and
    $\circ_{\omega(\sigma)} t_1' t_2' s \leadsto
    \circ_{\omega(\sigma_2)} (t_1' s) (t_2' s)$. Since
    $t_i,t_i'\in\val{\sigma}{\xi}{\omega}$, we have $t_i s
    \gteq{\sigma_2}{\xi}{\omega} t_i' s$ and $t_i s, t_i' s \in
    \val{\sigma_2}{\xi}{\omega}$. By the inductive hypothesis $\circ
    (t_1 s) (t_2 s) \gteq{\sigma_2}{\xi}{\omega} \circ (t_1' s) (t_2'
    s)$, so $\circ t_1 t_2 s \gteq{\sigma_2}{\xi}{\omega} \circ t_1'
    t_2' s$ by property~2 of comparison candidates. This implies
    $\circ t_1 t_2 \gteq{\sigma}{\xi}{\omega} \circ t_1' t_2'$.
  \item Follows from Lemma~\ref{lem:liftgreater} and property~1 of
    comparison candidates.
  \item Assume $t_1 \gteq{\sigma}{\xi}{\omega} t_2$. Then
    $\flatten_{\omega(\sigma)} t_i \leadsto
    \flatten_{\omega(\sigma_2)} (t_i (\lift_{\omega(\sigma_1)}0))$. By
    the inductive hypothesis and property~3 of candidate sets
    $\lift_{\omega(\sigma_1)}0 \in \val{\sigma_1}{\xi}{\omega}$. Hence
    $t_i (\lift_{\omega(\sigma_1)}0) \in \val{\sigma_2}{\xi}{\omega}$
    and $t_1 (\lift_{\omega(\sigma_1)}0) \gteq{\sigma_2}{\xi}{\omega}
    t_2 (\lift_{\omega(\sigma_1)}0)$. Thus by the inductive hypothesis
    $\flatten_{\omega(\sigma_2)} (t_1 (\lift_{\omega(\sigma_1)}0))
    \succeq_\nat \flatten_{\omega(\sigma_2)} (t_2
    (\lift_{\omega(\sigma_1)}0))$. This implies
    $\flatten_{\omega(\sigma)} t_1 \succeq_\nat
    \flatten_{\omega(\sigma)} t_2$.
  \item Follows directly from the inductive hypothesis.
  \end{enumerate}

  If $\sigma=\forall\alpha\tau$ then the proof is analogous to the
  case $\sigma=\sigma_1\arrtype\sigma_2$. If $\sigma=\varphi\psi$ or
  $\sigma=\lambda(\alpha:\kappa)\varphi$ then the claim follows from
  the inductive hypothesis and Lemma~\ref{lem_beta_wm_candidate}, like
  in the proof of Lemma~\ref{lem_val_computable}.
\end{proof}

\begin{lemma}\label{lem_wm_circ}
  $\circ \in \val{\forall \alpha . \alpha \arrtype \alpha \arrtype
    \alpha}{}{}$ for $\circ \in \{ \oplus, \otimes \}$.
\end{lemma}

\begin{proof}
  Let~$\tau$ be a closed type and let $X \in \Cb_{\tau}$. Let
  $\omega(\alpha) = \tau$ and $\xi(\alpha) = X$.

  Let $t_1,t_2 \in \val{\alpha}{\xi}{\omega} = X$. Then $\circ_{\tau}
  t_1 t_2 \in \val{\alpha}{\xi}{\omega}$ by property~2 of candidate
  sets.

  Let $t_2' \in \val{\alpha}{\xi}{\omega}$ be such that $t_2
  \gteq{\alpha}{\xi}{\omega} t_2'$, i.e., $t_2 \ge^X t_2'$. By
  properties~6 and~3 of comparison candidates we have we have
  $\circ_{\tau} t_1 t_2 \gteq{\alpha}{\xi}{\omega} \circ_{\tau} t_1
  t_2'$. This shows $\circ_{\tau} t_1 \in
  \val{\alpha\arrtype\alpha}{\xi}{\omega}$.

  Let $t_1' \in \val{\alpha}{\xi}{\omega}$ be such that $t_1
  \gteq{\alpha}{\xi}{\omega} t_1'$. Let $u \in
  \val{\alpha}{\xi}{\omega}$. By properties~6 and~3 of comparison
  candidates we have $\circ_{\tau} t_1 u \gteq{\alpha}{\xi}{\omega}
  \circ_{\tau} t_1' u$. Hence $\circ_{\tau} t_1
  \gteq{\alpha\arrtype\alpha}{\xi}{\omega} \circ_{\tau} t_1'$. This
  shows $\circ_{\tau} \in
  \val{\alpha\arrtype\alpha\arrtype\alpha}{\xi}{\omega}$.
\end{proof}

\begin{lemma}\label{lem_wm_lift}
  $\lift \in \val{\forall\alpha.\nat\arrtype\alpha}{}{}$.
\end{lemma}

\begin{proof}
  Let~$\tau$ be a closed type and let $X \in \Cb_{\tau}$. Let
  $\omega(\alpha) = \tau$ and $\xi(\alpha) = X$. By property~4 of
  comparison candidates we have $\lift_{\tau}s_1
  \gteq{\alpha}{\xi}{\omega} \lift_{\tau}s_2$ for all $s_i : \nat$
  with $s_1 \succeq_\nat s_2$. It remains to show that $\lift_{\tau}s
  \in \val{\alpha}{\xi}{\omega} = X$ for all $s : \nat$. This follows
  from property~3 of candidate sets.
\end{proof}

\begin{lemma}\label{lem_wm_flatten}
  $\flatten \in \val{\forall\alpha.\alpha\arrtype\nat}{}{}$.
\end{lemma}

\begin{proof}
  Follows from definitions and property~5 of comparison candidates.
\end{proof}

\begin{lemma}\label{lem_val_subst_wm}
  For any type constructors~$\sigma,\tau$ with $\alpha \notin
  \FTV(\tau)$, a mapping~$\omega$ closed for~$\sigma$ and for~$\tau$,
  and an $\omega$-valuation~$\xi$, we have:
  \[
  \val{\sigma[\subst{\alpha}{\tau}]}{\xi}{\omega} =
  \val{\sigma}{\xi[\subst{\alpha}{\val{\tau}{\xi}{\omega}}]}{\omega[\subst{\alpha}{\omega(\tau)}]}.
  \]
\end{lemma}

\begin{proof}
  Let~$\omega' = \omega[\subst{\alpha}{\omega(\tau)}]$ and $\xi' =
  \xi[\subst{\alpha}{\val{\tau}{\xi}{\omega}}]$. The proof by
  induction on~$\sigma$ is analogous to the proof of
  Lemma~\ref{lem_val_subst}. The main difference is that in the case
  $\sigma = \sigma_1\arrtype\sigma_2$ we need to show that if e.g.~$t
  \in \val{\sigma[\subst{\alpha}{\tau}]}{\xi}{\omega}$ and $s_1
  \gteq{\sigma_1}{\xi'}{\omega'} s_2$ then $t s_1
  \gteq{\sigma_2}{\xi'}{\omega'} t s_2$. But then $s_1
  \gteq{\sigma_1[\subst{\alpha}{\tau}]}{\xi}{\omega} s_2$ by the
  inductive hypothesis, so $t s_1
  \gteq{\sigma_2[\subst{\alpha}{\tau}]}{\xi}{\omega} t s_2$ by
  definition. Hence $t s_1 \gteq{\sigma_2}{\xi'}{\omega'} t s_2$ by
  the inductive hypothesis.
\end{proof}

\begin{lemma}\label{lem_wm_forall}
  Let $\tau$ be a type constructor of kind~$\kappa$. Assume $\omega$
  is closed for $\forall\alpha[\sigma]$ and for~$\tau$.
  \begin{enumerate}
  \item If $t \in \val{\forall(\alpha:\kappa)[\sigma]}{\xi}{\omega}$
    then $t (\omega(\tau)) \in
    \val{\sigma[\subst{\alpha}{\tau}]}{\xi}{\omega}$.
  \item If $t_1 \gteq{\forall(\alpha:\kappa)[\sigma]}{\xi}{\omega}
    t_2$ then $t_1 (\omega(\tau))
    \gteq{\sigma[\subst{\alpha}{\tau}]}{\xi}{\omega} t_2
    (\omega(\tau))$.
  \end{enumerate}
\end{lemma}

\begin{proof}
  Analogous to the proof of Lemma~\ref{lem_forall}, using
  Lemma~\ref{lem_val_wm_computable} and Lemma~\ref{lem_val_subst_wm}.
\end{proof}

\begin{lemma}\label{lem_beta_val_wm}
  If $\omega$ is closed for~$\sigma,\sigma'$ and $\sigma =_\beta
  \sigma'$ then $\val{\sigma}{\xi}{\omega} =
  \val{\sigma'}{\xi}{\omega}$ and ${\gteq{\sigma}{\xi}{\omega}} =
      {\gteq{\sigma'}{\xi}{\omega}}$.
\end{lemma}

\begin{proof}
  Analogous to the proof of Lemma~\ref{lem_beta_val}, using
  Lemma~\ref{lem_val_subst_wm}.
\end{proof}

For two replacements $\delta_1 = \gamma_1 \circ \omega$ and $\delta_2
= \gamma_2 \circ \omega$ (see Definition~\ref{def_closure}) and an
$\omega$-valuation~$\xi$ we write $\delta_1 \gteq{\tau}{\xi}{\omega}
\delta_2$ iff $\delta_1(x) \gteq{\tau}{\xi}{\omega} \delta_2(x)$ for
each~$x : \tau$.

\begin{lemma}\label{lem_typable_wm_computable}
  Assume $t : \sigma$ and $\delta_1=\gamma_1\circ\omega$,
  $\delta_2=\gamma_2\circ\omega$ are replacements and~$\xi$ an
  $\omega$-valuation such that $\delta_1 \gteq{}{\xi}{\omega}
  \delta_2$ and $\omega$ is closed for~$\sigma$ and $\FTV(\omega(t)) =
  \emptyset$ and for all $x^\tau \in \FTV(t)$ we have $\delta_i(x) \in
  \val{\tau}{\xi}{\omega}$. Then $\delta_i(t) \in
  \val{\sigma}{\xi}{\omega}$ and $\delta_1(t)
  \gteq{\sigma}{\xi}{\omega} \delta_2(t)$.
\end{lemma}

\begin{proof}
  Induction on the structure of~$t$. By the generation lemma for $t :
  \sigma$ there is a type~$\sigma'$ such that $\sigma' =_\beta \sigma$
  and $\FTV(\sigma') \subseteq \FTV(t)$ and one of the cases below
  holds. Note that $\omega$ is closed for~$\sigma'$, because it is
  closed for~$\sigma$ and $\FTV(\omega(t)) = \emptyset$. Hence by
  Lemma~\ref{lem_beta_val_wm} it suffices to show $\delta_i(t) \in
  \val{\sigma'}{\xi}{\omega}$ and $\delta_1(t)
  \gteq{\sigma'}{\xi}{\omega} \delta_2(t)$.
  \begin{itemize}
  \item If $t = x^{\sigma'}$ then $\delta_i(t) \in
    \val{\sigma'}{\xi}{\omega}$ by assumption. Also $\delta_1(t)
    \gteq{\sigma'}{\xi}{\omega} \delta_2(t)$ by assumption.
  \item If $t = n$ is a natural number and $\sigma' = \nat$ then
    $\delta_i(t) = t$ and thus $t \in \val{\nat}{}{}$ and $\delta_1(t)
    \gteq{\nat}{\xi}{\omega} \delta_2(t)$ by definition and the
    reflexivity of~$\gteq{\nat}{\xi}{\omega}$.
  \item If $t$ is a function symbol then the claim follows from
    Lemma~\ref{lem_wm_circ}, Lemma~\ref{lem_wm_lift} or
    Lemma~\ref{lem_wm_flatten}, and the reflexivity
    of~$\gteq{}{\xi}{\omega}$.
  \item If $t = \abs{x:\sigma_1}{u}$ then $\sigma' =
    \sigma_1\arrtype\sigma_2$ and $u : \sigma_2$. Let $s \in
    \val{\sigma_1}{\xi}{\omega}$ and
    $\delta_i'=\delta_i[\subst{x}{s}]$. This is well-defined because
    $s : \omega(\sigma_1)$ and~$\omega(x)$ has
    type~$\omega(\sigma_1)$. We have $\delta_1' \gteq{}{\xi}{\omega}
    \delta_2'$ by the reflexivity of~$\gteq{\sigma_1}{\xi}{\omega}$
    (Lemma~\ref{lem_val_wm_computable} and property~6 of comparison
    candidates). Hence by the inductive hypothesis $\delta_i'(u) \in
    \val{\sigma_2}{\xi}{\omega}$. We have $\delta_i(\abs{x}{u}) s
    \leadsto \delta_i'(u)$, so $\delta_i(\abs{x}{u}) s \in
    \val{\sigma_2}{\xi}{\omega}$ by Lemma~\ref{lem_val_wm_computable}
    and property~1 of candidate sets.

    Let $s_1,s_2 \in \val{\sigma_1}{\xi}{\omega}$ be such that $s_1
    \gteq{\sigma_1}{\xi}{\omega} s_2$. Let
    $\delta_i'=\delta_i[\subst{x}{s_i}]$. We have $\delta_1
    \gteq{}{\xi}{\omega} \delta_2$. Hence by the inductive hypothesis
    $\delta_1'(u)\gteq{\sigma_2}{\xi}{\omega}\delta_2'(u)$. We have
    $\delta_i(\abs{x}{u}) s_i \leadsto \delta_i'(u)$. Thus
    $\delta_1(t) s_1 \gteq{\sigma_2}{\xi}{\omega} \delta_2(t) s_2$ by
    Lemma~\ref{lem_val_wm_computable} and property~2 of comparison
    candidates.

    Finally, we show $\delta_1(t)
    \gteq{\sigma_1\arrtype\sigma_2}{\xi}{\omega} \delta_2(t)$. Let $s
    \in \val{\sigma_1}{\xi}{\omega}$ and
    $\delta_i'=\delta_i[\subst{x}{s}]$. We have $\delta_1'
    \gteq{}{\xi}{\omega} \delta_2'$. By the inductive hypothesis
    $\delta_1'(u) \gteq{\sigma_2}{\xi}{\omega} \delta_2'(u)$. We have
    $\delta_i(\abs{x}{u}) s \leadsto \delta_i'(u)$. Thus $\delta_1(t)
    s \gteq{\sigma_2}{\xi}{\omega} \delta_2(t) s$ by
    Lemma~\ref{lem_val_wm_computable} and property~2 of comparison
    candidates.
  \item If $t = \tabs{\alpha:\kappa}{u}$ then $\sigma' =
    \forall\alpha[\tau]$ and $u : \tau$. Let $\psi$ be a closed type
    constructor of kind~$\kappa$ and let $X \in \Cb_\psi$. Let
    $\omega' = \omega[\subst{\alpha}{\psi}]$ and
    $\xi'=\xi[\subst{\alpha}{X}]$. Then $\omega'$ is closed for~$\tau$
    and $\FTV(\omega'(u)) = \emptyset$. Let
    $\delta_i'=\gamma_i\circ\omega'$. By the inductive hypothesis
    $\delta_i'(u) \in \val{\tau}{\xi'}{\omega'}$ and $\delta_1'(u)
    \gteq{\tau}{\xi'}{\omega'} \delta_2'(u)$. We have
    $\delta_i(\tabs{\alpha}{u}) \psi \leadsto \delta_i'(u)$. Hence
    $\delta_i(\tabs{\alpha}{u}) \psi \in \val{\tau}{\xi'}{\omega'}$ by
    Lemma~\ref{lem_val_wm_computable} and property~1 of candidate
    sets. Thus $\delta_i(\tabs{\alpha}{u}) \in
    \val{\forall\alpha[\tau]}{\xi}{\omega}$. Also
    $\delta_1(\tabs{\alpha}{u}) \psi \gteq{\tau}{\xi'}{\omega'}
    \delta_2(\tabs{\alpha}{u}) \psi$ by
    Lemma~\ref{lem_val_wm_computable} and property~2 of comparison
    candidates. Thus $\delta_1(\tabs{\alpha}{u})
    \gteq{\forall\alpha[\tau]}{\xi'}{\omega'}
    \delta_2(\tabs{\alpha}{u})$.
  \item If $t = t_1 t_2$ then $t_1 : \tau\arrtype\sigma'$ and $t_2 :
    \tau$ and $\FTV(\tau) \subseteq \FTV(t)$. Hence~$\omega$ is closed
    for~$\tau$ and for~$\tau\arrtype\sigma'$. By the inductive
    hypothesis $\delta_i(t_1) \in
    \val{\tau\arrtype\sigma'}{\xi}{\omega}$ and $\delta_i(t_2) \in
    \val{\tau}{\xi}{\omega}$ and $\delta_1(t_1)
    \gteq{\tau\arrtype\sigma'}{\xi}{\omega} \delta_2(t_1)$ and
    $\delta_1(t_2) \gteq{\tau}{\xi}{\omega} \delta_2(t_2)$. By the
    definition of $\val{\tau\arrtype\sigma'}{\xi}{\omega}$ we have
    $\delta_i(t) = \delta_i(t_1)\delta_i(t_2) \in
    \val{\sigma'}{\xi}{\omega}$, and $\delta_1(t_1)\delta_1(t_2)
    \gteq{\sigma'}{\xi}{\omega} \delta_1(t_1)\delta_2(t_2)$. By the
    definition of~$\gteq{\tau\arrtype\sigma'}{\xi}{\omega}$ we have
    $\delta_1(t_1)\delta_2(t_2)\gteq{\sigma'}{\xi}{\omega}\delta_2(t_1)\delta_2(t_2)$. Hence
    $\delta_1(t)\gteq{\sigma'}{\xi}{\omega}\delta_2(t)$ by the
    transitivity of~$\gteq{\sigma'}{\xi}{\omega}$.
  \item If $t = s \psi$ then $s : \forall\alpha[\tau]$ and $\sigma' =
    \tau[\subst{\alpha}{\psi}]$. By the inductive hypothesis
    $\delta_i(s) \in \val{\forall\alpha[\tau]}{\xi}{\omega}$ and
    $\delta_1(s) \gteq{\forall\alpha[\tau]}{\xi}{\omega}
    \delta_2(s)$. Because $\FTV(\omega(t)) = \emptyset$, the mapping
    $\omega$ is closed for~$\psi$. So by Lemma~\ref{lem_wm_forall} we
    have $\delta_i(t) = \delta_i(s) \omega(\psi) \in
    \val{\tau[\subst{\alpha}{\psi}]}{\xi}{\omega}$ and $\delta_1(t)
    \gteq{\tau[\subst{\alpha}{\psi}]}{\xi}{\omega} \delta_2(t)$.\qedhere
  \end{itemize}
\end{proof}

\begin{corollary}\label{cor_typable_wm_computable}
  If $t$ is closed and $t : \sigma$ then $t \in \val{\sigma}{}{}$.
\end{corollary}

\begin{lemma}\label{lem_gteq_to_succeq}
  If $\sigma$ is a closed type and $t_1 \geq_\sigma t_2$ then
  $t_1 \succeq_{\sigma} t_2$.
\end{lemma}

\begin{proof}
  By coinduction. By Lemma~\ref{lem_beta_val_wm} we may assume
  that~$\sigma$ is in $\beta$-normal form. The case $\sigma=\alpha$ is
  impossible because~$\sigma$ is closed. If $\sigma = \nat$ then
  ${\geq_\nat} = {\succeq_\nat}$.

  Assume $\sigma=\sigma_1\arrtype\sigma_2$. Let $u : \sigma_1$ be
  closed. By Corollary~\ref{cor_typable_wm_computable} we have $u \in
  \val{\sigma_1}{}{}$. Hence $t_1 u \geq_{\sigma_2} t_2 u$. By the
  coinductive hypothesis $t_1 u \succeq_{\sigma_2} t_2 u$. This
  implies $t_1 \succeq_{\sigma} t_2$.

  Assume $\sigma=\forall(\alpha:\kappa)\tau$. Let $\varphi$ be a
  closed type constructor of kind~$\kappa$. By
  Lemma~\ref{lem_val_wm_computable} we have $\val{\varphi}{}{} \in
  \Cb_\varphi$. By the definition of~$\geq_{\forall\alpha\tau}$ and
  Lemma~\ref{lem_val_subst_wm} we have $t_1 \varphi
  \geq_{\tau[\subst{\alpha}{\varphi}]} t_2 \varphi$. Note that
  $\tau[\subst{\alpha}{\varphi}]$ is still closed. Hence by the
  coinductive hypothesis $t_1 \varphi
  \succeq_{\tau[\subst{\alpha}{\varphi}]} t_2 \varphi$. This implies
  $t_1 \succeq_{\sigma} t_2$.
\end{proof}

\begin{corollary}\label{cor_gteq_succeq}
  If~$\sigma$ is a closed type then ${\geq_{\sigma}} =
  {\succeq_\sigma}$.
\end{corollary}

\begin{proof}
  Follows from Lemma~\ref{lem_gteq_to_succeq},
  Lemma~\ref{lem_val_wm_computable} and property~1 of comparison
  candidates.
\end{proof}

{ \renewcommand{\thelemma}{\ref{lem_succeq_subst}}
\begin{lemma}[Weak monotonicity]
  If $s \succeq_\sigma s'$ then $t[\subst{x}{s}] \succeq_\tau t[\subst{x}{s'}]$.
\end{lemma}
\addtocounter{theorem}{-1}}

\begin{proof}
  It suffices to show this when
  $s,s',t[\subst{x}{s}],t[\subst{x}{s'}]$ and $\sigma,\tau$ are all
  closed. Assume $s \succeq_\sigma s'$. Then $s \geq_{\sigma} s'$ by
  Corollary~\ref{cor_gteq_succeq}. Thus $t[\subst{x}{s}] \geq_{\tau}
  t[\subst{x}{s'}]$ follows from
  Lemma~\ref{lem_typable_wm_computable}. Hence $t[\subst{x}{s}]
  \succeq_\tau t[\subst{x}{s'}]$ by Corollary~\ref{cor_gteq_succeq}.
\end{proof}

\subsection{Proofs for Section~\ref{sec_rule_removal}}\label{app_rule_removal}

{ \renewcommand{\thelemma}{\ref{lem:plusparts}}
\begin{lemma}
For all types $\sigma$, terms $s,t$ of type $\sigma$ and natural
numbers $n > 0$:
\begin{enumerate}
\item $s \oplus_{\sigma} t \succeq s$ and $s \oplus_{\sigma} t \succeq
  t$;
\item $s \oplus_{\sigma} (\lift_{\sigma} n) \succ s$ and
  $(\lift_{\sigma} n) \oplus_{\sigma} t \succ t$.
\end{enumerate}
\end{lemma}
\addtocounter{theorem}{-1}}

\begin{proof}
  It suffices to prove this for closed $s,t$ and closed $\sigma$ in
  $\beta$-normal form.
  \begin{enumerate}
  \item By coinduction we show $(s \oplus t) u_1 \ldots u_m
    \succeq_\sigma s u_1 \ldots u_m$ for closed $u_1,\ldots,u_m$. The
    second case is similar.

    First note that $(s \oplus t) u_1 \ldots u_m \leadsto^* s u_1
    \ldots u_m \oplus t u_1 \ldots u_m$.

    If $\sigma = \nat$ then $((s \oplus t) u_1 \ldots u_m)\da = (s u_1
    \ldots u_m)\da + (t u_1 \ldots u_m)\da \ge (s u_1 \ldots
    u_m)\da$. Hence $(s \oplus t) u_1 \ldots u_m) \succeq_\nat s u_1
    \ldots u_m$.

    If $\sigma = \tau_1\arrtype\tau_2$ then by the coinductive
    hypothesis $(s \oplus t) u_1 \ldots u_m q \succeq_{\tau_2} s u_1
    \ldots u_m q$ for any $q \in \Iterms_{\tau_1}^f$. Hence $(s \oplus
    t) u_1 \ldots u_m \succeq_\sigma s u_1 \ldots u_m$.

    If $\sigma = \forall(\alpha:\kappa)[\tau]$ then by the coinductive
    hypothesis $(s \oplus t) u_1 \ldots u_m \xi \succeq_{\sigma'} s
    u_1 \ldots u_m \xi$ for any closed $\xi \in \Tc_\kappa$, where
    $\sigma' = \tau[\subst{\alpha}{\xi}]$. Hence $(s \oplus t) u_1
    \ldots u_m \succeq_\sigma s u_1 \ldots u_m$.
  \item By coinduction we show $(s \oplus (\lift n)) u_1 \ldots u_m
    \succeq_\sigma s u_1 \ldots u_m$ for closed $u_1,\ldots,u_m$. The
    second case is similar.

    Note that $(s \oplus (\lift n)) u_1 \ldots u_m \leadsto^* s u_1
    \ldots u_m \oplus n$. From this the case $\sigma=\nat$
    follows. The other cases follow from the coinductive hypothesis,
    like in the first point above.\qedhere
  \end{enumerate}
\end{proof}

{ \renewcommand{\thelemma}{\ref{lem:approxproperties}}
\begin{lemma}
For all types $\sigma$ and all terms $s,t,u$ of type $\sigma$, we
have:
\begin{enumerate}
\item $s \oplus_\sigma t \approx t \oplus_\sigma s$ and $s
  \otimes_\sigma t \approx t \otimes_\sigma s$;
\item $s \oplus_\sigma (t \oplus_\sigma u) \approx (s \oplus_\sigma t)
  \oplus_\sigma u$ and $s \otimes_\sigma (t \otimes_\sigma u) \approx
  (s \otimes_\sigma t) \otimes_\sigma u$;
\item $s \otimes_\sigma (t \oplus_\sigma u) \approx (s \otimes_\sigma
  t) \oplus_\sigma (s \otimes_\sigma u)$;
\item $(\lift_\sigma 0) \oplus_\sigma s \approx s$ and $(\lift_\sigma
  1) \otimes_\sigma s \approx s$.
\end{enumerate}
\end{lemma}
\addtocounter{theorem}{-1}}

\begin{proof}
  The proof is again analogous to the proof of
  Lemma~\ref{lem:plusparts}. For instance, for closed $s,t$ and closed
  $\sigma$ in $\beta$-normal form, we show by coinduction that $(s
  \oplus t) w_1 \ldots w_n \succeq (t \oplus s) w_1 \ldots w_n$ for
  closed $w_1,\ldots,w_n$ (and then the same with $\preceq$).
\end{proof}

{ \renewcommand{\thelemma}{\ref{lem_lift_approx}}
\begin{lemma}
  \begin{enumerate}
  \item $\lift_\sigma(n+m) \approx_\sigma (\lift_\sigma n)
    \oplus_\sigma (\lift_\sigma n)$.
  \item $\lift_\sigma(n m) \approx_\sigma (\lift_\sigma n)
    \otimes_\sigma (\lift_\sigma n)$.
  \item $\flatten_\sigma(\lift_\sigma(n)) \approx n$.
  \end{enumerate}
\end{lemma}
\addtocounter{theorem}{-1}}

\begin{proof}
  It suffices to show this for closed~$\sigma$ in $\beta$-normal
  form. For the first two points, one proves by induction on~$\sigma$
  that
  $(\lift_\sigma(n+m))\da = (\lift_\sigma n \oplus_\sigma \lift_\sigma
  n)\da$ (analogously for multiplication). This suffices by
  Corollary~\ref{cor_succ_da} and the reflexivity of~$\approx$.

  For the third point, one proceeds by induction on~$\sigma$. For
  example, if $\sigma = \sigma_1\to\sigma_2$ then
  $\flatten_\sigma(\lift_\sigma(n)) \leadsto^*
  \flatten_{\sigma_2}((\lambda x . \lift_{\sigma_2} n)
  (\lift_{\sigma_1} 0)) \leadsto \flatten_{\sigma_2} (\lift_{\sigma_1}
  n)$. Then the claim follows from the inductive hypothesis and
  Lemma~\ref{lem_succ_red}.
\end{proof}

\clearpage
\section{Proving the inequalities in \refsec{examples}}\label{app_ineqs}

The system IPC2 can be seen as a PFS with the following type constructors:
\[
\begin{array}{c}
\Sigma^T_\kappa = \{\quad
  \bot : *,\quad
  \mathtt{or} : * \arrkind * \arrkind *,\quad
  \mathtt{and} : * \arrkind * \arrkind *,\quad
  \exists : (* \arrkind *) \arrkind *
  \}
\end{array}
\]
We also have the following function symbols:
\[
\begin{array}{rclcrcl}
@ & : & \forall \alpha \forall \beta . (\alpha \arrtype \beta) \arrtype \alpha \arrtype \beta &
\quad &
\epsilon & : & \forall \alpha . \bot \arrtype \alpha \\

\mathtt{tapp} & : & \forall \alpha : * \arrkind * . \forall \beta .
  (\forall \beta [\alpha \beta]) \arrtype \alpha \beta &
\quad &
\proj^1 & : & \forall \alpha \forall \beta . \mathtt{and}\, \alpha\, \beta \arrtype \alpha \\

\mathtt{pair} & : & \forall \alpha \forall \beta . \alpha \arrtype \beta \arrtype
  \mathtt{and}\, \alpha\, \beta &
\quad &
\proj^2 & : & \forall \alpha \forall \beta . \mathtt{and}\, \alpha\, \beta \arrtype \beta \\

\mathtt{case} & : & \forall \alpha \forall \beta \forall \gamma . \mathtt{or}\, \alpha\, \beta \arrtype
  (\alpha \arrtype \gamma) \arrtype (\beta \arrtype \gamma) \arrtype \gamma &
\quad &
\mathtt{in}^1 & : & \forall \alpha \forall \beta . \alpha \arrtype
  \mathtt{or}\, \alpha\, \beta \\

\mathtt{let} & : & \forall \alpha : * \arrkind * . \forall \beta .
  (\exists (\alpha)) \arrtype
  (\forall \gamma . \alpha \gamma \arrtype \beta) \arrtype \beta &
\quad &
\mathtt{in}^2 & : & \forall \alpha \forall \beta . \beta \arrtype
  \mathtt{or}\, \alpha\, \beta \\

\mathtt{ext} & : & \forall \alpha : * \arrkind * . \forall \beta . \alpha \beta \arrtype
  \exists (\alpha)
\end{array}
\]

The following are the \emph{core} rules ($\beta$-reductions):
\[
\begin{array}{rclrcl}
@_{\sigma,\tau}(\abs{x}{s},t) & \red & s[x:=t] &
\mathtt{case}_{\sigma,\tau,\rho}(\mathtt{in}^1_{\sigma,\tau}(u),
  \abs{x}{s},\abs{y}{t}) & \red & s[x:=u] \\

\mathtt{tapp}_{\abs{\alpha}{\sigma},\tau}(\tabs{\alpha}{s}) & \red &
  s[\alpha:=\tau] &
\mathtt{case}_{\sigma,\tau,\rho}(\mathtt{in}^2_{\sigma,\tau}(u),
  \abs{x}{s},\abs{y}{t}) & \red & t[x:=u] \\

\proj^1_{\sigma,\tau}(\mathtt{pair}_{\sigma,\tau}(s,t)) & \red & s &
\mathtt{let}_{\varphi,\rho}(\mathtt{ext}_{\varphi,\tau}(s),
  \tabs{\alpha}{\abs{x:\varphi \alpha}{t}}) & \red & t[\alpha:=\tau][x:=s] \\

\proj^2_{\sigma,\tau}(\mathtt{pair}_{\sigma,\tau}(s,t)) & \red & t \\
\end{array}
\]

Then the next rules simplify proofs from contradiction:
\[
\begin{array}{rclrcl}
\epsilon_\tau(\epsilon_\bot(s)) & \red & \epsilon_\tau(s) \\
\proj^1_{\sigma,\tau}(\epsilon_{\mathtt{and}\,\sigma\,\tau}(s)) & \red &
  \epsilon_\sigma(s) \\
\proj^2_{\sigma,\tau}(\epsilon_{\mathtt{and}\,\sigma\,\tau}(s)) & \red &
  \epsilon_\tau(s) \\
@_{\sigma,\tau}(\epsilon_{\sigma \arrtype \tau}(s),t) & \red &
  \epsilon_\tau(s) \\
\mathtt{tapp}_{\varphi,\tau}(
  \epsilon_{\quant{\alpha}{\varphi\alpha}}(s)) & \red &
  \epsilon_{\varphi\tau}(s) \\
\mathtt{case}_{\sigma,\tau,\rho}(\epsilon_{\mathtt{or}\,\sigma\,\tau}(
  u),\abs{x:\sigma}{s},\abs{y:\tau}{t}) & \red & \epsilon_\rho(u) \\
\mathtt{let}_{\varphi,\rho}(\epsilon_{\exists(\varphi)}(s),\tabs{\alpha}{\abs{x:\varphi\alpha}{t}}) & \red &
  \epsilon_\rho(s) \\
\end{array}
\]

When a $\mathtt{case}$ occurs in a first argument, then it is shifted
to the root of the term.
\[
\begin{array}{l}
\epsilon_\rho(\mathtt{case}_{\sigma,\tau,\bot}(u,\abs{x:\sigma}{s},
  \abs{y:\tau}{t})) \red
  \mathtt{case}_{\sigma,\tau,\rho}(u,\abs{x:\sigma}{\epsilon_\rho(s)},
  \abs{y:\tau}{\epsilon_\rho(t)}) \\
@_{\rho,\pi}(\mathtt{case}_{\sigma,\tau,\rho \arrtype \pi}(u,
  \abs{x:\sigma}{s},\abs{y:\tau}{t}),v) \red
  \mathtt{case}_{\sigma,\tau,\pi}(u,
  \abs{x:\sigma}{@_{\rho,\pi}(s,v)},\abs{y:\tau}{@_{\rho,\pi}(t,v)}) \\
\mathtt{tapp}_{\varphi,\pi}(\mathtt{case}_{\sigma,\tau,
  \quant{\alpha}{\varphi\alpha}}(u,\abs{x:\sigma}{s},\abs{y:\tau}{t}))
  \red
  \mathtt{case}_{\sigma,\tau,\varphi\pi}(u,
  \abs{x:\sigma}{\mathtt{tapp}_{\varphi,\pi}(s)},
  \abs{y:\tau}{\mathtt{tapp}_{\varphi,\pi}(t)}) \\
\proj^1_{\rho,\pi}(\mathtt{case}_{\sigma,\tau,\mathtt{and}\,\rho\,\pi}(u,
  \abs{x:\sigma}{s},\abs{y:\tau}{t})) \red
  \mathtt{case}_{\sigma,\tau,\rho}(u,\abs{x:\sigma}{\proj^1_{\rho,\pi}(s)},
  \abs{y:\tau}{\proj^1_{\rho,\pi}(t)}) \\
\proj^2_{\rho,\pi}(\mathtt{case}_{\sigma,\tau,\mathtt{and}\,\rho,\pi}(u,
  \abs{x:\sigma}{s},\abs{y:\tau}{t}))\red
  \mathtt{case}_{\sigma,\tau,\pi}(u,\abs{x:\sigma}{\proj^2_{\rho,\pi}(s)},
  \abs{y:\tau}{\proj^2_{\rho,\pi}(t)}) \\
\mathtt{case}_{\rho,\pi,\xi}(\mathtt{case}_{\sigma,\tau,\mathtt{or}\,
  \rho\,\pi}(u,\abs{x:\sigma}{s},\abs{y:\tau}{t}),\abs{z:\rho}{v},
  \abs{a:\pi}{w}) \red \\
\phantom{AB}
  \mathtt{case}_{\sigma,\tau,\xi}(u,\abs{x:\sigma}{
    \mathtt{case}_{\rho,\pi,\xi}(s,\abs{z:\rho}{v},\abs{a:\pi}{w})},
    \abs{y:\tau}{\mathtt{case}_{\rho,\pi,\xi}(t,\abs{z:\rho}{v},
    \abs{a:\pi}{w})}) \\
\mathtt{let}_{\varphi,\rho}(
  \mathtt{case}_{\sigma,\tau,\exists\varphi}(
  u,\abs{x:\sigma}{s},\abs{y:\tau}{t}),v) \red \\
\phantom{AB}
  \mathtt{case}_{\sigma,\tau,\rho}(u,
  \abs{x:\sigma}{\mathtt{let}_{\varphi,\rho}(s,v)},
  \abs{y:\tau}{\mathtt{let}_{\varphi,\rho}(t,v)}) \\
\end{array}
\]

And the same happens for the $\mathtt{let}$:
\[
\begin{array}{l}
\epsilon_\tau(\mathtt{let}_{\varphi,\bot}(s,\tabs{\alpha}{
  \abs{x:\varphi\alpha}{t}})) \red
  \mathtt{let}_{\varphi,\tau}(s,\tabs{\alpha}{\abs{x:\varphi\alpha}{
  \epsilon_\tau(t)}}) \\
@_{\tau,\rho}(\mathtt{let}_{\varphi, \tau \arrtype
  \rho}(s,\tabs{\alpha}{\abs{x:\varphi\alpha}{t}}),u) \red
  \mathtt{let}_{\varphi,\rho}(s,\tabs{\alpha}{\abs{x:\varphi\alpha}{
  @_{\tau,\rho}(t,u)}}) \\
\mathtt{tapp}_{\psi,\rho}(\mathtt{let}_{\varphi,
  \forall\beta[\psi\beta]}(s,\tabs{\alpha}{\abs{x:\varphi\alpha}{t}}))
  \red
  \mathtt{let}_{\varphi,\psi\rho}(s,\tabs{\alpha}{\abs{x:\varphi\alpha}{
  \mathtt{tapp}_{\psi,\rho}(t)}}) \\
\proj^1_{\tau,\rho}(\mathtt{let}_{\varphi,
  \mathtt{and}\,\tau,\rho}(s,\tabs{\alpha}{\abs{x:\varphi\alpha}{t}}))
  \red
  \mathtt{let}_{\varphi,\tau}(s,\tabs{\alpha}{\abs{x:\varphi\alpha}{
  \proj^1_{\tau,\rho}(t)}}) \\
\proj^2_{\tau,\rho}(\mathtt{let}_{\varphi,
  \mathtt{and}\,\tau\,\rho}(s,\tabs{\alpha}{\abs{x:\varphi\alpha}{t}}))
  \red
  \mathtt{let}_{\varphi,\rho}(s,\tabs{\alpha}{\abs{x:\varphi\alpha}{
  \proj^2_{\tau,\rho}(t)}}) \\
\mathtt{case}_{\tau,\rho,\pi}(
  \mathtt{let}_{\varphi,\mathtt{or}\,\tau\,\rho}(s,\tabs{\alpha}{
  \abs{x:\varphi\alpha}{t}}),\abs{x:\tau}{u},\abs{y:\rho}{v})
  \red \\
\phantom{AB}
  \mathtt{let}_{\varphi,\pi}(s,\tabs{\alpha}{\abs{x:\varphi\alpha}{
  \mathtt{case}_{\tau,\rho,\pi}(t,\abs{x:\tau}{u},\abs{y:\rho}{v})}}) \\
\mathtt{let}_{\psi,\rho}(\mathtt{let}_{\varphi,\exists\psi}(s,
  \tabs{\alpha}{\abs{x:\varphi\alpha}{t}}),u) \red
  \mathtt{let}_{\varphi,\rho}(s,\tabs{\alpha}{\abs{x:\varphi\alpha}{
  \mathtt{let}_{\psi,\rho}(t,u)}})
  \phantom{ABCDEFGHIJ}
\end{array}
\]

It is this last group of rules that is not oriented by our method.
For all other rules $\ell \red r$ we have $\interpret{\ell} \succ
\interpret{r}$, as demonstrated below.

We will use the fact that $\beta$-reduction provides the derived
reduction rules for $\pi^i$ and $\mathtt{let}$.

\begin{lemma}\label{lem:encodings_reduce}
  $\pi^i(\pair{t_1}{t_2}) \leadsto_\beta^* t_i$ and
  $\xlet{}{\expair{\tau}{t}}{\alpha,x}{s} \leadsto_\beta^*
  s[\subst{\alpha}{\tau}][\subst{x}{t}]$.
\end{lemma}

In the proofs below, we will often use that $\lift(n) \otimes s \oplus
t \succeq s$ if $n \geq 1$, which holds because $\lift(n) \otimes s
\oplus t \approx \lift(1) \otimes s \oplus (\lift(n-1) \otimes s
\oplus t) \approx s \oplus (\lift(n-1) \otimes s \oplus t) \succeq
s$, using the calculation rules. Having this, the core rules and the
contradiction simplifications are all quite easy due to the choice of
$\Termmap$:
\begin{itemize}
\item $\interpret{@_{\sigma,\tau}(\abs{x}{s},t)} \succ
  \interpret{s[x:=t]}$ \\ We have
  $\interpret{@_{\sigma,\tau}(\abs{x:\sigma}{s},t)} \arrrbeta
  \lift_{\typeinterpret{\tau}}(2) \otimes
  (\ (\abs{x:\typeinterpret{\sigma}}{\interpret{s}}) \cdot
  \interpret{t}\ ) \oplus
  \lift_{\typeinterpret{\tau}}(\langle\text{something}\rangle \oplus
  1) \leadsto
  \lift(2) \otimes \interpret{s}[x:=\interpret{t}] \oplus
  \lift(\langle\text{something}\rangle \oplus 1) \succ
  \interpret{s}[x:=\interpret{t}]$, which equals $\interpret{s[x:=t]}$
  by Lemma \ref{lem:substitutioninterpret}.
%  Since, by Lemma \ref{lem_succ_red} $\leadsto \cdot \succ$ is
%  contained in $\succ$ ,we have the required orientation of the rule.
\item $\interpret{\mathtt{tapp}_{\abs{\alpha}{\sigma},\tau}(
  \tabs{\alpha}{s})} \succ \interpret{s[\alpha:=\tau]}$ \\ We have
  $\interpret{\mathtt{tapp}_{\lambda \alpha.\sigma,\tau}(
  \tabs{\alpha}{s})} \arrrbeta \lift_{(\abs{\alpha}{
  \typeinterpret{\sigma}})\interpret{\tau}}(2) \otimes (
  (\tabs{\alpha}{\interpret{s}}) * \beta) \oplus
  \lift_{(\abs{\alpha}{\typeinterpret{\sigma}})\interpret{\tau}}(1)
  \leadsto
  \lift(2) \otimes \interpret{s}[\alpha:=\typeinterpret{\tau}] \oplus
  \lift(1) \succ \interpret{s}[\alpha:=\typeinterpret{\tau}] =
  \interpret{s[\alpha:=\tau]}$, using Lemma
  \ref{lem:substitutioninterpret}.
\item $\interpret{\proj^1_{\sigma,\tau}(\mathtt{pair}_{\sigma,\tau}(s,t))}
  \succ \interpret{s}$ \\ We have
  $\interpret{\mathtt{pair}_{\sigma,\tau}(s,t)} \arrrbeta
  \pair{\interpret{s}}{\interpret{t}} \oplus \lift_{\typeinterpret{
  \sigma} \times \typeinterpret{\tau}}(\flatten_{\typeinterpret{
  \sigma}}(\interpret{s}) \oplus \flatten_{\typeinterpret{\tau}}(
  \interpret{t}))$ and therefore
  $\interpret{\proj^1_{\alpha,\beta}(\mathtt{pair}_{\alpha,\beta}(s,t))}
  \arrrbeta \lift_{\typeinterpret{\sigma}}(2) \otimes
  \pi^1(\pair{\interpret{s}}{\interpret{t}} \oplus \langle
  \text{something}\rangle) \oplus \lift_{\typeinterpret{\sigma}}(1)
  \succeq \pi^1(\pair{\interpret{s}}{\interpret{t}}) \oplus
  \lift_{\typeinterpret{\sigma}}(1)$,
  which $\succ \interpret{s}$ by Lemma \ref{lem:encodings_reduce}.
\item $\interpret{\proj^2_{\sigma,\tau}(\mathtt{pair}_{\sigma,\tau}(s,t))}
  \succ \interpret{t}$ \\ Analogous to the inequality above.
\item $\interpret{\mathtt{case}_{\sigma,\tau,\rho}(\mathtt{in}^1_{
  \sigma,\tau}(u),\abs{x}{s},\abs{y}{t})} \succ \interpret{s[x:=u]}$ \\
  Write $A := \lift_{\typeinterpret{\sigma} \times \typeinterpret{
  \tau}}(\flatten_{\typeinterpret{\sigma}}(\interpret{u}))$;
  then $\interpret{\mathtt{in}^1_{\sigma,\tau}(u)} =
  \pair{\interpret{u}}{\lift_{\typeinterpret{\tau}}(1)} \oplus A$.
  Let $B := \flatten_{\typeinterpret{\sigma} \times
  \typeinterpret{\tau}}(\pair{\interpret{u}}{\lift_{\typeinterpret{
  \tau}}(1)} \oplus A)$ and $C := \interpret{\abs{y}{t}} \cdot \pi^2(
  \pair{\interpret{u}}{\lift_{\typeinterpret{\tau}}(1)} \oplus A)$.
  Then we can write:
  $\interpret{\mathtt{case}_{\sigma,\tau,\rho}(\mathtt{in}^1_{
  \sigma,\tau}(u),\abs{x}{s},\abs{y}{t})} = \lift_{\typeinterpret{
  \rho}}(2) \oplus \lift_{\typeinterpret{\rho}}(3 \otimes B) \oplus
  \lift_{\typeinterpret{\rho}}(B \oplus 1) \otimes (\
  \interpret{\abs{x}{s}} \cdot \pi^1(\pair{\interpret{u}}{
  \lift_{\typeinterpret{\tau}}(1)} \oplus A) \oplus C\ )$.
  By splitting additive terms, distribution, neutrality of 1 and
  absolute positiveness, this $\succ
  \interpret{\abs{x}{s}} \cdot \pi^1(\pair{\interpret{u}}{
  \lift_\tau(1)}) \leadsto^*
  \interpret{\abs{x}{s}} \cdot \interpret{u}$ (by
  Lemma \ref{lem:encodings_reduce}), $= (\abs{x}{\interpret{s}}) \cdot
  \interpret{u} \arrrbeta \interpret{s}[x:=\interpret{u}] =
  \interpret{s[x:=u]}$ by Lemma \ref{lem:substitutioninterpret}.
\item $\interpret{\mathtt{case}_{\sigma,\tau,\rho}(\mathtt{in}^2_{
  \sigma,\tau}(u),\abs{x}{s},\abs{y}{t})} \succ \interpret{s[x:=u]}$. \\
  Analogous to the inequality above.
\item $\interpret{\mathtt{let}_{\varphi,\rho}(
  \mathtt{ext}_{\varphi,\tau}(s),
  \tabs{\alpha}{\abs{x:\varphi\alpha}{t}})} \succ
  \interpret{t[\alpha:=\tau][x:=s]}$. \\ We have
  $\interpret{\mathtt{ext}_{\varphi,\tau}(s)} \succeq
  \expair{\typeinterpret{\tau}}{\interpret{s}}$ by absolute positiveness.
  Therefore, using monotonicity,
  $\interpret{\mathtt{let}_{\varphi,\rho}(\mathtt{ext}_{\varphi,\tau}(
  s),\tabs{\alpha}{\abs{x:\typeinterpret{\varphi} \alpha}{t}})} \succeq
  \lift_{\typeinterpret{\rho}}(2) \otimes (\xlet{\typeinterpret{\rho}}{
    \expair{\typeinterpret{\tau}}{\interpret{s}}
  }{
    \expair{\alpha}{x}
  }{
    \interpret{\tabs{\alpha}{\abs{x:\varphi \alpha}{t}}} * \alpha \cdot
    x
  }) \oplus
  \langle\text{something}\rangle \oplus
  \lift_{\typeinterpret{\rho}}(1)$.
  Again by absolute positiveness, this $\succ
  \xlet{\typeinterpret{\rho}}{\expair{\typeinterpret{\tau}}{
  \interpret{s}}}{\expair{\alpha}{x}}{\interpret{\tabs{\alpha}{
  \abs{x:\typeinterpret{\varphi}\alpha}{t}}} * \alpha \cdot x} \leadsto
  \xlet{\typeinterpret{\rho}}{\expair{\typeinterpret{\tau}}{
  \interpret{s}}}{\expair{\alpha}{x}}{\interpret{t}}$.
  By Lemma \ref{lem:encodings_reduce}, this term
  $\succeq \interpret{t}[\alpha:=\typeinterpret{\tau}][x:=\interpret{s}]$.
  We complete by Lemma
  \ref{lem:substitutioninterpret}.
\end{itemize}

\begin{itemize}
\item $\interpret{\epsilon_\tau(\epsilon_\bot(s))} \succ
  \interpret{\epsilon_\tau(s)}$. \\
  We have $\interpret{\epsilon_\tau(\epsilon_\bot(s))} =
  \lift_{\typeinterpret{\tau}}(2 \otimes \lift_\nat(2 \otimes
  \interpret{s} \oplus 1) \oplus 1) \approx
  \lift_{\typeinterpret{\tau}}(4 \otimes \interpret{s} \oplus 3)
  \succ \lift_{\typeinterpret{\tau}}(2 \otimes \interpret{s} \oplus 1) =
  \interpret{\epsilon_\tau(s)}$.
\item $\interpret{@_{\sigma,\tau}(\epsilon_{\sigma \arrtype \tau}(s),
  t)} \succ \interpret{\epsilon_\tau(s)}$. \\
  We have
  $\interpret{@_{\sigma,\tau}(\epsilon_{\sigma \arrtype \tau}(s),t)}
  = \lift_{\typeinterpret{\tau}}(2) \otimes (\
  \lift_{\typeinterpret{\sigma} \arrtype \typeinterpret{\tau}}(
  2 \otimes \interpret{s} \oplus 1) \cdot \interpret{t}\ )\ \oplus \\
  \lift_{\typeinterpret{\tau}}(\langle\text{something}\rangle \oplus
  1) \succ \lift_{\typeinterpret{\sigma} \arrtype \typeinterpret{\tau}}(
  2 \otimes \interpret{s} \oplus 1) \cdot \interpret{t} \leadsto
  \lift_{\typeinterpret{\tau}}(2 \otimes \interpret{s} \oplus 1) =
  \interpret{\epsilon_\tau(s)}$.
\item $\interpret{\mathtt{tapp}_{\varphi,\tau}(
  \epsilon_{\quant{\alpha}{\varphi\alpha}}(s))} \succ
  \interpret{\epsilon_{\varphi\tau}(s)}$ \\
  We have $\interpret{\mathtt{tapp}_{\varphi,\tau}(
  \epsilon_{\quant{\alpha}{\varphi\alpha}}(s))} =
  \lift_{\typeinterpret{\varphi}\typeinterpret{\tau}}(2) \otimes (\
  \lift_{\quant{\alpha}{\typeinterpret{\varphi}\alpha}}(2 \otimes
  \interpret{s} \oplus 1) * \typeinterpret{\tau}\ ) \oplus \lift_{
  \typeinterpret{\varphi}\typeinterpret{\tau}}(1) \succ
  \lift_{\quant{\alpha}{\typeinterpret{\varphi}\alpha}}(2 \otimes
  \interpret{s} \oplus 1) * \typeinterpret{\tau} =
  (\tabs{\alpha}{\lift_{\typeinterpret{\varphi}\alpha}(2 \otimes
  \interpret{s} \oplus 1)}) * \typeinterpret{\tau} \leadsto
  \lift_{\typeinterpret{\varphi}\typeinterpret{\tau}}(2 \otimes
  \interpret{s} \oplus 1) =
  \lift_{\typeinterpret{\varphi\tau}}(2 \otimes \interpret{s} \oplus 1)
  = \interpret{\epsilon_{\varphi\tau}(s)}$
\item $\interpret{\proj^1_{\sigma,\tau}(\epsilon_{\mathtt{and}\,
  \sigma\,\tau}(s))} \succ \interpret{\epsilon_\sigma(s)}$ \\ We have
  $\interpret{\proj^1_{\sigma,\tau}(\epsilon_{\mathtt{and}\,
  \sigma\,\tau}(s))} = \lift_{\typeinterpret{\sigma}}(2) \otimes
  \pi^1(\lift_{\typeinterpret{\sigma} \times
  \typeinterpret{\tau}}(2 \otimes \interpret{s} \oplus 1)) \oplus
  \lift_{\typeinterpret{\sigma}}(1) \succ
  \pi^1(\lift_{\typeinterpret{\sigma} \times
  \typeinterpret{\tau}}(2 \otimes \interpret{s} \oplus 1)) =
  \lift_{\forall p.(\typeinterpret{\sigma} \arrtype
  \typeinterpret{\tau} \arrtype p) \arrtype p}(2 \otimes
  \interpret{s} \oplus 1)) * \typeinterpret{\sigma} \cdot
  (\lambda xy.x) =
  (\tabs{p}{\lambda f.\lift_p(2 \otimes \interpret{s} \oplus 1)}) *
  \typeinterpret{\sigma} \cdot (\lambda xy.x) \leadsto^*
  \lift_{\typeinterpret{\sigma}}(2 \otimes \interpret{s} \oplus 1) =
  \interpret{\epsilon_\sigma(s)}$.
\item $\interpret{\proj^2_{\sigma,\tau}(\epsilon_{\mathtt{and}\,
  \sigma\,\tau}(s))} \succ \interpret{\epsilon_\tau(s)}$ \\
  Analogous to the inequality above.
\item $\interpret{\mathtt{case}_{\sigma,\tau,\rho}(
  \epsilon_{\mathtt{or}\,\sigma\,\tau}(u),\abs{x:\sigma}{s},
  \abs{y:\tau}{t})} \succ \interpret{\epsilon_\rho(u)}$. \\
  We have $\interpret{\mathtt{case}_{\sigma,\tau,\rho}(
   \epsilon_{\mathtt{or}\,\sigma\,\tau}(u),\abs{x}{s},\abs{y}{t})} =\\
  \lift_{\typeinterpret{\rho}}(2) \oplus
  \lift_{\typeinterpret{\rho}}(3 \otimes
    \flatten_{\interpret{\sigma} \times
    \interpret{\tau}}(\lift_{\interpret{\sigma} \times
    \interpret{\tau}}(2 \otimes \interpret{u} \oplus 1))) \oplus
    \langle\text{something}\rangle \succ
  \lift_{\typeinterpret{\rho}}(3 \otimes \flatten_{
    \typeinterpret{\sigma} \times
    \interpret{\tau}}(\lift_{\interpret{\sigma} \times
    \interpret{\tau}}(2 \otimes \interpret{u} \oplus 1))) \succeq\\
  \lift_{\typeinterpret{\rho}}(\flatten_{\interpret{\sigma} \times
    \interpret{\tau}}(\lift_{\interpret{\sigma} \times
    \interpret{\tau}}(2 \otimes \interpret{u} \oplus 1))) \approx
  \lift_{\typeinterpret{\rho}}(2 \otimes \interpret{u} \oplus 1) =
  \interpret{\epsilon_\rho(u)}$ because $\flatten_\sigma(\lift_\sigma(
  n)) \approx n$ for all $\sigma,n$.
\item $\interpret{\mathtt{let}_{\varphi,\rho}(
  \epsilon_{\exists(\varphi)}(s),\tabs{\alpha}{\abs{x:\varphi\alpha}{
  t}})} \succ \interpret{\epsilon_\rho(s)}$. \\
  $\interpret{\mathtt{let}_{\varphi,\rho}(\epsilon_{
    \exists(\varphi)}(s),\tabs{\alpha}{\abs{x}{t}})} =
    \lift_{\typeinterpret{\rho}}(2) \otimes
    (\xlet{\typeinterpret{\rho}}{\lift_{\Sigma \alpha.\typeinterpret{
    \varphi}\alpha}(2 \otimes \interpret{s} \oplus 1)}{
    \expair{\alpha}{x}}{  \\
    (\tabs{\alpha}{\abs{x}{\interpret{t}}}) * \alpha \cdot x}) \oplus
    \langle\text{something}\rangle \oplus
    \lift_{\typeinterpret{\rho}}(1) \succ
    \xlet{\typeinterpret{\rho}}{\lift_{\Sigma \alpha.\typeinterpret{
    \varphi}\alpha}(2 \otimes \interpret{s} \oplus 1)}{
    \expair{\alpha}{x}}{  \\
    \interpret{t}} =
    \lift_{\forall p.(\forall \alpha.\typeinterpret{\varphi}\alpha
    \arrtype p) \arrtype p}(2 \otimes \interpret{s}
    \oplus 1) * \typeinterpret{\rho} \cdot (\tabs{\alpha}{\abs{x}{
    \interpret{t}}}) \leadsto^*
    \lift_{\typeinterpret{\rho}}(2 \otimes \interpret{s} \oplus 1) =
    \interpret{\epsilon_\rho(s)}$.
\end{itemize}

Unfortunately, the rules where $\mathtt{case}$ is shifted to the root
are rather more complicated, largely due to the variable
multiplication in $\Termmap(\mathtt{case})$ -- which we had to choose
because these rules may duplicate variables.

\begin{itemize}
\item $\interpret{\epsilon_\rho(\mathtt{case}_{\sigma,\tau,\bot}(u,
  \abs{x:\sigma}{s},\abs{y:\tau}{t}))} \succ
  \interpret{\mathtt{case}_{\sigma,\tau,\rho}(u,
  \abs{x:\sigma}{\epsilon_\rho(s)},\abs{y:\tau}{\epsilon_\rho(t)})}$ \\
  On the left-hand side, we have:
  \[
  \begin{array}{l}
  \interpret{\epsilon_\rho(\mathtt{case}_{\sigma,\tau,\bot}(u,
  \abs{x:\sigma}{s},\abs{y:\tau}{t}))} \approx \\
  \lift_{\typeinterpret{\rho}}(2 \otimes (
  2\ \oplus \\
  \phantom{ABCDEFG,}
  3 \otimes \flatten_{\typeinterpret{\sigma} \times
    \typeinterpret{\tau}}(\interpret{u})\ \oplus \\
  \phantom{ABCDEFG,}
  (\flatten_{\typeinterpret{\sigma} \times \typeinterpret{\tau}}(
    \interpret{u}) \oplus 1) \otimes (\
      \interpret{s}[x:=\pi^1(\interpret{u})] \oplus
      \interpret{t}[y:=\pi^2(\interpret{u})]\ )
  )\ \oplus \\
  \phantom{ABCD,} 1) \approx \\
  \lift_{\typeinterpret{\rho}}(
  1 \oplus 4\ \oplus \\
    \phantom{ABCDe}
    6 \otimes \flatten_{\typeinterpret{\sigma} \times
    \typeinterpret{\tau}}(\interpret{u})\ \oplus \\
    \phantom{ABCDe}
  (2 \otimes \flatten_{\typeinterpret{\sigma} \times \typeinterpret{\tau}}(
    \interpret{u}) \oplus 2) \otimes (\
      \interpret{s}[x:=\pi^1(\interpret{u})] \oplus
      \interpret{t}[y:=\pi^2(\interpret{u})]\ )) \approx \\
  \lift_{\typeinterpret{\rho}}(5)\ \oplus \\
  \phantom{A}
    \lift_{\typeinterpret{\rho}}(6 \otimes
      \flatten_{\typeinterpret{\sigma} \times
      \typeinterpret{\tau}}(\interpret{u}))\ \oplus \\
  \phantom{A}
    \lift_{\typeinterpret{\rho}}(\
      (\ 2 \otimes \flatten_{\typeinterpret{\sigma} \times
      \typeinterpret{\tau}}(\interpret{u}) \oplus 2\ ) \otimes
      (\ \interpret{s}[x:=\pi^1(\interpret{u})] \oplus
      \interpret{t}[y:=\pi^2(\interpret{u})]\ )\ )
  \end{array}
  \]

  On the right-hand side, we have:
  \[
  \begin{array}{l}
  \interpret{\mathtt{case}_{\sigma,\tau,\rho}(u,
  \abs{x:\sigma}{\epsilon_\rho(s)},\abs{y:\tau}{\epsilon_\rho(t)})}
  \approx \\
  \lift_{\typeinterpret{\rho}}(2)\ \oplus \\
  \phantom{A}
  \lift_{\typeinterpret{\rho}}(3 \otimes \flatten_{\typeinterpret{
    \sigma} \times \typeinterpret{\tau}}(\interpret{u}))\ \oplus \\
  \phantom{A}
  \lift_{\typeinterpret{\rho}}(\flatten_{\typeinterpret{\sigma} \times
    \typeinterpret{\tau}}(\interpret{u}) \oplus 1)\ \otimes \\
  \phantom{ABC}
    (\ \lift_{\typeinterpret{\rho}}(2 \otimes \interpret{s} \oplus 1)
      [x:=\pi^1(\interpret{u})]
      \oplus
     \lift_{\typeinterpret{\rho}}(2 \otimes \interpret{t} \oplus 1)
      [y:=\pi^2(\interpret{u})]
    \ ) \approx \\
  \lift_{\typeinterpret{\rho}}(2)\ \oplus \\
  \phantom{A}
  \lift_{\typeinterpret{\rho}}(3 \otimes \flatten_{\typeinterpret{
    \sigma} \times \typeinterpret{\tau}}(\interpret{u}))\ \oplus \\
  \phantom{A}
  \lift_{\typeinterpret{\rho}}(\ (\ \flatten_{\typeinterpret{\sigma} \times
    \typeinterpret{\tau}}(\interpret{u}) \oplus 1\ )\ \otimes \\
  \phantom{ABCDEF}
    (\ 2 \otimes \interpret{s}[x:=\pi^1(\interpret{u})] \oplus 1
       \oplus
       2 \otimes \interpret{t}[y:=\pi^2(\interpret{u})] \oplus 1
    \ )\ ) \approx \\
  \lift_{\typeinterpret{\rho}}(2)\ \oplus \\
  \phantom{A}
  \lift_{\typeinterpret{\rho}}(3 \otimes \flatten_{\typeinterpret{
    \sigma} \times \typeinterpret{\tau}}(\interpret{u}))\ \oplus \\
  \phantom{A}
  \lift_{\typeinterpret{\rho}}((\ 2 \otimes \flatten_{\typeinterpret{
    \sigma} \times \typeinterpret{\tau}}(\interpret{u}) \oplus 2\ )
    \otimes (\ \interpret{s}[x:=\pi^1(\interpret{u})] \oplus
    \interpret{t}[y:=\pi^2(\interpret{u})]\ )\ \oplus \\
  \phantom{A}
  \lift_{\typeinterpret{\rho}}((\ \flatten_{\typeinterpret{\sigma} \times
    \typeinterpret{\tau}}(\interpret{u}) \oplus 1\ ) \otimes
    (\ 1 \oplus 1\ )\ ) \approx \\
  \end{array}
  \]
  \[
  \begin{array}{l}
  \lift_{\typeinterpret{\rho}}(2)\ \oplus \\
  \phantom{A}
  \lift_{\typeinterpret{\rho}}(3 \otimes \flatten_{\typeinterpret{
    \sigma} \times \typeinterpret{\tau}}(\interpret{u}))\ \oplus \\
  \phantom{A}
  \lift_{\typeinterpret{\rho}}((\ 2 \otimes \flatten_{\typeinterpret{
    \sigma} \times \typeinterpret{\tau}}(\interpret{u}) \oplus 2\ )
    \otimes (\ \interpret{s}[x:=\pi^1(\interpret{u})] \oplus
    \interpret{t}[y:=\pi^2(\interpret{u})]\ ))\ \oplus \\
  \phantom{A}
  \lift_{\typeinterpret{\rho}}(2 \otimes \flatten_{\typeinterpret{\sigma}
    \times \typeinterpret{\tau}}(\interpret{u}))\ \oplus \\
  \phantom{A}
    \lift_{\typeinterpret{\rho}}(2) \approx \\
  \lift_{\typeinterpret{\rho}}(4)\ \oplus \\
  \phantom{A}
  \lift_{\typeinterpret{\rho}}(5 \otimes \flatten_{\typeinterpret{
    \sigma} \times \typeinterpret{\tau}}(\interpret{u}))\ \oplus \\
  \phantom{A}
  \lift_{\typeinterpret{\rho}}((\ 2 \otimes \flatten_{\typeinterpret{
    \sigma} \times \typeinterpret{\tau}}(\interpret{u}) \oplus 2\ )
    \otimes (\ \interpret{s}[x:=\pi^1(\interpret{u})] \oplus
    \interpret{t}[y:=\pi^2(\interpret{u})]\ ) )
  \end{array}
  \]
  By absolute positiveness, it is clear that the rule is oriented
  with $\succeq$.
\item $\interpret{@_{\rho,\pi}(\mathtt{case}_{\sigma,\tau,\rho
  \arrtype \pi}(u,\abs{x:\sigma}{s},\abs{y:\tau}{t}),v)} \succ
  \interpret{\mathtt{case}_{\sigma,\tau,\pi}(u,\abs{x:\sigma}{
  @_{\rho,\pi}(s,v)},\abs{y:\tau}{@_{\rho,\pi}(t,v)})}$ \\
  On the left-hand side, we have:
  \[
  \begin{array}{l}
  \interpret{@_{\rho,\pi}(\mathtt{case}_{\sigma,\tau,\rho \arrtype
  \pi}(u,\abs{x:\sigma}{s},\abs{y:\tau}{t}),v)} \approx \\
  \lift_{\typeinterpret{\pi}}(2) \otimes (\\
    \phantom{ABC}
    (\ \lift_{\typeinterpret{\rho} \arrtype \typeinterpret{\pi}}(2)
       \oplus \lift_{\typeinterpret{\rho} \arrtype
       \typeinterpret{\pi}}(3 \otimes \flatten_{\typeinterpret{\sigma}
       \times \typeinterpret{\tau}}(\interpret{u}))\ \oplus \\
    \phantom{ABCD}
      \lift_{\typeinterpret{\rho} \arrtype \typeinterpret{\pi}}(
      \flatten_{\typeinterpret{\sigma} \times \typeinterpret{\tau}}(
      \interpret{u}) \oplus 1) \otimes
      (\interpret{s}[x:=\pi^1(\interpret{u})] \oplus
       \interpret{t}[y:=\pi^2(\interpret{u})]) \\
    \phantom{ABC}
    ) \cdot \interpret{v} \\
    \phantom{A} ) \oplus \lift_{\typeinterpret{\pi}}(\\
    \phantom{ABC}\flatten_{\typeinterpret{\sigma}}(\interpret{v})\
      \oplus \\
    \phantom{ABC}\flatten_{\typeinterpret{\sigma} \arrtype
      \typeinterpret{\tau}}( \\
      \phantom{ABCDE}
       \lift_{\typeinterpret{\rho} \arrtype \typeinterpret{\pi}}(2)
       \oplus \lift_{\typeinterpret{\rho} \arrtype
       \typeinterpret{\pi}}(3 \otimes \flatten_{\typeinterpret{\sigma}
       \times \typeinterpret{\tau}}(\interpret{u}))\ \oplus \\
    \phantom{ABCDE}
      \lift_{\typeinterpret{\rho} \arrtype \typeinterpret{\pi}}(
      \flatten_{\typeinterpret{\sigma} \times \typeinterpret{\tau}}(
      \interpret{u}) \oplus 1) \otimes
      (\interpret{s}[x:=\pi^1(\interpret{u})] \oplus
       \interpret{t}[y:=\pi^2(\interpret{u})]) \\
    \phantom{ABC} ) \otimes
    \flatten_{\typeinterpret{\sigma}}(\interpret{v}) \oplus 1 \\
    \phantom{A} )
  \end{array}
  \]
  Using that for $\circ \in \{\oplus,\otimes\}$ we always have
  $(s \circ t) \cdot v \approx (s \cdot v) \circ (t \cdot v)$ as well
  as $\lift_{\alpha\arrtype \beta}(s) \cdot v \approx \lift_\beta(s)$,
  and that always $\flatten_\alpha(\lift_\alpha(s)) \approx s)$, this
  term $\approx$
  \[
  \begin{array}{l}
  \lift_{\typeinterpret{\pi}}(2) \otimes (\\
    \phantom{ABC}
    (\ \lift_{\typeinterpret{\pi}}(2) \oplus
       \lift_{\typeinterpret{\pi}}(3 \otimes
          \flatten_{\typeinterpret{\sigma} \times
          \typeinterpret{\tau}}(\interpret{u}))\ \oplus \\
    \phantom{ABCD}
      \lift_{\typeinterpret{\pi}}(
      \flatten_{\typeinterpret{\sigma} \times \typeinterpret{\tau}}(
      \interpret{u}) \oplus 1) \otimes
      (\interpret{s}[x:=\pi^1(\interpret{u})] \cdot \interpret{v} \oplus
       \interpret{t}[y:=\pi^2(\interpret{u})] \cdot \interpret{v}) \\
    \phantom{ABC}
    )\\
    \phantom{A} ) \oplus \lift_{\typeinterpret{\pi}}(\\
    \phantom{ABC}\flatten_{\typeinterpret{\sigma}}(\interpret{v})\
      \oplus \\
    \phantom{ABC}(\ 2 \oplus 3 \otimes \flatten_{\typeinterpret{\sigma}
      \times \typeinterpret{\tau}}(\interpret{u})\ \oplus \\
    \phantom{ABCD}(\
    \flatten_{\typeinterpret{\sigma} \times \typeinterpret{\tau}}(
    \interpret{u}) \oplus 1\ ) \otimes
      \flatten_{\typeinterpret{\sigma} \arrtype
      \typeinterpret{\tau}}(
      \interpret{s}[x:=\pi^1(\interpret{u})] \oplus
       \interpret{t}[y:=\pi^2(\interpret{u})]) \\
    \phantom{ABC} ) \otimes
    \flatten_{\typeinterpret{\sigma}}(\interpret{v}) \oplus 1 \\
    \phantom{A} ) \approx \\
  \lift_{\typeinterpret{\pi}}(4)\ \oplus \\
  \phantom{A}
  \lift_{\typeinterpret{\pi}}(6 \otimes
    \flatten_{\typeinterpret{\sigma} \times \typeinterpret{\tau}}(
    \interpret{u}))\ \oplus \\
  \phantom{A}
  \lift_{\typeinterpret{\pi}}(2 \otimes
    \flatten_{\typeinterpret{\sigma} \times \typeinterpret{\tau}}(
    \interpret{u})) \otimes
    (\interpret{s}[x:=\pi^1(\interpret{u})] \cdot \interpret{v} \oplus
     \interpret{t}[y:=\pi^2(\interpret{u})] \cdot \interpret{v})\
     \oplus \\
  \phantom{A}
  \lift_{\typeinterpret{\pi}}(2) \otimes
    (\interpret{s}[x:=\pi^1(\interpret{u})] \cdot \interpret{v} \oplus
     \interpret{t}[y:=\pi^2(\interpret{u})] \cdot \interpret{v})\
     \oplus \\
  \phantom{A}
  \lift_{\typeinterpret{\pi}}(\flatten_{\typeinterpret{
    \sigma}}(\interpret{v}))\ \oplus \\
  \phantom{A}
  \lift_{\typeinterpret{\pi}}(2 \otimes \flatten_{\typeinterpret{
    \sigma}}(\interpret{v}))\ \oplus \\
  \phantom{A}
  \lift_{\typeinterpret{\pi}}(3 \otimes \flatten_{\typeinterpret{
    \sigma} \times \typeinterpret{\tau}}(\interpret{u}) \otimes
    \flatten_{\typeinterpret{\sigma}}(\interpret{v}))\ \oplus \\
  \phantom{A}
  \lift_{\typeinterpret{\pi}}(\flatten_{\typeinterpret{\sigma} \times
    \typeinterpret{\tau}}(\interpret{u}) \otimes
    \flatten_{\typeinterpret{\sigma}}(\interpret{v})\ \otimes \\
  \phantom{ABC}
    \flatten_{\typeinterpret{\sigma} \arrtype \typeinterpret{\tau}}(
    \interpret{s}[x:=\pi^1(\interpret{u})] \oplus
    \interpret{t}[y:=\pi^2(\interpret{u})]))\ \oplus \\
  \phantom{A}
  \lift_{\typeinterpret{\pi}}(
    \flatten_{\typeinterpret{\sigma}}(\interpret{v}) \otimes
    \flatten_{\typeinterpret{\sigma} \arrtype \typeinterpret{\tau}}(
    (\interpret{s}[x:=\pi^1(\interpret{u})] \oplus
    \interpret{t}[y:=\pi^2(\interpret{u})])))\ \oplus \\
  \phantom{A}
  \lift_{\typeinterpret{\pi}}(1)\ \approx \\
  \end{array}
  \]
  \[
  \begin{array}{l}
  \lift_{\typeinterpret{\pi}}(5)\ \oplus \\
  \phantom{A}
  \lift_{\typeinterpret{\pi}}(6 \otimes
    \flatten_{\typeinterpret{\sigma} \times \typeinterpret{\tau}}(
    \interpret{u}))\ \oplus \\
  \phantom{A}
  \lift_{\typeinterpret{\pi}}(3 \otimes \flatten_{\typeinterpret{
    \sigma}}(\interpret{v}))\ \oplus \\
  \phantom{A}
  \lift_{\typeinterpret{\pi}}(3 \otimes \flatten_{\typeinterpret{
    \sigma} \times \typeinterpret{\tau}}(\interpret{u}) \otimes
    \flatten_{\typeinterpret{\sigma}}(\interpret{v}))\ \oplus \\
  \phantom{A}
  \lift_{\typeinterpret{\pi}}(
    \flatten_{\typeinterpret{\sigma}}(\interpret{v}) \otimes
    \flatten_{\typeinterpret{\sigma} \arrtype \typeinterpret{\tau}}(
    \interpret{s}[x:=\pi^1(\interpret{u})]))\ \oplus \\
  \phantom{A}
  \lift_{\typeinterpret{\pi}}(
    \flatten_{\typeinterpret{\sigma}}(\interpret{v}) \otimes
    \flatten_{\typeinterpret{\sigma} \arrtype \typeinterpret{
    \tau}}(\interpret{t}[y:=\pi^2(\interpret{u})]))\ \oplus \\
  \phantom{A}
  \lift_{\typeinterpret{\pi}}(\flatten_{\typeinterpret{\sigma} \times
    \typeinterpret{\tau}}(\interpret{u}) \otimes
    \flatten_{\typeinterpret{\sigma}}(\interpret{v}) \otimes
    \flatten_{\typeinterpret{\sigma} \arrtype \typeinterpret{\tau}}(
    \interpret{s}[x:=\pi^1(\interpret{u})]))\ \oplus \\
  \phantom{A}
  \lift_{\typeinterpret{\pi}}(\flatten_{\typeinterpret{\sigma} \times
    \typeinterpret{\tau}}(\interpret{u}) \otimes
    \flatten_{\typeinterpret{\sigma}}(\interpret{v}) \otimes
    \flatten_{\typeinterpret{\sigma} \arrtype \typeinterpret{\tau}}(
    \interpret{t}[y:=\pi^2(\interpret{u})]))\ \oplus \\
  \phantom{A}
  \lift_{\typeinterpret{\pi}}(2) \otimes
    \interpret{s}[x:=\pi^1(\interpret{u})] \cdot \interpret{v}\ \oplus \\
  \phantom{A}
  \lift_{\typeinterpret{\pi}}(2) \otimes
     \interpret{t}[y:=\pi^2(\interpret{u})] \cdot \interpret{v}\
     \oplus \\
  \phantom{A}
  \lift_{\typeinterpret{\pi}}(2 \otimes
    \flatten_{\typeinterpret{\sigma} \times \typeinterpret{\tau}}(
    \interpret{u})) \otimes
    \interpret{s}[x:=\pi^1(\interpret{u})] \cdot \interpret{v}\ \oplus \\
  \phantom{A}
  \lift_{\typeinterpret{\pi}}(2 \otimes
    \flatten_{\typeinterpret{\sigma} \times \typeinterpret{\tau}}(
    \interpret{u})) \otimes
     \interpret{t}[y:=\pi^2(\interpret{u})] \cdot \interpret{v} \\
  \end{array}
  \]
  And on the right-hand side, we have:
  \[
  \begin{array}{l}
  \interpret{\mathtt{case}_{\sigma,\tau,\pi}(u,\abs{x:\sigma}{
  @_{\rho,\pi}(s,v)},\abs{y:\tau}{@_{\rho,\pi}(t,v)})} \approx \\
  \lift_{\typeinterpret{\pi}}(2)\ \oplus \\
  \phantom{A}
  \lift_{\typeinterpret{\pi}}(3 \otimes \flatten_{\typeinterpret{\sigma}
    \times \typeinterpret{\tau}}(\interpret{u}))\ \oplus \\
  \phantom{A}
  \lift_{\typeinterpret{\pi}}(\flatten_{\typeinterpret{\sigma} \times
    \typeinterpret{\tau}}(\interpret{u}) \oplus 1)\ \otimes \\
  \phantom{ABC}
  (\ (\abs{x}{\interpret{@_{\rho,\pi}(s,v)}}) \cdot \pi^1(
      \interpret{u}) \oplus
     (\abs{y}{\interpret{@_{\rho,\pi}(t,v)}}) \cdot \pi^2(
      \interpret{u})\ ) \approx \\
  \lift_{\typeinterpret{\pi}}(2)\ \oplus \\
  \phantom{A}
  \lift_{\typeinterpret{\pi}}(3 \otimes \flatten_{\typeinterpret{\sigma}
    \times \typeinterpret{\tau}}(\interpret{u}))\ \oplus \\
  \phantom{A}
  \lift_{\typeinterpret{\pi}}(\flatten_{\typeinterpret{\sigma} \times
    \typeinterpret{\tau}}(\interpret{u}) \oplus 1)\ \otimes \\
  \phantom{ABC}
  (\ (\abs{x}{\lift_{\typeinterpret{\pi}}(2) \otimes
    \interpret{s} \cdot \interpret{v} \oplus
    \lift_{\typeinterpret{\pi}}(\flatten_{
    \typeinterpret{\rho}}(\interpret{v})\ \oplus \\
  \phantom{ABCDEF}\flatten_{\typeinterpret{\rho} \arrtype
    \typeinterpret{\pi}}(\interpret{s}) \otimes
    \flatten_{\typeinterpret{\rho}}(\interpret{v}) \oplus
    1)}) \cdot \pi^1(\interpret{u})\ \oplus \\
  \phantom{ABCD}
     (\abs{y}{\lift_{\typeinterpret{\pi}}(2) \otimes
     \interpret{t} \cdot \interpret{v} \oplus
    \lift_{\typeinterpret{\pi}}(
    \flatten_{\typeinterpret{\rho}}(\interpret{v})\ \oplus \\
  \phantom{ABCDEF} \flatten_{\typeinterpret{\rho} \arrtype
    \typeinterpret{\pi}}(\interpret{t}) \otimes
    \flatten_{\typeinterpret{\rho}}(\interpret{v}) \oplus
    1)}) \cdot \pi^2(\interpret{u}) \\
  \phantom{ABC}) \approx \\
  \lift_{\typeinterpret{\pi}}(2)\ \oplus \\
  \phantom{A}
  \lift_{\typeinterpret{\pi}}(3 \otimes \flatten_{\typeinterpret{\sigma}
    \times \typeinterpret{\tau}}(\interpret{u}))\ \oplus \\
  \phantom{A}
  \lift_{\typeinterpret{\pi}}(\flatten_{\typeinterpret{\sigma} \times
    \typeinterpret{\tau}}(\interpret{u}) \oplus 1)\ \otimes \\
  \phantom{ABC}
  (\ \lift_{\typeinterpret{\pi}}(2) \otimes \interpret{s}[x:=\pi^1(
    \interpret{u})] \cdot \interpret{v} \oplus
    \lift_{\typeinterpret{\pi}}(\flatten_{
    \typeinterpret{\rho}}(\interpret{v})\ \oplus \\
  \phantom{ABCDE}\flatten_{\typeinterpret{\rho} \arrtype
    \typeinterpret{\pi}}(\interpret{s}[x:=\pi^1(\interpret{u})]) \otimes
    \flatten_{\typeinterpret{\rho}}(\interpret{v}) \oplus
    1)\ \oplus \\
  \phantom{ABCD}
  \lift_{\typeinterpret{\pi}}(2) \otimes \interpret{t}[y:=\pi^2(
    \interpret{u})] \cdot \interpret{v} \oplus
    \lift_{\typeinterpret{\pi}}(
    \flatten_{\typeinterpret{\rho}}(\interpret{v})\ \oplus \\
  \phantom{ABCDEF} \flatten_{\typeinterpret{\rho} \arrtype
    \typeinterpret{\pi}}(\interpret{t}[y:=\pi^2(\interpret{u})]) \otimes
    \flatten_{\typeinterpret{\rho}}(\interpret{v}) \oplus
    1) \\
  \phantom{ABC}) \approx \\
  \lift_{\typeinterpret{\pi}}(2)\ \oplus \\
  \phantom{A}
  \lift_{\typeinterpret{\pi}}(3 \otimes \flatten_{\typeinterpret{\sigma}
    \times \typeinterpret{\tau}}(\interpret{u}))\ \oplus \\
  \phantom{A}
  \lift_{\typeinterpret{\pi}}(\flatten_{\typeinterpret{\sigma} \times
    \typeinterpret{\tau}}(\interpret{u}) \oplus 1)\ \otimes \\
  \phantom{ABC}
  (\ \lift_{\typeinterpret{\pi}}(2) \otimes \interpret{s}[x:=\pi^1(
    \interpret{u})] \cdot \interpret{v}\ \oplus \\
  \phantom{ABCD}
  \lift_{\typeinterpret{\pi}}(\flatten_{
    \typeinterpret{\rho}}(\interpret{v}))\ \oplus \\
  \phantom{ABCD}
  \lift_{\typeinterpret{\pi}}(
    \flatten_{\typeinterpret{\rho}}(\interpret{v}) \otimes
    \flatten_{\typeinterpret{\rho} \arrtype \typeinterpret{\pi}}(
    \interpret{s}[x:=\pi^1(\interpret{u})]))\ \oplus \\
  \phantom{ABCD}
  \lift_{\typeinterpret{\pi}}(1)\ \oplus \\
  \phantom{ABCD}
  \lift_{\typeinterpret{\pi}}(2) \otimes \interpret{t}[y:=\pi^2(
    \interpret{u})] \cdot \interpret{v}\ \oplus \\
  \phantom{ABCD}
  \lift_{\typeinterpret{\pi}}(
    \flatten_{\typeinterpret{\rho}}(\interpret{v}))\ \oplus \\
  \phantom{ABCD}
  \lift_{\typeinterpret{\pi}}(
    \flatten_{\typeinterpret{\rho}}(\interpret{v}) \otimes
    \flatten_{\typeinterpret{\rho} \arrtype
    \typeinterpret{\pi}}(\interpret{t}[y:=\pi^2(\interpret{u})]))\
    \oplus\\
  \phantom{ABCD}
  \lift_{\typeinterpret{\pi}}(1) \\
  \phantom{ABC}) \approx \\
  \end{array}
  \]
  \[
  \begin{array}{l}
  \lift_{\typeinterpret{\pi}}(2)\ \oplus \\
  \phantom{A}
  \lift_{\typeinterpret{\pi}}(3 \otimes \flatten_{\typeinterpret{\sigma}
    \times \typeinterpret{\tau}}(\interpret{u}))\ \oplus \\
  \phantom{A}
  \lift_{\typeinterpret{\pi}}(\flatten_{\typeinterpret{\sigma} \times
    \typeinterpret{\tau}}(\interpret{u}) \oplus 1)\ \otimes \\
  \phantom{ABC}
  (\ \lift_{\typeinterpret{\pi}}(2)\ \oplus \\
  \phantom{ABCD}
  \lift_{\typeinterpret{\pi}}(2 \otimes \flatten_{
    \typeinterpret{\rho}}(\interpret{v}))\ \oplus \\
  \phantom{ABCD}
  \lift_{\typeinterpret{\pi}}(2) \otimes \interpret{s}[x:=\pi^1(
    \interpret{u})] \cdot \interpret{v}\ \oplus \\
  \phantom{ABCD}
  \lift_{\typeinterpret{\pi}}(2) \otimes \interpret{t}[y:=\pi^2(
    \interpret{u})] \cdot \interpret{v}\ \oplus \\
  \phantom{ABCD}
  \lift_{\typeinterpret{\pi}}(
    \flatten_{\typeinterpret{\rho}}(\interpret{v}) \otimes
    \flatten_{\typeinterpret{\rho} \arrtype \typeinterpret{\pi}}(
    \interpret{s}[x:=\pi^1(\interpret{u})]))\ \oplus \\
  \phantom{ABCD}
  \lift_{\typeinterpret{\pi}}(
    \flatten_{\typeinterpret{\rho}}(\interpret{v}) \otimes
    \flatten_{\typeinterpret{\rho} \arrtype
    \typeinterpret{\pi}}(\interpret{t}[y:=\pi^2(\interpret{u})])) \\
  \phantom{ABC}) \approx \\
  \lift_{\typeinterpret{\pi}}(2)\ \oplus \\
  \phantom{A}
  \lift_{\typeinterpret{\pi}}(3 \otimes \flatten_{\typeinterpret{\sigma}
    \times \typeinterpret{\tau}}(\interpret{u}))\ \oplus \\
  \phantom{A}\lift_{\typeinterpret{\pi}}(2 \otimes
    \flatten_{\typeinterpret{\sigma} \times
    \typeinterpret{\tau}}(\interpret{u}))\ \oplus \\
  \phantom{A}
  \lift_{\typeinterpret{\pi}}(2 \otimes
    \flatten_{\typeinterpret{\sigma} \times
    \typeinterpret{\tau}}(\interpret{u}) \otimes
    \flatten_{\typeinterpret{\rho}}(\interpret{v}))\ \oplus \\
  \phantom{A}
  \lift_{\typeinterpret{\pi}}(2 \otimes
    \flatten_{\typeinterpret{\sigma} \times
    \typeinterpret{\tau}}(\interpret{u})) \otimes
    \interpret{s}[x:=\pi^1(\interpret{u})] \cdot \interpret{v}\ \oplus\\
  \phantom{A}
  \lift_{\typeinterpret{\pi}}(2 \otimes
    \flatten_{\typeinterpret{\sigma} \times
    \typeinterpret{\tau}}(\interpret{u})) \otimes
    \interpret{t}[y:=\pi^2(\interpret{u})] \cdot \interpret{v}\ \oplus \\
  \phantom{A}
  \lift_{\typeinterpret{\pi}}(\flatten_{\typeinterpret{\sigma} \times
    \typeinterpret{\tau}}(\interpret{u}) \otimes
    \flatten_{\typeinterpret{\rho}}(\interpret{v}) \otimes
    \flatten_{\typeinterpret{\rho} \arrtype \typeinterpret{\pi}}(
    \interpret{s}[x:=\pi^1(\interpret{u})]))\ \oplus \\
  \phantom{A}
  \lift_{\typeinterpret{\pi}}(\flatten_{\typeinterpret{\sigma} \times
    \typeinterpret{\tau}}(\interpret{u}) \otimes
    \flatten_{\typeinterpret{\rho}}(\interpret{v}) \otimes
    \flatten_{\typeinterpret{\rho} \arrtype \typeinterpret{\pi}}(
    \interpret{t}[y:=\pi^2(\interpret{u})]))\ \oplus \\
  \phantom{A}
  \lift_{\typeinterpret{\pi}}(2)\ \oplus \\
  \phantom{A}
  \lift_{\typeinterpret{\pi}}(2 \otimes \flatten_{
    \typeinterpret{\rho}}(\interpret{v}))\ \oplus \\
  \phantom{A}
  \lift_{\typeinterpret{\pi}}(2) \otimes \interpret{s}[x:=\pi^1(
    \interpret{u})] \cdot \interpret{v}\ \oplus \\
  \phantom{A}
  \lift_{\typeinterpret{\pi}}(2) \otimes \interpret{t}[y:=\pi^2(
    \interpret{u})] \cdot \interpret{v}\ \oplus \\
  \phantom{A}
  \lift_{\typeinterpret{\pi}}(
    \flatten_{\typeinterpret{\rho}}(\interpret{v}) \otimes
    \flatten_{\typeinterpret{\rho} \arrtype \typeinterpret{\pi}}(
    \interpret{s}[x:=\pi^1(\interpret{u})]))\ \oplus \\
  \phantom{A}
  \lift_{\typeinterpret{\pi}}(
    \flatten_{\typeinterpret{\rho}}(\interpret{v}) \otimes
    \flatten_{\typeinterpret{\rho} \arrtype
    \typeinterpret{\pi}}(\interpret{t}[y:=\pi^2(\interpret{u})])) \\
  \end{array}
  \]
  This we can reorder to:
  \[
  \begin{array}{l}
  \lift_{\typeinterpret{\pi}}(4)\ \oplus \\
  \phantom{A}
  \lift_{\typeinterpret{\pi}}(5 \otimes \flatten_{\typeinterpret{\sigma}
    \times \typeinterpret{\tau}}(\interpret{u}))\ \oplus \\
  \phantom{A}
  \lift_{\typeinterpret{\pi}}(2 \otimes \flatten_{
    \typeinterpret{\rho}}(\interpret{v}))\ \oplus \\
  \phantom{A}
  \lift_{\typeinterpret{\pi}}(2 \otimes
    \flatten_{\typeinterpret{\sigma} \times
    \typeinterpret{\tau}}(\interpret{u}) \otimes
    \flatten_{\typeinterpret{\rho}}(\interpret{v}))\ \oplus \\
  \phantom{A}
  \lift_{\typeinterpret{\pi}}(
    \flatten_{\typeinterpret{\rho}}(\interpret{v}) \otimes
    \flatten_{\typeinterpret{\rho} \arrtype \typeinterpret{\pi}}(
    \interpret{s}[x:=\pi^1(\interpret{u})]))\ \oplus \\
  \phantom{A}
  \lift_{\typeinterpret{\pi}}(
    \flatten_{\typeinterpret{\rho}}(\interpret{v}) \otimes
    \flatten_{\typeinterpret{\rho} \arrtype
    \typeinterpret{\pi}}(\interpret{t}[y:=\pi^2(\interpret{u})]))\
    \oplus \\
  \phantom{A}
  \lift_{\typeinterpret{\pi}}(\flatten_{\typeinterpret{\sigma} \times
    \typeinterpret{\tau}}(\interpret{u}) \otimes
    \flatten_{\typeinterpret{\rho}}(\interpret{v}) \otimes
    \flatten_{\typeinterpret{\rho} \arrtype \typeinterpret{\pi}}(
    \interpret{s}[x:=\pi^1(\interpret{u})]))\ \oplus \\
  \phantom{A}
  \lift_{\typeinterpret{\pi}}(\flatten_{\typeinterpret{\sigma} \times
    \typeinterpret{\tau}}(\interpret{u}) \otimes
    \flatten_{\typeinterpret{\rho}}(\interpret{v}) \otimes
    \flatten_{\typeinterpret{\rho} \arrtype \typeinterpret{\pi}}(
    \interpret{t}[y:=\pi^2(\interpret{u})]))\ \oplus \\
  \phantom{A}
  \lift_{\typeinterpret{\pi}}(2) \otimes \interpret{s}[x:=\pi^1(
    \interpret{u})] \cdot \interpret{v}\ \oplus \\
  \phantom{A}
  \lift_{\typeinterpret{\pi}}(2) \otimes \interpret{t}[y:=\pi^2(
    \interpret{u})] \cdot \interpret{v}\ \oplus \\
  \phantom{A}
  \lift_{\typeinterpret{\pi}}(2 \otimes
    \flatten_{\typeinterpret{\sigma} \times
    \typeinterpret{\tau}}(\interpret{u})) \otimes
    \interpret{s}[x:=\pi^1(\interpret{u})] \cdot \interpret{v}\ \oplus\\
  \phantom{A}
  \lift_{\typeinterpret{\pi}}(2 \otimes
    \flatten_{\typeinterpret{\sigma} \times
    \typeinterpret{\tau}}(\interpret{u})) \otimes
    \interpret{t}[y:=\pi^2(\interpret{u})] \cdot \interpret{v} \\
  \end{array}
  \]
  Using absolute positiveness, it is clear that the inequality is
  oriented.
\item $\interpret{\mathtt{tapp}_{\varphi,\pi}(\mathtt{case}_{\sigma,\tau,
  \quant{\alpha}{\varphi\alpha}}(u,\abs{x:\sigma}{s},\abs{y:\tau}{t}))}
  \succ \\ \interpret{\mathtt{case}_{\sigma,\tau,\varphi\pi}(u,
  \abs{x:\sigma}{\mathtt{tapp}_{\varphi,\pi}(s)},
  \abs{y:\tau}{\mathtt{tapp}_{\varphi,\pi}(t)})}$ \\
  On the left-hand side, we have
  \[
  \begin{array}{l}
  \interpret{\mathtt{tapp}_{\varphi,\pi}(\mathtt{case}_{\sigma,\tau,
  \quant{\alpha}{\varphi\alpha}}(u,\abs{x}{s},\abs{y}{t}))} \approx \\
  \lift_{\typeinterpret{\varphi\pi}}(2) \otimes ( \\
    \phantom{ABC}
    \lift_{\quant{\alpha}{\typeinterpret{\varphi}\alpha}}(2) \oplus
    \lift_{\quant{\alpha}{\typeinterpret{\varphi}\alpha}}(3 \otimes
      \flatten_{\typeinterpret{\sigma} \times \typeinterpret{\tau}}(
      \interpret{u})) \oplus \\
    \phantom{ABC}
    \lift_{\quant{\alpha}{\typeinterpret{\varphi}\alpha}}(
      \flatten_{\typeinterpret{\sigma} \times \typeinterpret{\tau}}(
      \interpret{u}) \oplus 1) \otimes (
      \interpret{s}[x:=\pi^1(\interpret{u})] \oplus
      \interpret{t}[y:=\pi^2(\interpret{u})]) \\
  \phantom{A}) * \typeinterpret{\pi} \oplus
  \lift_{\typeinterpret{\varphi\pi}}(1) \approx \\
  \end{array}
  \]
  \[
  \begin{array}{l}
  \lift_{\typeinterpret{\varphi\pi}}(4)\ \oplus \\
  \phantom{A}
  \lift_{\typeinterpret{\varphi\pi}}(6 \otimes
      \flatten_{\typeinterpret{\sigma} \times \typeinterpret{\tau}}(
      \interpret{u})) \oplus \\
  \phantom{A}
  \lift_{\typeinterpret{\varphi\pi}}(2) \otimes
    \interpret{s}[x:=\pi^1(\interpret{u})] * \typeinterpret{\pi}\ \oplus\\
  \phantom{A}
  \lift_{\typeinterpret{\varphi\pi}}(2) \otimes
    \interpret{t}[y:=\pi^2(\interpret{u})] * \typeinterpret{\pi}\ \oplus\\
  \phantom{A}
  \lift_{\typeinterpret{\varphi\pi}}(2 \otimes
    \flatten_{\typeinterpret{\sigma} \times \typeinterpret{\tau}}(
      \interpret{u})) \otimes \interpret{s}[x:=\pi^1(\interpret{u})]
      * \typeinterpret{\pi}\ \oplus \\
  \phantom{A}
  \lift_{\typeinterpret{\varphi\pi}}(2 \otimes
    \flatten_{\typeinterpret{\sigma} \times \typeinterpret{\tau}}(
      \interpret{u})) \otimes \interpret{t}[x:=\pi^2(\interpret{u})]
      * \typeinterpret{\pi}\ \oplus \\
  \phantom{A}
  \lift_{\typeinterpret{\varphi\pi}}(1) \approx \\
  \lift_{\typeinterpret{\varphi\pi}}(5)\ \oplus \\
  \phantom{A}
  \lift_{\typeinterpret{\varphi\pi}}(6 \otimes
      \flatten_{\typeinterpret{\sigma} \times \typeinterpret{\tau}}(
      \interpret{u})) \oplus \\
  \phantom{A}
  \lift_{\typeinterpret{\varphi\pi}}(2) \otimes
    \interpret{s}[x:=\pi^1(\interpret{u})] * \typeinterpret{\pi}\ \oplus\\
  \phantom{A}
  \lift_{\typeinterpret{\varphi\pi}}(2) \otimes
    \interpret{t}[y:=\pi^2(\interpret{u})] * \typeinterpret{\pi}\ \oplus\\
  \phantom{A}
  \lift_{\typeinterpret{\varphi\pi}}(2 \otimes
    \flatten_{\typeinterpret{\sigma} \times \typeinterpret{\tau}}(
      \interpret{u})) \otimes \interpret{s}[x:=\pi^1(\interpret{u})]
      * \typeinterpret{\pi}\ \oplus \\
  \phantom{A}
  \lift_{\typeinterpret{\varphi\pi}}(2 \otimes
    \flatten_{\typeinterpret{\sigma} \times \typeinterpret{\tau}}(
      \interpret{u})) \otimes \interpret{t}[x:=\pi^2(\interpret{u})]
      * \typeinterpret{\pi} \\
  \end{array}
  \]

  On the right-hand side, we have:
  \[
  \begin{array}{l}
  \interpret{\mathtt{case}_{\sigma,\tau,\varphi\pi}(u,
    \abs{x:\sigma}{\mathtt{tapp}_{\varphi,\pi}(s)},
    \abs{y:\tau}{\mathtt{tapp}_{\varphi,\pi}(t)})} \approx \\
  \lift_{\typeinterpret{\varphi\pi}}(2)\ \oplus \\
  \phantom{A}
  \lift_{\typeinterpret{\varphi\pi}}(3 \otimes \flatten_{
    \typeinterpret{\sigma} \times \typeinterpret{\tau}}(
    \interpret{u}))\ \oplus \\
  \phantom{A}
  \lift_{\typeinterpret{\varphi\pi}}(\flatten_{\typeinterpret{\sigma}
    \times \typeinterpret{\tau}}(\interpret{u}) \oplus 1)\ \otimes \\
  \phantom{ABC}(\
    \lift_{\typeinterpret{\varphi\pi}}(2) \otimes
    (\typeinterpret{s}[x:=\pi^1(\interpret{u})]
      * \typeinterpret{\pi}) \oplus
    \lift_{\typeinterpret{\varphi\pi}}(1)\ \oplus \\
  \phantom{ABCD}
    \lift_{\typeinterpret{\varphi\pi}}(2) \otimes
    (\typeinterpret{t}[x:=\pi^2(\interpret{u})]
      * \typeinterpret{\pi}) \oplus
    \lift_{\typeinterpret{\varphi\pi}}(1)
  \ ) \approx \\
  \lift_{\typeinterpret{\varphi\pi}}(4)\ \oplus \\
  \phantom{A}
  \lift_{\typeinterpret{\varphi\pi}}(5 \otimes \flatten_{
    \typeinterpret{\sigma} \times \typeinterpret{\tau}}(
    \interpret{u}))\ \oplus \\
  \phantom{A}
  \lift_{\typeinterpret{\varphi\pi}}(\flatten_{\typeinterpret{\sigma}
    \times \typeinterpret{\tau}}(\interpret{u}) \oplus 1)\ \otimes \\
  \phantom{ABC}(\
    \lift_{\typeinterpret{\varphi\pi}}(2) \otimes
    (\typeinterpret{s}[x:=\pi^1(\interpret{u})]
      * \typeinterpret{\pi})\ \oplus \\
  \phantom{ABCD}
    \lift_{\typeinterpret{\varphi\pi}}(2) \otimes
    (\typeinterpret{t}[x:=\pi^2(\interpret{u})]
      * \typeinterpret{\pi})
  \ ) \approx \\
  \lift_{\typeinterpret{\varphi\pi}}(4)\ \oplus \\
  \phantom{A}
  \lift_{\typeinterpret{\varphi\pi}}(5 \otimes \flatten_{
    \typeinterpret{\sigma} \times \typeinterpret{\tau}}(
    \interpret{u}))\ \oplus \\
  \phantom{A}
  \lift_{\typeinterpret{\varphi\pi}}(2) \otimes
    \interpret{s}[x:=\pi^1(\interpret{u})] * \typeinterpret{\pi}\ \oplus\\
  \phantom{A}
  \lift_{\typeinterpret{\varphi\pi}}(2) \otimes
    \interpret{t}[y:=\pi^2(\interpret{u})] * \typeinterpret{\pi}\ \oplus\\
  \phantom{A}
  \lift_{\typeinterpret{\varphi\pi}}(2 \otimes
    \flatten_{\typeinterpret{\sigma} \times \typeinterpret{\tau}}(
      \interpret{u})) \otimes \interpret{s}[x:=\pi^1(\interpret{u})]
      * \typeinterpret{\pi}\ \oplus \\
  \phantom{A}
  \lift_{\typeinterpret{\varphi\pi}}(2 \otimes
    \flatten_{\typeinterpret{\sigma} \times \typeinterpret{\tau}}(
      \interpret{u})) \otimes \interpret{t}[x:=\pi^2(\interpret{u})]
      * \typeinterpret{\pi} \\
  \end{array}
  \]
  Again, it is clear that the required inequality holds.

\item $\interpret{\proj^1_{\rho,\pi}(\mathtt{case}_{\sigma,\tau,
  \mathtt{and}\,\rho\,\pi}(u,\abs{x:\sigma}{s},\abs{y:\tau}{t}))} \succ \\
  \interpret{\mathtt{case}_{\sigma,\tau,\rho}(u,\abs{x:\sigma}{
  \proj^1_{\rho,\pi}(s)},\abs{y:\tau}{\proj^1_{\rho,\pi}(t)})}$ \\
  On the left-hand side, we have: \\
  \[
  \begin{array}{l}
  \interpret{\proj^1_{\rho,\pi}(\mathtt{case}_{\sigma,\tau,
  \mathtt{and}\,\rho\,\pi}(u,\abs{x}{s},\abs{y}{t}))} \approx \\
  \lift_{\typeinterpret{\rho}}(2) \otimes \pi^1( \\
  \phantom{AB}
    \lift_{\typeinterpret{\rho} \times \typeinterpret{\pi}}(2)\ \oplus \\
  \phantom{AB}
    \lift_{\typeinterpret{\rho} \times \typeinterpret{\pi}}(3 \otimes
    \flatten_{\typeinterpret{\sigma} \times \typeinterpret{\tau}}(
    \interpret{u}))\ \oplus \\
  \phantom{AB}
    \lift_{\typeinterpret{\rho} \times \typeinterpret{\pi}}(
    \flatten_{\typeinterpret{\sigma} \times \typeinterpret{\tau}}(
    \interpret{u}) \oplus 1)\ \otimes \\
  \phantom{ABCD}
    (\interpret{s}[x:=\pi^1(\interpret{u})] \oplus
     \interpret{t}[y:=\pi^2(\interpret{u})]) \\
  \phantom{A}) \oplus \lift_{\typeinterpret{\rho}}(1) \\
  \end{array}
  \]
  Taking into account that $\typeinterpret{\rho} \times
  \typeinterpret{\tau}$ is just shorthand notation for
  $\quant{p}{(\typeinterpret{\rho} \arrtype \typeinterpret{\tau}
  \arrtype p) \arrtype p}$, that $\pi^1(x) = x * \typeinterpret{\rho}
  \cdot (\abs{xy}{x})$, and that $\lift_{\sigma \arrtype \tau}(x)
  \cdot y \approx \lift_\tau(x)$, this term $\approx$
  \[
  \begin{array}{l}
  \lift_{\typeinterpret{\rho}}(5)\ \oplus \\
  \phantom{A}
    \lift_{\typeinterpret{\rho}}(6 \otimes
    \flatten_{\typeinterpret{\sigma} \times \typeinterpret{\tau}}(
    \interpret{u}))\ \oplus \\
  \phantom{A}
    \lift_{\typeinterpret{\rho}}(2 \otimes
    \flatten_{\typeinterpret{\sigma} \times
    \typeinterpret{\tau}}(\interpret{u}) \oplus 2) \otimes
    \pi^1(\interpret{s}[x:=\pi^1(\interpret{u})])\ \oplus \\
  \phantom{A}
    \lift_{\typeinterpret{\rho}}(2 \otimes
    \flatten_{\typeinterpret{\sigma} \times
    \typeinterpret{\tau}}(\interpret{u}) \oplus 2) \otimes
    \pi^1(\interpret{t}[y:=\pi^2(\interpret{u})]) \approx \\
  \lift_{\typeinterpret{\rho}}(5)\ \oplus \\
  \phantom{A}
    \lift_{\typeinterpret{\rho}}(6 \otimes
    \flatten_{\typeinterpret{\sigma} \times \typeinterpret{\tau}}(
    \interpret{u}))\ \oplus \\
  \phantom{A}
    \lift_{\typeinterpret{\rho}}(2 \otimes
    \flatten_{\typeinterpret{\sigma} \times
    \typeinterpret{\tau}}(2) \otimes
    \pi^1(\interpret{s}[x:=\pi^1(\interpret{u})])\ \oplus \\
  \phantom{A}
    \lift_{\typeinterpret{\rho}}(2 \otimes
    \flatten_{\typeinterpret{\sigma} \times
    \typeinterpret{\tau}}(2) \otimes
    \pi^1(\interpret{t}[y:=\pi^2(\interpret{u})])\ \oplus \\
  \phantom{A}
    \lift_{\typeinterpret{\rho}}(2 \otimes
    \flatten_{\typeinterpret{\sigma} \times
    \typeinterpret{\tau}}(\interpret{u})) \otimes
    \pi^1(\interpret{s}[x:=\pi^1(\interpret{u})])\ \oplus \\
  \phantom{A}
    \lift_{\typeinterpret{\rho}}(2 \otimes
    \flatten_{\typeinterpret{\sigma} \times
    \typeinterpret{\tau}}(\interpret{u})) \otimes
    \pi^1(\interpret{t}[y:=\pi^2(\interpret{u})])
  \end{array}
  \]
  On the right-hand side, we have:
  \[
  \begin{array}{l}
  \interpret{\mathtt{case}_{\sigma,\tau,\rho}(u,\abs{x:\sigma}{
  \proj^1_{\rho,\pi}(s)},\abs{y:\tau}{\proj^1_{\rho,\pi}(t)})} \approx \\
  \lift_{\typeinterpret{\rho}}(2)\ \oplus \\
  \phantom{A}
    \lift_{\typeinterpret{\rho}}(3 \otimes
    \flatten_{\typeinterpret{\sigma} \times \typeinterpret{\tau}}(
    \interpret{u}))\ \oplus \\
  \phantom{A}
    \lift_{\typeinterpret{\rho}}(
    \flatten_{\typeinterpret{\sigma} \times \typeinterpret{\tau}}(
    \interpret{u}) \oplus 1)\ \otimes \\
  \phantom{ABC}
    (\ \lift_{\typeinterpret{\rho}}(2) \otimes \pi^1(\interpret{s})[
    x:=\pi^1(\interpret{u})] \oplus \lift_{\typeinterpret{\rho}}(1)
    \oplus \\
  \phantom{ABCD}
    \lift_{\typeinterpret{\rho}}(2) \otimes \pi^1(\interpret{t})[
    y:=\pi^2(\interpret{u})] \oplus \lift_{\typeinterpret{\rho}}(1)
    \ )
  \end{array}
  \]
  Following the definition of $\pi^1$, we can pull the substitution
  inside $\pi^1$, and rewrite this term to:
  \[
  \begin{array}{l}
  \lift_{\typeinterpret{\rho}}(2)\ \oplus \\
  \phantom{A}
    \lift_{\typeinterpret{\rho}}(3 \otimes
    \flatten_{\typeinterpret{\sigma} \times \typeinterpret{\tau}}(
    \interpret{u}))\ \oplus \\
  \phantom{A}
    \lift_{\typeinterpret{\rho}}(
    \flatten_{\typeinterpret{\sigma} \times \typeinterpret{\tau}}(
    \interpret{u}) \oplus 1)\ \otimes \\
  \phantom{ABC}
    (\ \lift_{\typeinterpret{\rho}}(2) \otimes \pi^1(\interpret{s}[
    x:=\pi^1(\interpret{u})])\ \oplus \\
  \phantom{ABCD}
    \lift_{\typeinterpret{\rho}}(2) \otimes \pi^1(\interpret{t}[
    y:=\pi^2(\interpret{u})]) \oplus \lift_{\typeinterpret{\rho}}(2)
    \ ) \approx \\
  \lift_{\typeinterpret{\rho}}(4)\ \oplus \\
  \phantom{A}
    \lift_{\typeinterpret{\rho}}(5 \otimes
    \flatten_{\typeinterpret{\sigma} \times \typeinterpret{\tau}}(
    \interpret{u}))\ \oplus \\
  \phantom{A}
    \lift_{\typeinterpret{\rho}}(
    \flatten_{\typeinterpret{\sigma} \times \typeinterpret{\tau}}(
    \interpret{u}) \oplus 1)\ \otimes \\
  \phantom{ABC}
    (\ \lift_{\typeinterpret{\rho}}(2) \otimes \pi^1(\interpret{s}[
    x:=\pi^1(\interpret{u})]) \oplus
    \lift_{\typeinterpret{\rho}}(2) \otimes \pi^1(\interpret{t}[
    y:=\pi^2(\interpret{u})])
    \ ) \approx \\
  \lift_{\typeinterpret{\rho}}(4)\ \oplus \\
  \phantom{A}
    \lift_{\typeinterpret{\rho}}(5 \otimes
    \flatten_{\typeinterpret{\sigma} \times \typeinterpret{\tau}}(
    \interpret{u}))\ \oplus \\
  \phantom{A}
    \lift_{\typeinterpret{\rho}}(2) \otimes \pi^1(\interpret{s}[
    x:=\pi^1(\interpret{u})])\ \oplus \\
  \phantom{A}
    \lift_{\typeinterpret{\rho}}(2) \otimes \pi^1(\interpret{t}[
    y:=\pi^2(\interpret{u})])\ \oplus \\
  \phantom{A}
    \lift_{\typeinterpret{\rho}}(2 \otimes
    \flatten_{\typeinterpret{\sigma} \times \typeinterpret{\tau}}(
    \interpret{u}) \otimes
    \lift_{\typeinterpret{\rho}}(2) \otimes \pi^1(\interpret{s}[
    x:=\pi^1(\interpret{u})])\ \oplus \\
  \phantom{A}
    \lift_{\typeinterpret{\rho}}(2 \otimes
    \flatten_{\typeinterpret{\sigma} \times \typeinterpret{\tau}}(
    \interpret{u}) \otimes
    \lift_{\typeinterpret{\rho}}(2) \otimes \pi^1(\interpret{t}[
    y:=\pi^2(\interpret{u})])
  \end{array}
  \]
  This is once more oriented by absolute positiveness.

\item $\interpret{\proj^2_{\rho,\pi}(\mathtt{case}_{\sigma,\tau,
  \mathtt{and}\,\rho,\pi}(u,\abs{x:\sigma}{s},\abs{y:\tau}{t}))}
  \succ \interpret{\mathtt{case}_{\sigma,\tau,\pi}(u,\abs{x:
  \sigma}{\proj^2_{\rho,\pi}(s)},\abs{y:\tau}{\proj^2_{\rho,\pi}(t)})}$ \\
  Analogous to the inequality above.

\item $\interpret{\mathtt{case}_{\rho,\pi,\xi}(\mathtt{case}_{\sigma,
  \tau,\mathtt{or}\,\rho\,\pi}(u,\abs{x:\sigma}{s},\abs{y:\tau}{t}),
  \abs{z:\rho}{v},\abs{a:\pi}{w})} \succ \\
  \interpret{\mathtt{case}_{\sigma,\tau,\xi}(u,\abs{x:\sigma}{
  \mathtt{case}_{\rho,\pi,\xi}(s,\abs{z:\rho}{v},\abs{a:\pi}{w})},
  \abs{y:\tau}{\mathtt{case}_{\rho,\pi,\xi}(t,\abs{z:\rho}{v},\abs{a:
  \pi}{w})})}$ \\
  This is the longest of the inequalities.  As before, we turn first
  to the left-hand side.
  \[
  \begin{array}{l}
  \interpret{\mathtt{case}_{\rho,\pi,\xi}(\mathtt{case}_{\sigma,
  \tau,\mathtt{or}\,\rho\,\pi}(u,\abs{x:\sigma}{s},\abs{y:\tau}{t}),
  \abs{z:\rho}{v},\abs{a:\pi}{w})} \approx \\
  \mathcal{J}(\mathtt{case})_{\typeinterpret{\rho},\typeinterpret{\pi},
    \typeinterpret{\xi}}(
    \mathcal{J}(\mathtt{case})_{\typeinterpret{\sigma},
    \typeinterpret{\tau},\typeinterpret{\rho} \times
    \typeinterpret{pi}}(\interpret{u},\abs{x}{\interpret{s}},
    \abs{y}{\interpret{t}}),\abs{z}{\interpret{v}},
    \abs{a}{\interpret{w}}) \approx \\
  \mathcal{J}(\mathtt{case})_{\typeinterpret{\rho},
    \typeinterpret{\pi},\typeinterpret{\xi}}( \\
  \phantom{ABC}
    \lift_{\typeinterpret{\rho} \times \typeinterpret{\pi}}(2) \oplus
    \lift_{\typeinterpret{\rho} \times \typeinterpret{\pi}}(3 \otimes
      \flatten_{\typeinterpret{\sigma} \times \typeinterpret{\tau}}(
      \interpret{u}))\ \oplus \\
  \end{array}
  \]
  \[
  \begin{array}{l}
  \phantom{ABC}
    \interpret{s}[x:=\pi^1(\interpret{u})] \oplus
    \interpret{t}[y:=\pi^2(\interpret{u})]\ \oplus \\
  \phantom{ABC}
    \lift_{\typeinterpret{\rho} \times \typeinterpret{\pi}}(
      \flatten_{\typeinterpret{\sigma} \times \typeinterpret{\tau}}(
      \interpret{u})) \otimes \interpret{s}[x:=\pi^1(\interpret{u})]
      \oplus \\
  \phantom{ABC}
    \lift_{\typeinterpret{\rho} \times \typeinterpret{\pi}}(
      \flatten_{\typeinterpret{\sigma} \times \typeinterpret{\tau}}(
      \interpret{u})) \otimes \interpret{t}[y:=\pi^2(\interpret{u})] \\
  \phantom{A}
  , \abs{z}{\interpret{v}},\ \abs{a}{\interpret{w}}\ ) \\
  \end{array}
  \]
  Once we start filling in the outer $\mathtt{case}$ interpretation,
  this is going to get very messy indeed.  So, we will use the
  following shorthand notation: \\
  $su = \interpret{s}[x:=\pi^1(\interpret{u})]$ \\
  $tu = \interpret{t}[y:=\pi^2(\interpret{u})]$ \\
  $A = \,
    \lift_{\typeinterpret{\rho} \times \typeinterpret{\pi}}(2) \oplus
    \lift_{\typeinterpret{\rho} \times \typeinterpret{\pi}}(3 \otimes
      \flatten_{\typeinterpret{\sigma} \times \typeinterpret{\tau}}(
      \interpret{u})) \oplus su \oplus tu\ \oplus \\
  \phantom{ABC}
    \lift_{\typeinterpret{\rho} \times \typeinterpret{\pi}}(
      \flatten_{\typeinterpret{\sigma} \times \typeinterpret{\tau}}(
      \interpret{u})) \otimes su\ \oplus \\
  \phantom{ABC}
    \lift_{\typeinterpret{\rho} \times \typeinterpret{\pi}}(
      \flatten_{\typeinterpret{\sigma} \times \typeinterpret{\tau}}(
      \interpret{u})) \otimes\ tu \\
  $ \\
  Then, the left-hand side $\approx$
  \[
  \begin{array}{l}
  \mathcal{J}(\mathtt{case})_{\typeinterpret{\rho},
    \typeinterpret{\pi},\typeinterpret{\xi}}(A,
      \abs{z}{\interpret{v}}, \abs{a}{\interpret{w}}\ ) \approx \\
  \lift_{\typeinterpret{\xi}}(2) \oplus
    \lift_{\typeinterpret{\xi}}(3 \otimes \flatten_{
    \typeinterpret{\rho} \times \typeinterpret{\pi}}(A))\ \oplus \\
  \phantom{A}
    \interpret{v}[z:=\pi^1(A)] \oplus
    \interpret{w}[a:=\pi^2(A)]\ \oplus\ \\
  \phantom{A}
    \lift_{\typeinterpret{\xi}}(\flatten_{
    \typeinterpret{\rho} \times \typeinterpret{\pi}}(A)) \otimes
    \interpret{v}[z:=\pi^1(A)]\ \oplus \\
  \phantom{A}
    \lift_{\typeinterpret{\xi}}(\flatten_{
    \typeinterpret{\rho} \times \typeinterpret{\pi}}(A)) \otimes
    \interpret{w}[a:=\pi^2(A)] \approx \\
  \lift_{\typeinterpret{\xi}}(2)\ \oplus \\
  \phantom{A}
     \lift_{\typeinterpret{\xi}}(6)\ \oplus \\
  \phantom{A}
     \lift_{\typeinterpret{\xi}}(9 \otimes
      \flatten_{\typeinterpret{\sigma} \times \typeinterpret{\tau}}(
      \interpret{u}))\ \oplus \\
  \phantom{A}
    \lift_{\typeinterpret{\xi}}(3 \otimes \flatten_{\typeinterpret{\rho}
    \times \typeinterpret{\pi}}(su))\ \oplus \\
  \phantom{A}
    \lift_{\typeinterpret{\xi}}(3 \otimes \flatten_{\typeinterpret{\rho}
    \times \typeinterpret{\pi}}(tu))\ \oplus \\
  \phantom{A}
    \lift_{\typeinterpret{\xi}}(3 \otimes \flatten_{\typeinterpret{\sigma}
    \times \typeinterpret{\tau}}(\interpret{u}) \otimes
    \flatten_{\typeinterpret{\rho} \times \typeinterpret{\pi}}(su))\
    \oplus \\
  \phantom{A}
    \lift_{\typeinterpret{\xi}}(3 \otimes \flatten_{\typeinterpret{\sigma}
    \times \typeinterpret{\tau}}(\interpret{u}) \otimes
    \flatten_{\typeinterpret{\rho} \times \typeinterpret{\pi}}(tu))\
    \oplus \\
  \phantom{A}
    \interpret{v}[z:=\pi^1(A)] \oplus
    \interpret{w}[a:=\pi^2(A)]\ \oplus\ \\
  \phantom{A}
    \lift_{\typeinterpret{\xi}}(2) \otimes
     \interpret{v}[z:=\pi^1(A)]
    \ \oplus \\
  \phantom{A}
    \lift_{\typeinterpret{\xi}}(3 \otimes
    \flatten_{\typeinterpret{\sigma} \times \typeinterpret{\tau}}(
    \interpret{u})) \otimes
     \interpret{v}[z:=\pi^1(A)]
    \ \oplus \\
  \phantom{A}
    \lift_{\typeinterpret{\xi}}(\flatten_{\typeinterpret{\rho} \times
    \typeinterpret{\pi}}(su)) \otimes
     \interpret{v}[z:=\pi^1(A)]
    \ \oplus\\
  \phantom{A}
    \lift_{\typeinterpret{\xi}}(\flatten_{\typeinterpret{\rho} \times
    \typeinterpret{\pi}}(tu)) \otimes
     \interpret{v}[z:=\pi^1(A)]
    \ \oplus \\
  \phantom{A}
    \lift_{\typeinterpret{\xi}}(\flatten_{\typeinterpret{\sigma} \times
    \typeinterpret{\tau}}(\interpret{u}) \otimes
    \flatten_{\typeinterpret{\rho} \times \typeinterpret{\pi}}(su))
    \otimes
     \interpret{v}[z:=\pi^1(A)]
    \ \oplus \\
  \phantom{A}
    \lift_{\typeinterpret{\xi}}(\flatten_{\typeinterpret{\sigma} \times
    \typeinterpret{\tau}}(\interpret{u}) \otimes
    \flatten_{\typeinterpret{\rho} \times \typeinterpret{\pi}}(tu))
    \otimes
     \interpret{v}[z:=\pi^1(A)]
    \ \oplus \\
  \phantom{A}
    \lift_{\typeinterpret{\xi}}(2) \otimes
    \interpret{w}[a:=\pi^2(A)]
    \ \oplus \\
  \phantom{A}
    \lift_{\typeinterpret{\xi}}(3 \otimes
    \flatten_{\typeinterpret{\sigma} \times \typeinterpret{\tau}}(
    \interpret{u})) \otimes
    \interpret{w}[a:=\pi^2(A)]
    \ \oplus \\
  \phantom{A}
    \lift_{\typeinterpret{\xi}}(\flatten_{\typeinterpret{\rho} \times
    \typeinterpret{\pi}}(su)) \otimes
    \interpret{w}[a:=\pi^2(A)]
    \ \oplus\\
  \phantom{A}
    \lift_{\typeinterpret{\xi}}(\flatten_{\typeinterpret{\rho} \times
    \typeinterpret{\pi}}(tu)) \otimes
    \interpret{w}[a:=\pi^2(A)]
    \ \oplus \\
  \phantom{A}
    \lift_{\typeinterpret{\xi}}(\flatten_{\typeinterpret{\sigma} \times
    \typeinterpret{\tau}}(\interpret{u}) \otimes
    \flatten_{\typeinterpret{\rho} \times \typeinterpret{\pi}}(su))
    \otimes
    \interpret{w}[a:=\pi^2(A)]
    \ \oplus \\
  \phantom{A}
    \lift_{\typeinterpret{\xi}}(\flatten_{\typeinterpret{\sigma} \times
    \typeinterpret{\tau}}(\interpret{u}) \otimes
    \flatten_{\typeinterpret{\rho} \times \typeinterpret{\pi}}(tu))
    \otimes
    \interpret{w}[a:=\pi^2(A)]
    \ \oplus \\
  \end{array}
  \]
  We can \emph{slightly} shorten this term by combining parts, but
  the result is still quite long:
  \[
  \begin{array}{l}
  \langle\text{the left-hand side}\rangle \approx \\
  \lift_{\typeinterpret{\xi}}(8)\ \oplus \\
  \phantom{A}
     \lift_{\typeinterpret{\xi}}(9 \otimes
      \flatten_{\typeinterpret{\sigma} \times \typeinterpret{\tau}}(
      \interpret{u}))\ \oplus \\
  \phantom{A}
    \lift_{\typeinterpret{\xi}}(3 \otimes \flatten_{\typeinterpret{\rho}
    \times \typeinterpret{\pi}}(su))\ \oplus \\
  \phantom{A}
    \lift_{\typeinterpret{\xi}}(3 \otimes \flatten_{\typeinterpret{\rho}
    \times \typeinterpret{\pi}}(tu))\ \oplus \\
  \phantom{A}
    \lift_{\typeinterpret{\xi}}(3 \otimes \flatten_{\typeinterpret{\sigma}
    \times \typeinterpret{\tau}}(\interpret{u}) \otimes
    \flatten_{\typeinterpret{\rho} \times \typeinterpret{\pi}}(su))\
    \oplus \\
  \phantom{A}
    \lift_{\typeinterpret{\xi}}(3 \otimes \flatten_{\typeinterpret{\sigma}
    \times \typeinterpret{\tau}}(\interpret{u}) \otimes
    \flatten_{\typeinterpret{\rho} \times \typeinterpret{\pi}}(tu))\
    \oplus \\
  \phantom{A}
    \lift_{\typeinterpret{\xi}}(3) \otimes \interpret{v}[z:=\pi^1(A)]
    \ \oplus \\
  \end{array}
  \]
  \[
  \begin{array}{l}
  \phantom{A}
    \lift_{\typeinterpret{\xi}}(3 \otimes
    \flatten_{\typeinterpret{\sigma} \times \typeinterpret{\tau}}(
    \interpret{u})) \otimes
     \interpret{v}[z:=\pi^1(A)]
    \ \oplus \\
  \phantom{A}
    \lift_{\typeinterpret{\xi}}(\flatten_{\typeinterpret{\rho} \times
    \typeinterpret{\pi}}(su)) \otimes
     \interpret{v}[z:=\pi^1(A)]
    \ \oplus\\
  \phantom{A}
    \lift_{\typeinterpret{\xi}}(\flatten_{\typeinterpret{\rho} \times
    \typeinterpret{\pi}}(tu)) \otimes
     \interpret{v}[z:=\pi^1(A)]
    \ \oplus \\
  \phantom{A}
    \lift_{\typeinterpret{\xi}}(\flatten_{\typeinterpret{\sigma} \times
    \typeinterpret{\tau}}(\interpret{u}) \otimes
    \flatten_{\typeinterpret{\rho} \times \typeinterpret{\pi}}(su))
    \otimes
     \interpret{v}[z:=\pi^1(A)]
    \ \oplus \\
  \phantom{A}
    \lift_{\typeinterpret{\xi}}(\flatten_{\typeinterpret{\sigma} \times
    \typeinterpret{\tau}}(\interpret{u}) \otimes
    \flatten_{\typeinterpret{\rho} \times \typeinterpret{\pi}}(tu))
    \otimes
     \interpret{v}[z:=\pi^1(A)]
    \ \oplus \\
  \phantom{A}
    \lift_{\typeinterpret{\xi}}(3) \otimes \interpret{w}[a:=\pi^2(A)]
    \ \oplus \\
  \phantom{A}
    \lift_{\typeinterpret{\xi}}(3 \otimes
    \flatten_{\typeinterpret{\sigma} \times \typeinterpret{\tau}}(
    \interpret{u})) \otimes
    \interpret{w}[a:=\pi^2(A)]
    \ \oplus \\
  \phantom{A}
    \lift_{\typeinterpret{\xi}}(\flatten_{\typeinterpret{\rho} \times
    \typeinterpret{\pi}}(su)) \otimes
    \interpret{w}[a:=\pi^2(A)]
    \ \oplus\\
  \phantom{A}
    \lift_{\typeinterpret{\xi}}(\flatten_{\typeinterpret{\rho} \times
    \typeinterpret{\pi}}(tu)) \otimes
    \interpret{w}[a:=\pi^2(A)]
    \ \oplus \\
  \phantom{A}
    \lift_{\typeinterpret{\xi}}(\flatten_{\typeinterpret{\sigma} \times
    \typeinterpret{\tau}}(\interpret{u}) \otimes
    \flatten_{\typeinterpret{\rho} \times \typeinterpret{\pi}}(su))
    \otimes
    \interpret{w}[a:=\pi^2(A)]
    \ \oplus \\
  \phantom{A}
    \lift_{\typeinterpret{\xi}}(\flatten_{\typeinterpret{\sigma} \times
    \typeinterpret{\tau}}(\interpret{u}) \otimes
    \flatten_{\typeinterpret{\rho} \times \typeinterpret{\pi}}(tu))
    \otimes
    \interpret{w}[a:=\pi^2(A)]
  \end{array}
  \]
  Now, let us turn to the right-hand side.
  \[
  \begin{array}{l}
  \interpret{\mathtt{case}_{\sigma,\tau,\xi}(u,\abs{x:\sigma}{
  \mathtt{case}_{\rho,\pi,\xi}(s,\abs{z:\rho}{v},\abs{a:\pi}{w})},
  \abs{y:\tau}{\mathtt{case}_{\rho,\pi,\xi}(t,\abs{z:\rho}{v},\abs{a:
  \pi}{w})})} \approx \\
  \mathcal{J}(\mathtt{case})_{\sigma,\tau,\xi}(\interpret{u},\abs{x}{
    \lift_{\typeinterpret{\xi}}(2) \oplus
      \lift_{\typeinterpret{\xi}}(3 \otimes \flatten_{
      \typeinterpret{\rho} \times \typeinterpret{\pi}}(\interpret{s}))\
      \oplus \\
    \phantom{ABCDEFGHIJKL}
      \interpret{v}[z:=\pi^1(\interpret{s})] \oplus
      \interpret{w}[a:=\pi^2(\interpret{s})]\ \oplus \\
    \phantom{ABCDEFGHIJKL}
      \lift_{\typeinterpret{\xi}}(\flatten_{\typeinterpret{\rho} \times
      \typeinterpret{\pi}}(\interpret{s}))
      \otimes \interpret{v}[z:=\pi^1(\interpret{s})]\
      \oplus \\
    \phantom{ABCDEFGHIJKL}
      \lift_{\typeinterpret{\xi}}(\flatten_{\typeinterpret{\rho}
      \times \typeinterpret{\pi}}(\interpret{s}))
      \otimes \interpret{w}[a:=\pi^2(\interpret{s})]
    },\ \\
    \phantom{ABCDEFGHIJ}\abs{y}{
    \lift_{\typeinterpret{\xi}}(2) \oplus
      \lift_{\typeinterpret{\xi}}(3 \otimes \flatten_{
      \typeinterpret{\rho} \times \typeinterpret{\pi}}(\interpret{t}))\
      \oplus \\
    \phantom{ABCDEFGHIJKL}
      \interpret{v}[z:=\pi^1(\interpret{t})] \oplus
      \interpret{w}[a:=\pi^2(\interpret{t})]\ \oplus \\
    \phantom{ABCDEFGHIJKL}
      \lift_{\typeinterpret{\xi}}(\flatten_{\typeinterpret{\rho} \times
      \typeinterpret{\pi}}(\interpret{t}))
      \otimes \interpret{v}[z:=\pi^1(\interpret{t})]\
      \oplus \\
    \phantom{ABCDEFGHIJKL}
      \lift_{\typeinterpret{\xi}}(\flatten_{\typeinterpret{\rho}
      \times \typeinterpret{\pi}}(\interpret{t}))
      \otimes \interpret{w}[a:=\pi^2(\interpret{t})]
    }\ )
  \end{array}
  \]
  For brevity, we introduce another shorthand notation:
  for a given term $q$: \\
  $B_q =
    \lift_{\typeinterpret{\xi}}(2) \oplus
      \lift_{\typeinterpret{\xi}}(3 \otimes \flatten_{
      \typeinterpret{\rho} \times \typeinterpret{\pi}}(q))\ \oplus \\
    \phantom{ABC}
      \interpret{v}[z:=\pi^1(q)] \oplus
      \interpret{w}[a:=\pi^2(q)]\ \oplus \\
    \phantom{ABC}
      \lift_{\typeinterpret{\xi}}(\flatten_{\typeinterpret{\rho} \times
      \typeinterpret{\pi}}(q))
      \otimes \interpret{v}[z:=\pi^1(q)]\
      \oplus \\
    \phantom{ABC}
      \lift_{\typeinterpret{\xi}}(\flatten_{\typeinterpret{\rho}
      \times \typeinterpret{\pi}}(q)) \otimes \interpret{w}[a:=\pi^2(q)]
  $. \\
  With this, we have:
  \[
  \begin{array}{l}
  \langle\text{the right-hand side}\rangle \approx \\
  \mathcal{J}(\mathtt{case})_{\typeinterpret{\sigma},
  \typeinterpret{\tau},\typeinterpret{\xi}}(\interpret{u},\abs{x}{
  B_{\interpret{s}}},\abs{y}{B_{\interpret{t}}}) \approx \\
  \lift_{\typeinterpret{\xi}}(2) \oplus
    \lift_{\typeinterpret{\xi}}(3 \otimes \flatten_{\typeinterpret{
    \sigma} \times \typeinterpret{\tau}}(\interpret{u}))\ \oplus \\
  \phantom{A}
  B_{\interpret{s}}[x:=\pi^1(\interpret{u})] \oplus
  B_{\interpret{t}}[x:=\pi^2(\interpret{u})]\ \oplus \\
  \phantom{A}
  \lift_{\typeinterpret{\xi}}(\flatten_{\typeinterpret{\sigma}
    \times \typeinterpret{\tau}}(\interpret{u})) \otimes
    B_{\interpret{s}}[x:=\pi^1(\interpret{u})]\ \oplus \\
  \phantom{A}
    \lift_{\typeinterpret{\xi}}(\flatten_{\typeinterpret{\sigma}
    \times \typeinterpret{\tau}}(\interpret{u})) \otimes
    B_{\interpret{t}}[x:=\pi^2(\interpret{u})]
  \end{array}
  \]
  Note that $x$ is a bound variable in $s$ and $y$ a bound variable
  in $t$; these variables do not occur in $B_q$.  So, we can rewrite
  the above term to:
  \[
  \begin{array}{l}
  \langle\text{the right-hand side}\rangle \approx \\
  \lift_{\typeinterpret{\xi}}(2) \oplus
    \lift_{\typeinterpret{\xi}}(3 \otimes \flatten_{\typeinterpret{
    \sigma} \times \typeinterpret{\tau}}(\interpret{u}))\ \oplus \\
  \phantom{A}
  B_{su} \oplus
  B_{tu}\ \oplus \\
  \phantom{A}
  \lift_{\typeinterpret{\xi}}(\flatten_{\typeinterpret{\sigma}
    \times \typeinterpret{\tau}}(\interpret{u})) \otimes
    B_{su}\ \oplus \\
  \phantom{A}
    \lift_{\typeinterpret{\xi}}(\flatten_{\typeinterpret{\sigma}
    \times \typeinterpret{\tau}}(\interpret{u})) \otimes
    B_{tu} \approx \\
  \lift_{\typeinterpret{\xi}}(2)\ \oplus \\
  \phantom{A}
    \lift_{\typeinterpret{\xi}}(3 \otimes \flatten_{\typeinterpret{
    \sigma} \times \typeinterpret{\tau}}(\interpret{u}))\ \oplus \\
  \phantom{A}
    \lift_{\typeinterpret{\xi}}(2)\ \oplus \\
  \phantom{A}
    \lift_{\typeinterpret{\xi}}(3 \otimes \flatten_{
    \typeinterpret{\rho} \times \typeinterpret{\pi}}(su))\ \oplus \\
  \phantom{A}
    \interpret{v}[z:=\pi^1(su)] \oplus
    \interpret{w}[a:=\pi^2(su)]\ \oplus \\
  \phantom{A}
    \lift_{\typeinterpret{\xi}}(\flatten_{\typeinterpret{\rho} \times
    \typeinterpret{\pi}}(su)) \otimes \interpret{v}[z:=\pi^1(su)]\
    \oplus \\
  \phantom{A}
    \lift_{\typeinterpret{\xi}}(\flatten_{\typeinterpret{\rho}
    \times \typeinterpret{\pi}}(su)) \otimes \interpret{w}[a:=\pi^2(su)]
    \ \oplus \\
  \end{array}
  \]
  \[
  \begin{array}{l}
  \phantom{A}
    \lift_{\typeinterpret{\xi}}(2)\ \oplus \\
  \phantom{A}
    \lift_{\typeinterpret{\xi}}(3 \otimes \flatten_{
    \typeinterpret{\rho} \times \typeinterpret{\pi}}(tu))\ \oplus \\
  \phantom{A}
    \interpret{v}[z:=\pi^1(tu)] \oplus
    \interpret{w}[a:=\pi^2(tu)]\ \oplus \\
  \phantom{A}
    \lift_{\typeinterpret{\xi}}(\flatten_{\typeinterpret{\rho} \times
    \typeinterpret{\pi}}(tu)) \otimes \interpret{v}[z:=\pi^1(tu)]\
    \oplus \\
  \phantom{A}
    \lift_{\typeinterpret{\xi}}(\flatten_{\typeinterpret{\rho}
    \times \typeinterpret{\pi}}(tu)) \otimes \interpret{w}[a:=\pi^2(tu)]
    \ \oplus \\
  \phantom{A}
    \lift_{\typeinterpret{\xi}}(2 \otimes \flatten_{
    \typeinterpret{\sigma} \times \typeinterpret{\tau}}(
    \interpret{u}))\ \oplus \\
  \phantom{A}
    \lift_{\typeinterpret{\xi}}(3 \otimes \flatten_{
    \typeinterpret{\sigma} \times \typeinterpret{\tau}}(
    \interpret{u}) \otimes \flatten_{
    \typeinterpret{\rho} \times \typeinterpret{\pi}}(su))\ \oplus \\
  \phantom{A}
    \lift_{\typeinterpret{\xi}}(\flatten_{\typeinterpret{\sigma}
    \times \typeinterpret{\tau}}(\interpret{u})) \otimes
    \interpret{v}[z:=\pi^1(su)]\ \oplus \\
  \phantom{A}
    \lift_{\typeinterpret{\xi}}(\flatten_{\typeinterpret{\sigma}
    \times \typeinterpret{\tau}}(\interpret{u})) \otimes
    \interpret{w}[a:=\pi^2(su)]\ \oplus \\
  \phantom{A}
    \lift_{\typeinterpret{\xi}}(\flatten_{\typeinterpret{\sigma}
    \times \typeinterpret{\tau}}(\interpret{u}) \otimes
    \flatten_{\typeinterpret{\rho} \times
    \typeinterpret{\pi}}(su)) \otimes \interpret{v}[z:=\pi^1(su)]\
    \oplus \\
  \phantom{A}
    \lift_{\typeinterpret{\xi}}(\flatten_{\typeinterpret{\sigma}
    \times \typeinterpret{\tau}}(\interpret{u}) \otimes
    \flatten_{\typeinterpret{\rho}
    \times \typeinterpret{\pi}}(su)) \otimes \interpret{w}[a:=\pi^2(su)]
    \ \oplus \\
  \phantom{A}
    \lift_{\typeinterpret{\xi}}(2 \otimes \flatten_{
    \typeinterpret{\sigma} \times \typeinterpret{\tau}}(
    \interpret{u}))\ \oplus \\
  \phantom{A}
    \lift_{\typeinterpret{\xi}}(3 \otimes \flatten_{
    \typeinterpret{\sigma} \times \typeinterpret{\tau}}(
    \interpret{u}) \otimes \flatten_{
    \typeinterpret{\rho} \times \typeinterpret{\pi}}(tu))\ \oplus \\
  \phantom{A}
    \lift_{\typeinterpret{\xi}}(\flatten_{\typeinterpret{\sigma}
    \times \typeinterpret{\tau}}(\interpret{u})) \otimes
    \interpret{v}[z:=\pi^1(tu)]\ \oplus \\
  \phantom{A}
    \lift_{\typeinterpret{\xi}}(\flatten_{\typeinterpret{\sigma}
    \times \typeinterpret{\tau}}(\interpret{u})) \otimes
    \interpret{w}[a:=\pi^2(tu)]\ \oplus \\
  \phantom{A}
    \lift_{\typeinterpret{\xi}}(\flatten_{\typeinterpret{\sigma}
    \times \typeinterpret{\tau}}(\interpret{u}) \otimes
    \flatten_{\typeinterpret{\rho} \times
    \typeinterpret{\pi}}(tu)) \otimes \interpret{v}[z:=\pi^1(tu)]\
    \oplus \\
  \phantom{A}
    \lift_{\typeinterpret{\xi}}(\flatten_{\typeinterpret{\sigma}
    \times \typeinterpret{\tau}}(\interpret{u}) \otimes
    \flatten_{\typeinterpret{\rho}
    \times \typeinterpret{\pi}}(tu)) \otimes \interpret{w}[a:=\pi^2(tu)]
  \end{array}
  \]
  Here, we can do some further combinations.
  Let us denote:
  \begin{itemize}
  \item $vsu := \interpret{v}[z:=\pi^1(su)] =
    \interpret{v}[z:=\pi^1(\interpret{s}[x:=\pi^1(\interpret{u})])]$
  \item $wsu := \interpret{w}[a:=\pi^2(su)] =
    \interpret{w}[a:=\pi^2(\interpret{s}[x:=\pi^1(\interpret{u})])]$
  \item $vtu := \interpret{v}[z:=\pi^1(tu)] =
    \interpret{v}[z:=\pi^1(\interpret{t}[y:=\pi^2(\interpret{u})])]$
  \item $wtu := \interpret{w}[a:=\pi^1(tu)] =
    \interpret{w}[a:=\pi^2(\interpret{t}[y:=\pi^2(\interpret{u})])]$
  \end{itemize}

  \medskip
  Then:
  \[
  \begin{array}{l}
  \langle\text{the right-hand side}\rangle \approx \\
  \lift_{\typeinterpret{\xi}}(6)\ \oplus \\
  \phantom{A}
    \lift_{\typeinterpret{\xi}}(7 \otimes \flatten_{\typeinterpret{
    \sigma} \times \typeinterpret{\tau}}(\interpret{u}))\ \oplus \\
  \phantom{A}
    \lift_{\typeinterpret{\xi}}(3 \otimes \flatten_{
    \typeinterpret{\rho} \times \typeinterpret{\pi}}(su))\ \oplus \\
  \phantom{A}
    \lift_{\typeinterpret{\xi}}(3 \otimes \flatten_{
    \typeinterpret{\rho} \times \typeinterpret{\pi}}(tu))\ \oplus \\
  \phantom{A}
    vsu \oplus wsu \oplus
    vtu \oplus wtu\ \oplus \\
  \phantom{A}
    \lift_{\typeinterpret{\xi}}(\flatten_{\typeinterpret{\rho} \times
    \typeinterpret{\pi}}(su)) \otimes vsu\ \oplus \\
  \phantom{A}
    \lift_{\typeinterpret{\xi}}(\flatten_{\typeinterpret{\rho}
    \times \typeinterpret{\pi}}(su)) \otimes wsu\ \oplus \\
  \phantom{A}
    \lift_{\typeinterpret{\xi}}(\flatten_{\typeinterpret{\rho} \times
    \typeinterpret{\pi}}(tu)) \otimes vtu\ \oplus \\
  \phantom{A}
    \lift_{\typeinterpret{\xi}}(\flatten_{\typeinterpret{\rho}
    \times \typeinterpret{\pi}}(tu)) \otimes wtu\ \oplus \\
  \phantom{A}
    \lift_{\typeinterpret{\xi}}(3 \otimes \flatten_{
    \typeinterpret{\sigma} \times \typeinterpret{\tau}}(
    \interpret{u}) \otimes \flatten_{
    \typeinterpret{\rho} \times \typeinterpret{\pi}}(su))\ \oplus \\
  \phantom{A}
    \lift_{\typeinterpret{\xi}}(\flatten_{\typeinterpret{\sigma}
    \times \typeinterpret{\tau}}(\interpret{u})) \otimes vsu\ \oplus \\
  \phantom{A}
    \lift_{\typeinterpret{\xi}}(\flatten_{\typeinterpret{\sigma}
    \times \typeinterpret{\tau}}(\interpret{u})) \otimes wsu\ \oplus \\
  \phantom{A}
    \lift_{\typeinterpret{\xi}}(\flatten_{\typeinterpret{\sigma}
    \times \typeinterpret{\tau}}(\interpret{u}) \otimes
    \flatten_{\typeinterpret{\rho} \times
    \typeinterpret{\pi}}(su)) \otimes vsu\ \oplus \\
  \phantom{A}
    \lift_{\typeinterpret{\xi}}(\flatten_{\typeinterpret{\sigma}
    \times \typeinterpret{\tau}}(\interpret{u}) \otimes
    \flatten_{\typeinterpret{\rho}
    \times \typeinterpret{\pi}}(su)) \otimes wsu\ \oplus \\
  \phantom{A}
    \lift_{\typeinterpret{\xi}}(3 \otimes \flatten_{
    \typeinterpret{\sigma} \times \typeinterpret{\tau}}(
    \interpret{u}) \otimes \flatten_{
    \typeinterpret{\rho} \times \typeinterpret{\pi}}(tu))\ \oplus \\
  \phantom{A}
    \lift_{\typeinterpret{\xi}}(\flatten_{\typeinterpret{\sigma}
    \times \typeinterpret{\tau}}(\interpret{u})) \otimes vtu\ \oplus \\
  \phantom{A}
    \lift_{\typeinterpret{\xi}}(\flatten_{\typeinterpret{\sigma}
    \times \typeinterpret{\tau}}(\interpret{u})) \otimes wtu\ \oplus \\
  \phantom{A}
    \lift_{\typeinterpret{\xi}}(\flatten_{\typeinterpret{\sigma}
    \times \typeinterpret{\tau}}(\interpret{u}) \otimes
    \flatten_{\typeinterpret{\rho} \times
    \typeinterpret{\pi}}(tu)) \otimes vtu\ \oplus \\
  \phantom{A}
    \lift_{\typeinterpret{\xi}}(\flatten_{\typeinterpret{\sigma}
    \times \typeinterpret{\tau}}(\interpret{u}) \otimes
    \flatten_{\typeinterpret{\rho}
    \times \typeinterpret{\pi}}(tu)) \otimes wtu
  \end{array}
  \]

  Now, if we strike out equal terms in the left-hand side and the
  right-hand side (after splitting additive terms where needed)
  the following inequality remains:
  \[
  \begin{array}{l}
  \lift_{\typeinterpret{\xi}}(2)\ \oplus \\
  \phantom{A}
     \lift_{\typeinterpret{\xi}}(2 \otimes
      \flatten_{\typeinterpret{\sigma} \times \typeinterpret{\tau}}(
      \interpret{u}))\ \oplus \\
  \phantom{A}
    \lift_{\typeinterpret{\xi}}(3) \otimes \interpret{v}[z:=\pi^1(A)]
    \ \oplus \\
  \phantom{A}
    \lift_{\typeinterpret{\xi}}(3 \otimes
    \flatten_{\typeinterpret{\sigma} \times \typeinterpret{\tau}}(
    \interpret{u})) \otimes
     \interpret{v}[z:=\pi^1(A)]
    \ \oplus \\
  \phantom{A}
    \lift_{\typeinterpret{\xi}}(\flatten_{\typeinterpret{\rho} \times
    \typeinterpret{\pi}}(su)) \otimes
     \interpret{v}[z:=\pi^1(A)]
    \ \oplus\\
  \phantom{A}
    \lift_{\typeinterpret{\xi}}(\flatten_{\typeinterpret{\rho} \times
    \typeinterpret{\pi}}(tu)) \otimes
     \interpret{v}[z:=\pi^1(A)]
    \ \oplus \\
  \phantom{A}
    \lift_{\typeinterpret{\xi}}(\flatten_{\typeinterpret{\sigma} \times
    \typeinterpret{\tau}}(\interpret{u}) \otimes
    \flatten_{\typeinterpret{\rho} \times \typeinterpret{\pi}}(su))
    \otimes
     \interpret{v}[z:=\pi^1(A)]
    \ \oplus \\
  \phantom{A}
    \lift_{\typeinterpret{\xi}}(\flatten_{\typeinterpret{\sigma} \times
    \typeinterpret{\tau}}(\interpret{u}) \otimes
    \flatten_{\typeinterpret{\rho} \times \typeinterpret{\pi}}(tu))
    \otimes
     \interpret{v}[z:=\pi^1(A)]
    \ \oplus \\
  \phantom{A}
    \lift_{\typeinterpret{\xi}}(3) \otimes \interpret{w}[a:=\pi^2(A)]
    \ \oplus \\
  \phantom{A}
    \lift_{\typeinterpret{\xi}}(3 \otimes
    \flatten_{\typeinterpret{\sigma} \times \typeinterpret{\tau}}(
    \interpret{u})) \otimes
    \interpret{w}[a:=\pi^2(A)]
    \ \oplus \\
  \phantom{A}
    \lift_{\typeinterpret{\xi}}(\flatten_{\typeinterpret{\rho} \times
    \typeinterpret{\pi}}(su)) \otimes
    \interpret{w}[a:=\pi^2(A)]
    \ \oplus\\
  \phantom{A}
    \lift_{\typeinterpret{\xi}}(\flatten_{\typeinterpret{\rho} \times
    \typeinterpret{\pi}}(tu)) \otimes
    \interpret{w}[a:=\pi^2(A)]
    \ \oplus \\
  \phantom{A}
    \lift_{\typeinterpret{\xi}}(\flatten_{\typeinterpret{\sigma} \times
    \typeinterpret{\tau}}(\interpret{u}) \otimes
    \flatten_{\typeinterpret{\rho} \times \typeinterpret{\pi}}(su))
    \otimes
    \interpret{w}[a:=\pi^2(A)]
    \ \oplus \\
  \phantom{A}
    \lift_{\typeinterpret{\xi}}(\flatten_{\typeinterpret{\sigma} \times
    \typeinterpret{\tau}}(\interpret{u}) \otimes
    \flatten_{\typeinterpret{\rho} \times \typeinterpret{\pi}}(tu))
    \otimes
    \interpret{w}[a:=\pi^2(A)] \succ \\
  vsu \oplus wsu \oplus
    vtu \oplus wtu\ \oplus \\
  \phantom{A}
    \lift_{\typeinterpret{\xi}}(\flatten_{\typeinterpret{\rho} \times
    \typeinterpret{\pi}}(su)) \otimes vsu\ \oplus \\
  \phantom{A}
    \lift_{\typeinterpret{\xi}}(\flatten_{\typeinterpret{\rho}
    \times \typeinterpret{\pi}}(su)) \otimes wsu\ \oplus \\
  \phantom{A}
    \lift_{\typeinterpret{\xi}}(\flatten_{\typeinterpret{\rho} \times
    \typeinterpret{\pi}}(tu)) \otimes vtu\ \oplus \\
  \phantom{A}
    \lift_{\typeinterpret{\xi}}(\flatten_{\typeinterpret{\rho}
    \times \typeinterpret{\pi}}(tu)) \otimes wtu\ \oplus \\
  \phantom{A}
    \lift_{\typeinterpret{\xi}}(\flatten_{\typeinterpret{\sigma}
    \times \typeinterpret{\tau}}(\interpret{u})) \otimes vsu\ \oplus \\
  \phantom{A}
    \lift_{\typeinterpret{\xi}}(\flatten_{\typeinterpret{\sigma}
    \times \typeinterpret{\tau}}(\interpret{u})) \otimes wsu\ \oplus \\
  \phantom{A}
    \lift_{\typeinterpret{\xi}}(\flatten_{\typeinterpret{\sigma}
    \times \typeinterpret{\tau}}(\interpret{u}) \otimes
    \flatten_{\typeinterpret{\rho} \times
    \typeinterpret{\pi}}(su)) \otimes vsu\ \oplus \\
  \phantom{A}
    \lift_{\typeinterpret{\xi}}(\flatten_{\typeinterpret{\sigma}
    \times \typeinterpret{\tau}}(\interpret{u}) \otimes
    \flatten_{\typeinterpret{\rho}
    \times \typeinterpret{\pi}}(su)) \otimes wsu\ \oplus \\
  \phantom{A}
    \lift_{\typeinterpret{\xi}}(\flatten_{\typeinterpret{\sigma}
    \times \typeinterpret{\tau}}(\interpret{u})) \otimes vtu\ \oplus \\
  \phantom{A}
    \lift_{\typeinterpret{\xi}}(\flatten_{\typeinterpret{\sigma}
    \times \typeinterpret{\tau}}(\interpret{u})) \otimes wtu\ \oplus \\
  \phantom{A}
    \lift_{\typeinterpret{\xi}}(\flatten_{\typeinterpret{\sigma}
    \times \typeinterpret{\tau}}(\interpret{u}) \otimes
    \flatten_{\typeinterpret{\rho} \times
    \typeinterpret{\pi}}(tu)) \otimes vtu\ \oplus \\
  \phantom{A}
    \lift_{\typeinterpret{\xi}}(\flatten_{\typeinterpret{\sigma}
    \times \typeinterpret{\tau}}(\interpret{u}) \otimes
    \flatten_{\typeinterpret{\rho}
    \times \typeinterpret{\pi}}(tu)) \otimes wtu
  \end{array}
  \]
  But now note that $A \succeq su$ and $A \succeq tu$.  Therefore, by
  monotonicity, each term L$i \succeq$ R$i$ below:
  \[
  \begin{array}{l}
  \lift_{\typeinterpret{\xi}}(2)\ \oplus \\
  \phantom{A}
     \lift_{\typeinterpret{\xi}}(2 \otimes
      \flatten_{\typeinterpret{\sigma} \times \typeinterpret{\tau}}(
      \interpret{u}))\ \oplus \\
  \phantom{A}
    \lift_{\typeinterpret{\xi}}(3) \otimes \interpret{v}[z:=\pi^1(A)]
    \ \oplus \\
  \phantom{A}
    \lift_{\typeinterpret{\xi}}(3 \otimes
    \flatten_{\typeinterpret{\sigma} \times \typeinterpret{\tau}}(
    \interpret{u})) \otimes
     \interpret{v}[z:=\pi^1(A)]
    \ \oplus \\
  \phantom{A}
    \lift_{\typeinterpret{\xi}}(\flatten_{\typeinterpret{\rho} \times
    \typeinterpret{\pi}}(su)) \otimes
     \interpret{v}[z:=\pi^1(A)]
    \ \oplus \hfill (L1) \\
  \phantom{A}
    \lift_{\typeinterpret{\xi}}(\flatten_{\typeinterpret{\rho} \times
    \typeinterpret{\pi}}(tu)) \otimes
     \interpret{v}[z:=\pi^1(A)]
    \ \oplus \hfill (L2) \\
  \phantom{A}
    \lift_{\typeinterpret{\xi}}(\flatten_{\typeinterpret{\sigma} \times
    \typeinterpret{\tau}}(\interpret{u}) \otimes
    \flatten_{\typeinterpret{\rho} \times \typeinterpret{\pi}}(su))
    \otimes
     \interpret{v}[z:=\pi^1(A)]
    \ \oplus \hfill (L3) \\
  \phantom{A}
    \lift_{\typeinterpret{\xi}}(\flatten_{\typeinterpret{\sigma} \times
    \typeinterpret{\tau}}(\interpret{u}) \otimes
    \flatten_{\typeinterpret{\rho} \times \typeinterpret{\pi}}(tu))
    \otimes
     \interpret{v}[z:=\pi^1(A)]
    \ \oplus \hfill (L4) \\
  \phantom{A}
    \lift_{\typeinterpret{\xi}}(3) \otimes \interpret{w}[a:=\pi^2(A)]
    \ \oplus \\
  \phantom{A}
    \lift_{\typeinterpret{\xi}}(3 \otimes
    \flatten_{\typeinterpret{\sigma} \times \typeinterpret{\tau}}(
    \interpret{u})) \otimes
    \interpret{w}[a:=\pi^2(A)]
    \ \oplus \\
  \phantom{A}
    \lift_{\typeinterpret{\xi}}(\flatten_{\typeinterpret{\rho} \times
    \typeinterpret{\pi}}(su)) \otimes
    \interpret{w}[a:=\pi^2(A)]
    \ \oplus \hfill (L5) \\
  \phantom{A}
    \lift_{\typeinterpret{\xi}}(\flatten_{\typeinterpret{\rho} \times
    \typeinterpret{\pi}}(tu)) \otimes
    \interpret{w}[a:=\pi^2(A)]
    \ \oplus \hfill (L6) \\
  \phantom{A}
    \lift_{\typeinterpret{\xi}}(\flatten_{\typeinterpret{\sigma} \times
    \typeinterpret{\tau}}(\interpret{u}) \otimes
    \flatten_{\typeinterpret{\rho} \times \typeinterpret{\pi}}(su))
    \otimes
    \interpret{w}[a:=\pi^2(A)]
    \ \oplus \hfill (L7) \\
  \phantom{A}
    \lift_{\typeinterpret{\xi}}(\flatten_{\typeinterpret{\sigma} \times
    \typeinterpret{\tau}}(\interpret{u}) \otimes
    \flatten_{\typeinterpret{\rho} \times \typeinterpret{\pi}}(tu))
    \otimes
    \interpret{w}[a:=\pi^2(A)] \hfill (L8) \\
  \succ \\
  vsu \oplus wsu \oplus
    vtu \oplus wtu\ \oplus \\
  \phantom{A}
    \lift_{\typeinterpret{\xi}}(\flatten_{\typeinterpret{\rho} \times
    \typeinterpret{\pi}}(su)) \otimes vsu\ \oplus
    \hfill (R1) \\
  \phantom{A}
    \lift_{\typeinterpret{\xi}}(\flatten_{\typeinterpret{\rho}
    \times \typeinterpret{\pi}}(su)) \otimes wsu\ \oplus
    \hfill (R5) \\
  \end{array}\]\[\begin{array}{l}
  \phantom{A}
    \lift_{\typeinterpret{\xi}}(\flatten_{\typeinterpret{\rho} \times
    \typeinterpret{\pi}}(tu)) \otimes vtu\ \oplus \hfill (R2) \\
  \phantom{A}
    \lift_{\typeinterpret{\xi}}(\flatten_{\typeinterpret{\rho}
    \times \typeinterpret{\pi}}(tu)) \otimes wtu\ \oplus
    \hfill (R6) \\
  \phantom{A}
    \lift_{\typeinterpret{\xi}}(\flatten_{\typeinterpret{\sigma}
    \times \typeinterpret{\tau}}(\interpret{u})) \otimes vsu\ \oplus \\
  \phantom{A}
    \lift_{\typeinterpret{\xi}}(\flatten_{\typeinterpret{\sigma}
    \times \typeinterpret{\tau}}(\interpret{u})) \otimes wsu\ \oplus \\
  \phantom{A}
    \lift_{\typeinterpret{\xi}}(\flatten_{\typeinterpret{\sigma}
    \times \typeinterpret{\tau}}(\interpret{u}) \otimes
    \flatten_{\typeinterpret{\rho} \times
    \typeinterpret{\pi}}(su)) \otimes vsu\ \oplus \hfill (R3) \\
  \phantom{A}
    \lift_{\typeinterpret{\xi}}(\flatten_{\typeinterpret{\sigma}
    \times \typeinterpret{\tau}}(\interpret{u}) \otimes
    \flatten_{\typeinterpret{\rho}
    \times \typeinterpret{\pi}}(su)) \otimes wsu\ \oplus \hfill (R7) \\
  \phantom{A}
    \lift_{\typeinterpret{\xi}}(\flatten_{\typeinterpret{\sigma}
    \times \typeinterpret{\tau}}(\interpret{u})) \otimes vtu\ \oplus \\
  \phantom{A}
    \lift_{\typeinterpret{\xi}}(\flatten_{\typeinterpret{\sigma}
    \times \typeinterpret{\tau}}(\interpret{u})) \otimes wtu\ \oplus \\
  \phantom{A}
    \lift_{\typeinterpret{\xi}}(\flatten_{\typeinterpret{\sigma}
    \times \typeinterpret{\tau}}(\interpret{u}) \otimes
    \flatten_{\typeinterpret{\rho} \times
    \typeinterpret{\pi}}(tu)) \otimes vtu\ \oplus \hfill (R4) \\
  \phantom{A}
    \lift_{\typeinterpret{\xi}}(\flatten_{\typeinterpret{\sigma}
    \times \typeinterpret{\tau}}(\interpret{u}) \otimes
    \flatten_{\typeinterpret{\rho}
    \times \typeinterpret{\pi}}(tu)) \otimes wtu \hfill (R8)
  \end{array}
  \]
  This merely leaves the following proof obligation:
  \[
  \begin{array}{l}
  \lift_{\typeinterpret{\xi}}(2)\ \oplus \\
  \phantom{A}
     \lift_{\typeinterpret{\xi}}(2 \otimes
      \flatten_{\typeinterpret{\sigma} \times \typeinterpret{\tau}}(
      \interpret{u}))\ \oplus \\
  \phantom{A}
    \lift_{\typeinterpret{\xi}}(3) \otimes \interpret{v}[z:=\pi^1(A)]
    \ \oplus \\
  \phantom{A}
    \lift_{\typeinterpret{\xi}}(3 \otimes
    \flatten_{\typeinterpret{\sigma} \times \typeinterpret{\tau}}(
    \interpret{u})) \otimes
     \interpret{v}[z:=\pi^1(A)]
    \ \oplus \\
  \phantom{A}
    \lift_{\typeinterpret{\xi}}(3) \otimes \interpret{w}[a:=\pi^2(A)]
    \ \oplus \\
  \phantom{A}
    \lift_{\typeinterpret{\xi}}(3 \otimes
    \flatten_{\typeinterpret{\sigma} \times \typeinterpret{\tau}}(
    \interpret{u})) \otimes
    \interpret{w}[a:=\pi^2(A)]
  \succ \\
  vsu \oplus wsu \oplus vtu \oplus wtu\ \oplus \\
  \phantom{A}
    \lift_{\typeinterpret{\xi}}(\flatten_{\typeinterpret{\sigma}
    \times \typeinterpret{\tau}}(\interpret{u})) \otimes vsu\ \oplus \\
  \phantom{A}
    \lift_{\typeinterpret{\xi}}(\flatten_{\typeinterpret{\sigma}
    \times \typeinterpret{\tau}}(\interpret{u})) \otimes wsu\ \oplus \\
  \phantom{A}
    \lift_{\typeinterpret{\xi}}(\flatten_{\typeinterpret{\sigma}
    \times \typeinterpret{\tau}}(\interpret{u})) \otimes vtu\ \oplus \\
  \phantom{A}
    \lift_{\typeinterpret{\xi}}(\flatten_{\typeinterpret{\sigma}
    \times \typeinterpret{\tau}}(\interpret{u})) \otimes wtu \\
  \end{array}
  \]
  Since $\lift_{\typeinterpret{\xi}}(3) \otimes s \approx
  s \oplus s \oplus s$, we can eliminate all remaining terms (for
  example: $\lift_{\typeinterpret{\xi}}(3) \otimes
  \interpret{v}[z:=\pi^1(A)] \approx
  \interpret{v}[z:=\pi^1(A)] \oplus \interpret{v}[z:=\pi^1(A)] \oplus
  \interpret{v}[z:=\pi^1(A)] \succeq vsu \oplus vtu$); thus, the
  inequality holds.

\item
  $\interpret{\mathtt{let}_{\varphi,\rho}(
  \mathtt{case}_{\sigma,\tau,\exists\varphi}(
  u,\abs{x:\sigma}{s},\abs{y:\tau}{t}),v)} \succ \\
  \interpret{\mathtt{case}_{\sigma,\tau,\rho}(u,
  \abs{x:\sigma}{\mathtt{let}_{\varphi,\rho}(s,v)},
  \abs{y:\tau}{\mathtt{let}_{\varphi,\rho}(t,v)})}$ \\

  In the following, let us denote $v_N :=
  \interpret{v} * \nat \cdot \lift_{\typeinterpret{\varphi}\nat}(0)$
  and $u_f := \flatten_{\typeinterpret{\sigma} \times \typeinterpret{
  \tau}}(\interpret{u})$.  With these abbreviations, we have the
  following on the left-hand side:
  \[
  \begin{array}{l}
  \interpret{\mathtt{let}_{\varphi,\rho}(
    \mathtt{case}_{\sigma,\tau,\exists\varphi}(u,
    \abs{x:\sigma}{s},\abs{y:\tau}{t}),v)} \approx \\
  \lift_{\interpret{\rho}}(1)\ \oplus \\
  \phantom{A}
  \lift_{\interpret{\rho}}(2) \otimes
    \interpret{\mathtt{case}_{\sigma,\tau,\exists\varphi}(u,
    \abs{x}{s},\abs{y}{t})} * \typeinterpret{\rho} \cdot
    (\tabs{\alpha}{\abs{z}{\interpret{v} * \alpha \cdot z}})\ \oplus \\
  \phantom{A}
  \lift_{\interpret{\rho}}(\flatten_{\Sigma\gamma.\typeinterpret{
    \varphi}\gamma}(\interpret{\mathtt{case}_{\sigma,\tau,
    \exists\varphi}(u,\abs{x}{s},\abs{y}{t})}) \oplus 1) \otimes
    v_N \approx \\
  \lift_{\interpret{\rho}}(1) \oplus v_N \oplus
  \lift_{\interpret{\rho}}(2)\ \otimes \\
  \phantom{A}
  (\ \lift_{\Sigma\gamma.\typeinterpret{\varphi}\gamma}(2) \oplus
    \lift_{\Sigma\gamma.\typeinterpret{\varphi}\gamma}(3 \otimes u_f)\
    \oplus \\
  \phantom{AB}
     \lift_{\Sigma\gamma.\typeinterpret{\varphi}\gamma}(u_n \oplus 1)
     \otimes
     (\interpret{s}[x:=\pi^1(\interpret{u})] \oplus
      \interpret{t}[y:=\pi^2(\interpret{u})])
  \\
  \phantom{A}) * \typeinterpret{\rho} \cdot
    (\tabs{\alpha}{\abs{z}{\interpret{v} * \alpha \cdot z}})\ \oplus \\
  \phantom{A}
  \lift_{\typeinterpret{\rho}}(\flatten_{\Sigma\gamma.\typeinterpret{
    \varphi}\gamma}( \\
  \phantom{AB}
  \lift_{\Sigma\gamma.\typeinterpret{\varphi}\gamma}(2) \oplus
    \lift_{\Sigma\gamma.\typeinterpret{\varphi}\gamma}(3 \otimes
    u_f)\ \oplus \\
  \phantom{AB}
     \lift_{\Sigma\gamma.\typeinterpret{\varphi}\gamma}(u_f \oplus 1)
     \otimes
     (\interpret{s}[x:=\pi^1(\interpret{u})] \oplus
      \interpret{t}[y:=\pi^2(\interpret{u})])
  \\
  \phantom{A})) \otimes v_N \approx \\
  \lift_{\interpret{\rho}}(1) \oplus v_N
    \oplus \lift_{\interpret{\rho}}(2)\ \otimes \\
  \phantom{A}
  (\ \lift_{\typeinterpret{\rho}}(2) \oplus
    \lift_{\typeinterpret{\rho}}(3 \otimes u_f) \oplus
     \lift_{\typeinterpret{\rho}}(u_f \oplus 1)\ \otimes \\
  \phantom{ABC}(\ \interpret{s}[x:=\pi^1(\interpret{u})] *
    \typeinterpret{\rho} \cdot
    (\tabs{\alpha}{\abs{z}{\interpret{v} * \alpha \cdot z}})\ \oplus \\
  \phantom{ABCD} \interpret{t}[y:=\pi^2(\interpret{u})] *
    \typeinterpret{\rho} \cdot
    (\tabs{\alpha}{\abs{z}{\interpret{v} * \alpha \cdot z}})\ ) \\
  \phantom{A})\ \oplus \\
  \end{array}\]\[\begin{array}{l}
  \phantom{A}(\
  \lift_{\typeinterpret{\rho}}(2) \oplus
    \lift_{\typeinterpret{\rho}}(3 \otimes u_f) \oplus
  \lift_{\typeinterpret{\rho}}(u_f \oplus 1)\ \otimes \\
  \phantom{ABC}
    \lift_{\typeinterpret{\rho}}(\flatten_{\Sigma\gamma.\typeinterpret{
    \varphi}\gamma}(\interpret{s}[x:=\pi^1(\interpret{u})] \oplus
      \interpret{t}[y:=\pi^2(\interpret{u})])) \\
  \phantom{A}) \otimes v_N \approx \\
  \lift_{\interpret{\rho}}(1) \oplus v_N \oplus
    \lift_{\interpret{\rho}}(4) \oplus
    \lift_{\typeinterpret{\rho}}(6 \otimes u_f)\ \oplus \\
  \phantom{A}
  \lift_{\interpret{\rho}}(2) \otimes
      (\interpret{s}[x:=\pi^1(\interpret{u})] * \typeinterpret{\rho} \cdot
      (\tabs{\alpha}{\abs{z}{\interpret{v} * \alpha \cdot z}}))\ \oplus \\
  \phantom{A}
  \lift_{\interpret{\rho}}(2) \otimes
      (\interpret{t}[y:=\pi^2(\interpret{u})] * \typeinterpret{\rho} \cdot
      (\tabs{\alpha}{\abs{z}{\interpret{v} * \alpha \cdot z}}))\ \oplus \\
  \phantom{A}
  \lift_{\interpret{\rho}}(2 \otimes u_f) \otimes
    (\interpret{s}[x:=\pi^1(\interpret{u})] *
    \typeinterpret{\rho} \cdot
    (\tabs{\alpha}{\abs{z}{\interpret{v} * \alpha \cdot z}}))\ \oplus \\
  \phantom{A}
  \lift_{\interpret{\rho}}(2 \otimes u_f) \otimes
    (\interpret{t}[y:=\pi^2(\interpret{u})] *
    \typeinterpret{\rho} \cdot
    (\tabs{\alpha}{\abs{z}{\interpret{v} * \alpha \cdot z}}))\ \oplus \\
  \phantom{A}
  \lift_{\typeinterpret{\rho}}(2) \otimes v_N\ \oplus \\
  \phantom{A}
  \lift_{\typeinterpret{\rho}}(3 \otimes u_f) \otimes v_N\ \oplus \\
  \phantom{A}
  \lift_{\typeinterpret{\rho}}(\flatten_{\Sigma\gamma.\typeinterpret{
    \varphi}\gamma}(\interpret{s}[x:=\pi^1(\interpret{u})]))
    \otimes v_N\ \oplus \\
  \phantom{A}
  \lift_{\typeinterpret{\rho}}(\flatten_{\Sigma\gamma.\typeinterpret{
    \varphi}\gamma}(\interpret{t}[y:=\pi^2(\interpret{u})]))
    \otimes v_N\ \oplus \\
  \phantom{A}
  \lift_{\typeinterpret{\rho}}(u_f \otimes
    \flatten_{\Sigma\gamma.\typeinterpret{\varphi}\gamma}(
    \interpret{s}[x:=\pi^1(\interpret{u})]))
    \otimes v_N\ \oplus \\
  \phantom{A} \lift_{\typeinterpret{\rho}}(u_f \otimes
    \flatten_{\Sigma\gamma.\typeinterpret{
    \varphi}\gamma}(\interpret{t}[y:=\pi^2(\interpret{u})])) \otimes
    v_N \approx \\
  \lift_{\interpret{\rho}}(5)\ \oplus \\
  \phantom{A}
  \lift_{\typeinterpret{\rho}}(6 \otimes u_f)\ \oplus \\
  \phantom{A}
  \lift_{\interpret{\rho}}(2) \otimes
    (\xlet{\typeinterpret{\rho}}{\interpret{s}[x:=\pi^1(
    \interpret{u})]}{\expair{\alpha}{z}}{\interpret{v} * \alpha
    \cdot z})\ \oplus \\
  \phantom{A}
  \lift_{\interpret{\rho}}(2) \otimes
    (\xlet{\typeinterpret{\rho}}{\interpret{y}[y:=\pi^2(
    \interpret{u})]}{\expair{\alpha}{z}}{\interpret{v} * \alpha
    \cdot z})\ \oplus \\
  \phantom{A}
  \lift_{\interpret{\rho}}(2 \otimes u_f) \otimes
    (\xlet{\typeinterpret{\rho}}{\interpret{s}[x:=\pi^1(
    \interpret{u})]}{\expair{\alpha}{z}}{\interpret{v} * \alpha
    \cdot z})\ \oplus \\
  \phantom{A}
  \lift_{\interpret{\rho}}(2 \otimes u_f) \otimes
    (\xlet{\typeinterpret{\rho}}{\interpret{y}[y:=\pi^2(
    \interpret{u})]}{\expair{\alpha}{z}}{\interpret{v} * \alpha
    \cdot z})\ \oplus \\
  \phantom{A}
  \lift_{\typeinterpret{\rho}}(3) \otimes v_N\ \oplus \\
  \phantom{A}
  \lift_{\typeinterpret{\rho}}(3 \otimes u_f) \otimes v_N\ \oplus \\
  \phantom{A}
  \lift_{\typeinterpret{\rho}}(\flatten_{\Sigma\gamma.\typeinterpret{
    \varphi}\gamma}(\interpret{s}[x:=\pi^1(\interpret{u})]))
    \otimes v_N\ \oplus \\
  \phantom{A}
  \lift_{\typeinterpret{\rho}}(\flatten_{\Sigma\gamma.\typeinterpret{
    \varphi}\gamma}(\interpret{t}[y:=\pi^2(\interpret{u})]))
    \otimes v_N\ \oplus \\
  \phantom{A}
  \lift_{\typeinterpret{\rho}}(u_f \otimes
    \flatten_{\Sigma\gamma.\typeinterpret{\varphi}\gamma}(
    \interpret{s}[x:=\pi^1(\interpret{u})]))
    \otimes v_N\ \oplus \\
  \phantom{A} \lift_{\typeinterpret{\rho}}(u_f \otimes
    \flatten_{\Sigma\gamma.\typeinterpret{
    \varphi}\gamma}(\interpret{t}[y:=\pi^2(\interpret{u})])) \otimes
    v_N \approx \\
  \end{array}
  \]

  On the right-hand side, we have:
  \[
  \begin{array}{l}
  \interpret{\mathtt{case}_{\sigma,\tau,\rho}(u,
  \abs{x:\sigma}{\mathtt{let}_{\varphi,\rho}(s,v)},
  \abs{y:\tau}{\mathtt{let}_{\varphi,\rho}(t,v)})} \approx \\
  \lift_{\typeinterpret{\rho}}(2) \oplus
  \lift_{\typeinterpret{\rho}}(3 \otimes u_f) \oplus
  \lift_{\typeinterpret{\rho}}(u_f \oplus 1)\ \otimes \\
  \phantom{AB}
    (\interpret{\mathtt{let}_{\varphi,\rho}(s,v)}[x:=\pi^1(
    \interpret{u})] \oplus
    \interpret{\mathtt{let}_{\varphi,\rho}(t,v)}[y:=\pi^2(
    \interpret{u})]) \approx \\
  \lift_{\typeinterpret{\rho}}(2) \oplus
  \lift_{\typeinterpret{\rho}}(3 \otimes u_f)\ \oplus \\
  \phantom{A}
  \interpret{\mathtt{let}_{\varphi,\rho}(s,v)}[x:=\pi^1(
    \interpret{u})]\ \oplus \\
  \phantom{A}
    \interpret{\mathtt{let}_{\varphi,\rho}(t,v)}[y:=\pi^2(
    \interpret{u})]\ \oplus \\
  \phantom{A}
  \lift_{\typeinterpret{\rho}}(u_f) \otimes
    \interpret{\mathtt{let}_{\varphi,\rho}(s,v)}[x:=\pi^1(
    \interpret{u})]\ \oplus \\
  \phantom{A}
  \lift_{\typeinterpret{\rho}}(u_f) \otimes
    \interpret{\mathtt{let}_{\varphi,\rho}(t,v)}[y:=\pi^2(
    \interpret{u})] \approx \\
  \lift_{\typeinterpret{\rho}}(2) \oplus
  \lift_{\typeinterpret{\rho}}(3 \otimes u_f)\ \oplus \\
  \phantom{A}
  (\ \lift_{\typeinterpret{\rho}}(1) \oplus
    \lift_{\typeinterpret{\rho}}(2) \otimes
    (\xlet{\typeinterpret{\rho}}{\interpret{s}}{\expair{\alpha}{z}}{
      \interpret{v} * \alpha \cdot z})\ \oplus \\
  \phantom{AB}
    \lift_{\typeinterpret{\rho}}(\flatten(\interpret{s}) \oplus 1)
    \otimes v_N\ )[x:=\pi^1(\interpret{u})]\ \oplus \\
  \phantom{A}
  (\ \lift_{\typeinterpret{\rho}}(1) \oplus
    \lift_{\typeinterpret{\rho}}(2) \otimes
    (\xlet{\typeinterpret{\rho}}{\interpret{t}}{\expair{\alpha}{z}}{
      \interpret{v} * \alpha \cdot z})\ \oplus \\
  \phantom{AB}
    \lift_{\typeinterpret{\rho}}(\flatten(\interpret{t}) \oplus 1)
    \otimes v_N\ )[y:=\pi^2(\interpret{u})]\ \oplus \\
  \phantom{A}
  \lift_{\typeinterpret{\rho}}(u_f) \otimes
  (\ \lift_{\typeinterpret{\rho}}(1) \oplus
    \lift_{\typeinterpret{\rho}}(2) \otimes
    (\xlet{\typeinterpret{\rho}}{\interpret{s}}{\expair{\alpha}{z}}{
      \interpret{v} * \alpha \cdot z})\ \oplus \\
  \phantom{ABCDEFGHIJ}
    \lift_{\typeinterpret{\rho}}(\flatten(\interpret{s}) \oplus 1)
    \otimes v_N\ )[x:=\pi^1(\interpret{u})]\ \oplus \\
  \phantom{A}
  \lift_{\typeinterpret{\rho}}(u_f) \otimes
  (\ \lift_{\typeinterpret{\rho}}(1) \oplus
    \lift_{\typeinterpret{\rho}}(2) \otimes
    (\xlet{\typeinterpret{\rho}}{\interpret{t}}{\expair{\alpha}{z}}{
      \interpret{v} * \alpha \cdot z})\ \oplus \\
  \phantom{ABCDEFGHIJ}
    \lift_{\typeinterpret{\rho}}(\flatten(\interpret{t}) \oplus 1)
    \otimes v_N\ )[y:=\pi^2(\interpret{u})] \approx \\
  \end{array}\]\[\begin{array}{l}
  \lift_{\typeinterpret{\rho}}(2) \oplus
  \lift_{\typeinterpret{\rho}}(3 \otimes u_f)\ \oplus \\
  \phantom{A}
  \lift_{\typeinterpret{\rho}}(1) \oplus
  \lift_{\typeinterpret{\rho}}(2) \otimes
  (\xlet{\typeinterpret{\rho}}{\interpret{s}[x:=\pi^1(\interpret{u})}{
    \expair{\alpha}{z}}{\interpret{v} * \alpha \cdot z})\ \oplus \\
  \phantom{A}
  v_N \oplus \lift_{\typeinterpret{\rho}}(\flatten(\interpret{s}[x:=
    \pi^1(\interpret{u}))) \otimes v_N\ \oplus \\
  \phantom{A}
  \lift_{\typeinterpret{\rho}}(1) \oplus
    \lift_{\typeinterpret{\rho}}(2) \otimes
    (\xlet{\typeinterpret{\rho}}{\interpret{t}[y:=\pi^2(\interpret{u})]}{
    \expair{\alpha}{z}}{\interpret{v} * \alpha \cdot z})\ \oplus \\
  \phantom{A}
  v_N \oplus \lift_{\typeinterpret{\rho}}(\flatten(\interpret{t}[y:=
    \pi^2(\interpret{u})])) \otimes v_N\ \oplus \\
  \phantom{A}
  \lift_{\typeinterpret{\rho}}(u_f) \oplus
    \lift_{\typeinterpret{\rho}}(2 \otimes u_f) \otimes
    (\xlet{\typeinterpret{\rho}}{\interpret{s}[x:=\pi^1(
    \interpret{u})]}{\expair{\alpha}{z}}{
      \interpret{v} * \alpha \cdot z})\ \oplus \\
  \phantom{A}
  \lift_{\typeinterpret{\rho}}(u_f) \otimes v_N \oplus
    \lift_{\typeinterpret{\rho}}(u_f \otimes
    \flatten(\interpret{s}[x:=\pi^1(\interpret{u})])) \otimes v_N\
    \oplus \\
  \phantom{A}
  \lift_{\typeinterpret{\rho}}(u_f) \oplus
    \lift_{\typeinterpret{\rho}}(2 \otimes u_f) \otimes
    (\xlet{\typeinterpret{\rho}}{\interpret{t}[y:=\pi^2(
    \interpret{u})]}{\expair{\alpha}{z}}{\interpret{v} * \alpha
    \cdot z})\ \oplus \\
  \phantom{A}
  \lift_{\typeinterpret{\rho}}(u_f) \otimes v_N \oplus
    \lift_{\typeinterpret{\rho}}(u_f \otimes\flatten(\interpret{t}[y:=
    \pi^2(\interpret{u})])) \otimes v_N \\
  \end{array}
  \]
  The last step follows because $x$ occurs only in $s$, and $y$ occurs
  only in $t$.  This term can now be reordered to:
  \[
  \begin{array}{l}
  \lift_{\typeinterpret{\rho}}(4)\ \oplus \\
  \phantom{A}
  \lift_{\typeinterpret{\rho}}(5 \otimes u_f)\ \oplus \\
  \phantom{A}
  \lift_{\typeinterpret{\rho}}(2) \otimes
  (\xlet{\typeinterpret{\rho}}{\interpret{s}[x:=\pi^1(\interpret{u})}{
    \expair{\alpha}{z}}{\interpret{v} * \alpha \cdot z})\ \oplus \\
  \phantom{A}
  \lift_{\typeinterpret{\rho}}(2) \otimes
    (\xlet{\typeinterpret{\rho}}{\interpret{t}[y:=\pi^2(\interpret{u})]}{
    \expair{\alpha}{z}}{\interpret{v} * \alpha \cdot z})\ \oplus \\
  \phantom{A}
    \lift_{\typeinterpret{\rho}}(2 \otimes u_f) \otimes
    (\xlet{\typeinterpret{\rho}}{\interpret{s}[x:=\pi^1(
    \interpret{u})]}{\expair{\alpha}{z}}{
    \interpret{v} * \alpha \cdot z})\ \oplus \\
  \phantom{A}
    \lift_{\typeinterpret{\rho}}(2 \otimes u_f) \otimes
    (\xlet{\typeinterpret{\rho}}{\interpret{t}[y:=\pi^2(
    \interpret{u})]}{\expair{\alpha}{z}}{
    \interpret{v} * \alpha \cdot z})\ \oplus \\
  \phantom{A}
  \lift_{\typeinterpret{\rho}}(2) \otimes v_N\ \oplus \\
  \phantom{A}
  \lift_{\typeinterpret{\rho}}(2 \otimes u_f) \otimes v_N\ \oplus \\
  \phantom{A}
  \lift_{\typeinterpret{\rho}}(\flatten(\interpret{s}[x:=
    \pi^1(\interpret{u}))) \otimes v_N\ \oplus \\
  \phantom{A}
  \lift_{\typeinterpret{\rho}}(\flatten(\interpret{t}[y:=
    \pi^2(\interpret{u})])) \otimes v_N\ \oplus \\
  \phantom{A}
  \lift_{\typeinterpret{\rho}}(u_f \otimes
    \flatten(\interpret{s}[x:=\pi^1(\interpret{u})])) \otimes v_N\
    \oplus \\
  \phantom{A}
  \lift_{\typeinterpret{\rho}}(u_f \otimes\flatten(\interpret{t}[y:=
    \pi^2(\interpret{u})])) \otimes v_N \\
  \end{array}
  \]
  We conclude once more by absolute positiveness.
\end{itemize}

\end{document}